\documentclass[a4paper,fleqn,leqno]{book}

\usepackage[english]{babel}
\usepackage[colorlinks]{hyperref}
\usepackage{amsfonts}
\usepackage{amsmath}
\usepackage{amssymb}
\usepackage{amsthm}
\usepackage{booktabs}
\usepackage{graphicx}
\usepackage{listings}
\usepackage[slantedGreek]{mathpazo} 
\usepackage{makeidx}
\usepackage{mathrsfs}
\usepackage{microtype}
\usepackage{rgalg}
\usepackage{slemph}
\usepackage[center]{subfigure}
\usepackage{tikz}
\usepackage{xcolor}
\usepackage{xspace}
\usepackage{hyperref}  
\usepgflibrary{arrows}
\usepgflibrary{shapes.geometric}
\definecolor{darkblue}{rgb}{0,0,0.4}
\definecolor{verylightgray}{rgb}{0.95,0.95,0.95}
\hypersetup{colorlinks,linkcolor=darkblue,citecolor=darkblue,urlcolor=darkblue}
\hypersetup{
  pdfauthor={Radu Grigore},
  pdftitle={The Design and Algorithms of a Verification Condition Generator}}

\tikzstyle{arr}=[->,>=stealth']
\tikzstyle{predcirc}=[
  circle,
  very thick,
  fill=green!10,
  draw=green,
  minimum size=14pt,
  inner sep=0pt]
\tikzstyle{predrect}=[predcirc,rectangle,inner sep=2pt]

\newcommand\fgnodeR{2pt}  
\tikzstyle{fgdraw}=[
  minimum size=2*\fgnodeR,inner sep=0pt,outer sep=1pt,
  draw,thick]
\tikzstyle{fgfill}=[fill=black]
\ifx\fgnode\undefined\else\errmessage{\string\fgnode already defined!}\fi
\def\fgnode#1#2(#3) at (#4){
    \begin{scope}[shift={(#4)},shift only]
      \clip (-\fgnodeR,-#1*\fgnodeR) rectangle (\fgnodeR,#2*\fgnodeR);
      \node[fgdraw,circle,fgfill] {};
    \end{scope}
    \node[fgdraw,circle] (#3) at (#4)}
\newcommand{\oonode}{\fgnode00} 
\newcommand{\ronode}{\fgnode01} 
\newcommand{\wonode}{\fgnode10} 
\newcommand{\rwnode}{\fgnode11} 
\newcommand{\gnode}{\node[fgdraw,circle,fill=gray]}  
\newcommand{\cnode}{\node[fgdraw,fgfill,rectangle]} 
\newcommand{\enode}{\oonode} 
\newcommand{\fnode}{\rwnode} 
\makeindex
\lstdefinestyle{boogie}{
  morekeywords={procedure,returns,assume,assert,havoc,goto,return,
    int,bool,type,while,if,true,false,function,bool,returns,axiom,
    forall,var}
}
\lstdefinestyle{jml}{
  language=java,
  morekeywords={requires,ensures,old,invariant,forall,exists,axiom,also,
    result,pure,assert,modifies},
  deletekeywords={label}
}
\lstdefinestyle{smt}{
  morekeywords={ite,true,false,BG_PUSH,IFF,FORALL,EQ,NEQ,NOT,TRUE,FALSE,
    IMPLIES}
}

\newcommand{\boogieCode}{\lstinline[style=boogie,basicstyle=\normalsize]}
\newcommand{\jmlCode}{\lstinline[style=jml,basicstyle=\normalsize]}
\newcommand{\smtCode}{\lstinline[style=smt,basicstyle=\normalsize]}
\newcommand{\deflang}[1]{\lstnewenvironment{#1}[1][]{\lstset{style=#1,##1}}{}}
\deflang{jml}
\deflang{boogie}
\deflang{smt}
\makeatletter 
  \g@addto@macro\verbatim{\microtypesetup{activate=false}}
\makeatother

\def\fb#1{{\bf #1}}

\newcommand{\R}{\mathbb{R}}

\newcommand{\Z}{\mathbb{Z}}

\newcommand{\csharp}{C$^\sharp$\xspace}
\newcommand{\dts}{\mathinner{\ldotp\ldotp}}
\newcommand{\edge}{\mathrel-\joinrel\joinrel\mathrel-}
\newcommand{\escjava}{ESC\slash Java\xspace}

\newcommand{\fls}{\bot}
\newcommand{\framac}{\hbox{Frama-C}}
\newcommand{\fx}{F\kern-0.1667em\lower.5ex\hbox{X}\kern-.125em 7\xspace}
\newcommand{\hoare}[3]{\{#1\}\>\text{#2}\>\{#3\}}

\newcommand{\indentline}[1]{\\\leftline{\indent#1}}
\newcommand{\limp}{\Rightarrow}
\newcommand{\macro}[1]{\texttt{\char`\\#1}}

\newcommand{\shell}[1]{\indentline{\texttt{#1}} }

\newcommand{\specsharp}{Spec$^\sharp$\xspace}
\newcommand{\spp}[2]{\mathit{sp}\>{#1}\>{#2}}
\newcommand{\sqatop}[2]{\genfrac{[}{]}{0pt}{}{#1}{#2}}

\newcommand{\tru}{\top}

\newcommand{\wpp}[2]{\mathit{wp}\>{#1}\>{#2}}

\newcommand{\startgrammar}{
  \begingroup
  \def\is{$\>\to$&}
  \def\|{$\mid$}
  \def\b##1{\textbf{##1}}
  \def\i##1{\textsl{##1}}
  \def\?{$^?$}
  \def\*{$^\ast$}
  \begin{figure}
  \centering
  \footnotesize
  \begin{tabular}{rl}
}
\newcommand{\stopgrammar}[2]{
  \end{tabular}
  \caption{#1}\label{#2}
  \end{figure}
  \endgroup
}

\newcommand{\bc}{\begin{figure}\centering\begin{tabular}{c}} 
\newcommand{\ec}[2]{\end{tabular}\caption{#1}\label{#2}\end{figure}} 

\newcommand{\chquote}[2]{\hfill\hbox{\vbox{%
  \hsize=.6\hsize\leftskip=0pt plus1fill%
  \noindent\textsf{\textsl{``#1''}\\\smallskip%
  \rightline{--- #2}} }}\bigskip}

\newtheoremstyle{slanted}{}{}{\slshape}{}{\bf}{.}{.5em}{}
\theoremstyle{slanted}
\newtheorem{problem}{Problem}
\newtheorem{conjecture}{Conjecture}
\newtheorem{lemma}{Lemma}
\newtheorem{proposition}{Proposition}
\newtheorem{theorem}{Theorem}
\newtheorem{corollary}{Corollary}
\theoremstyle{definition}
\newtheorem{definition}{Definition}
\newtheorem{example}{Example}
\theoremstyle{remark}
\newtheorem{remark}{Remark}

\begin{document}
\frontmatter
\baselineskip=15pt 
\parskip=0pt plus 1.25pt  
\pagestyle{plain}
\begin{titlepage}
\begin{center}
\vspace*{\stretch{1}}
{\huge\bf The Design and Algorithms of \\[.5ex] 
a Verification Condition Generator}
\\[1ex]
{\large FreeBoogie} 
\\[3ex]
{\Large Radu Grigore}

\bigskip\bigskip
The thesis is submitted to University College Dublin for
the degree of PhD in the College of Engineering, Mathematical
\& Physical Sciences.

\bigskip
March 2010

\bigskip
School of Computer Science and Informatics \\
\emph{Head of School}: Prof.~Joe Carthy\\
\emph{Supervisor}: Assoc.~Prof.~Joseph Roland Kiniry\\
\emph{Second supervisor}: Prof.~Simon Dobson
\vspace*{\stretch{1}}
\end{center}
\end{titlepage}

\setcounter{tocdepth}{1}\tableofcontents

\newpage\section*{Acknowledgments}

In October 2005, Joe Kiniry picked me up from Dublin Airport. The same
week, he showed me \escjava and asked me to start fixing its bugs. To do
so, I had to learn about Hoare triples, guarded commands, and a few other
things: He taught me by throwing at me the right problems. He also provided
a good environment for research, by building from scratch in UCD a group
focused on applied formal methods. Joe helped me as a friend when I had
difficulties in my personal life.

I tend to spend most of my time learning, rather than doing.  However, in
the summer of~2007 the reverse was true. The cause was the epidemic
enthusiasm of Micha{\l} Moskal, who visited our group.  For the rest of the
four years that I spent in Dublin, Mikol\'a\v{s} Janota was the main target
of my technical ramblings, since we lived and worked together. 

Rustan Leino provided much feedback on a draft of this
dissertation.  He observed typesetting problems and suggested how
to fix them; he alerted me to subtle errors; he suggested
numerous improvements; he made important high-level observations.
Henry McLoughlin, Joe, and Mikol\'a\v{s} also provided
substantial feedback.
 
Fintan Fairmichael, Joe, Julien Charles, Micha{\l}, and
Mikol\'a\v{s} are coauthors of the papers on which this
dissertation is based.  During my trips to conferences and to
graduate schools I learned much by discussing and by listening to
many, many people. Some of them are Christian Haack, Cl\'ement
Hurlin, Cormac Flanagan, Erik Poll, Jan Smans, and Matt
Parkinson.  I particularly enjoyed attempting to solve Rustan's
puzzles.

I remember Dublin as the friendliest city I have ever been to. I remember
the great Chinese food prepared by Juan `Erica' Ye. I remember the great
lectures given by Henry.  I remember playing pool with Javi G\'omez during
our internship with Google.  I remember Julien being over-excited by things
I would hardly notice, such as weird music.  I remember a few colleagues,
like Lorcan Coyle, Fintan, Rosemary Monahan, and Ross Shannon, with whom I
wish I had communicated more and whom I hope to meet again. I remember many
former UCD undergraduates, such as Eugene Kenny, with whom I also hope to
meet~again.

It was a pleasure to spend these years doing nothing but research. However,
I hope in the future I will strike a better balance between family and
work.  I hope I will be a better husband to Claudia, a better father to
Mircea, and a better son to Mariana and Corneliu.  All of them were very
patient and supportive.  This dissertation is for them, and especially for
my son Mircea who, I hope, will read it some day.

\eject
This page intentionally contains only this sentence.
\newpage\section*{Summary}

This dissertation belongs to the broad field of formal methods, which is,
roughly, about using mathematics to improve the quality of software.
Theoreticians help by teaching future programmers how to understand
programs and how to construct them. Once people use mathematical techniques
to do something, it is usually a matter of time until computers take over.
Researchers in applied formal methods try to produce tools that that help
programmers to write high-quality new code and also to find problems in old
code. \specsharp is one attempt to produce such a tool. Looking at its
backend, the Boogie tool, I noticed questions that begged to be answered.

At some point the program is brought into passive form, meaning
that assignments are replaced by equivalent statements. It was
clear that sometimes the passive form was bigger than necessary.
Would it be possible to always compute the best passive form?  To
define precisely what `best' means one first needs a precise
definition of what constitutes a passive form. A part of the
dissertation gives such a precise definition and explores what
follows from it.  Briefly, there are multiple natural meanings
for `best', for some of them it is easy to compute the `best'
passive form, for others it is hard.

Later the program is transformed into a logic formula that is
valid if and only if the program is correct. There are two ways
of doing this, based on the weakest precondition and based on
strongest postcondition, but it was unclear what is the trade-off
involved in choosing one method over the other. 

One use-case for program verifiers is to monitor programmers
writing new code and point out bugs as they are introduced.
Batch processing is bound to be
inefficient. Is there a way to reuse the results of a run when
the program text changes as little as adding a statement or
tweaking a loop condition? The answer is related to Craig
interpolants.

Another question is why the Boogie tool checks only for partial
correctness, and does not perform other analyzes as well. In
particular, checking for semantic reachability of Boogie
statements may reveal a wide range of problems such as
inconsistent specifications, doomed code, bugs in the frontend of
\specsharp, and high-level dead code. Of course, one can manually
insert \textbf{assert false} statements, but this is cumbersome.
Also, the dissertation shows that it is possible to solve the
task much more efficiently than by simply replacing each
statement in turn with \textbf{assert false}.
 
\mainmatter
\chapter{Introduction} \label{ch:intro}

\chquote{The purpose of your paper is not $\ldots$ to describe
the WizWoz system. Your reader does not have a WizWoz. She is
primarily interested in re-usable brain-stuff, not executable
artifacts.}{Simon Peyton-Jones~\cite{peyton-htwgrp}}

\section{Motivation} \label{sec:motivation}

Ideal programs are \emph{correct}, \emph{efficient}, and easy to
\emph{evolve}. Tools can help with all three aspects:
Type-checkers ensure that certain classes of errors do not occur,
profilers identify performance hot-spots, and IDEs (\fb
integrated \fb development \fb environments) refactor programs.
Automation allows humans to focus on the interesting issues.
Knuth~\cite{knuth1974cpa} put it differently: ``Science is
knowledge which we understand so well that we can teach it to a
computer; and if we don't fully understand something, it is an
art to deal with.'' For example, a bit string can represent both
a text and an integer, we understand well how to check that a
program does not mix the two interpretations, and we leave the
task to type-checkers. If, on the other hand, we want to check
that a program computes the transitive closure of a graph we
usually do it by hand.

\bc\begin{jml}
void fw(boolean[][] G)  // $G$ is an adjacency matrix 
  requires square(G); 
  requires (forall i; G[i][i]); 
  ensures (forall i, j; G[i][j] == path(old(G), i, j)); 
{ final int n = G.length; 
  for (int k = 0; k < n; ++k) invariant (forall i, j; G[i][j] == pathK(old(G), i, j, k)); 
    for (int i = 0; i < n; ++i)
      for (int j = 0; j < n; ++j) 
        G[i][j] = G[i][j] || (G[i][k] && G[k][j]); 
}

axiom (forall G, i, j, k; pathK(G, i, j, k) = G[i][j] || (exists q; q<k && G[i][q] && pathK(G, q, j, k))); 
axiom (forall G, i, j; path(G, i, j) = pathK(G, i, j, G.length));
\end{jml}\ec{An example of what program verifiers can ideally do}{lst:fw}

A \emph{program verifier} automatically checks whether code agrees with
specifications. Figure~\ref{lst:fw} shows an example.  Variables $i$,~$j$,
$k$, and~$q$ range over nonnegative integers.  The code is an
implementation of the Roy--Warshall
algorithm~\cite{roy1959,warshall1962tbm}\index{Roy--Warshall algorithm}.
The programmer spent energy to specify \emph{what} the algorithm does
(\textbf{requires}, \textbf{ensures}) and \emph{how} it works
(\textbf{invariant}). In words, $\mathit{path}(G,i,j)$ means that there is
a path~$i \leadsto j$ in the graph~$G$, and $\mathit{pathK}(G,i,j,k)$ means
that there is such a path whose intermediate nodes come only from the
set~$0\dts k-1$. The example illustrates what program verifiers ought to be
able to do.

Note that it is trivial to establish the invariant (because \textit{pathK}
reduces to $G$ when $k=0$) and to infer the postcondition from the
invariant (because \textit{pathK} reduces to \textit{path} when $k=n$).
Proving that the invariant is preserved is trickier than it might seem
mainly because of the in-place update, but could conceivably be done
automatically.  The proof of the invariant is sometimes omitted from
informal explanations of the algorithm, but the information contained by
the definitions of \textit{path} and \textit{pathK} is always communicated.
The long term goal illustrated here is: A program verifier should be
able to automatically check a program even when its annotations contain no
more than what we would say to a reasonably good programmer to explain what
the code does and how it~works.

From an engineering perspective, program verifiers are similar to
compilers. The input is a program written in a high-level
language, and the output is a set of warnings (or errors) that
indicate possible bugs. For a good program verifier, the lack of
warnings should be a better indicator that the program is correct
than any human-made argument. The architecture \emph{often}
consists of a front-end that translates a high-level language
into a much simpler intermediate language, and a back-end that
does the interesting work. The same back-end may be connected to
different front-ends, each supporting some high-level language.
The back-end is itself split into a VC (\fb verification \fb
condition) generator and an SMT (\fb satisfiability \fb modulo
\fb theories) solver. A few alternatives to this architecture are
discussed later (Section~\ref{sec:intro.related}).

The VC generator is a trusted component of program
verifiers. Therefore it is important to study it carefully,
including its less interesting corners. This dissertation shows
that even such corners come to life when analyzed in detail from
the point of view of correctness and efficiency. The insights
gained from such an analysis sometimes lead to simpler and
cleaner implementations and sometimes lead to more efficient
implementations.

\section{History} \label{sec:intro.history}

In 1957, `programming' was not a profession. At least that's what
Dijkstra was told~\cite{ewd340}. ``And, believe it or not, but
under the heading \emph{profession} my marriage act shows the
ridiculous entry \emph{theoretical physicist}!'' That is only one
story showing that in those times programmers were second class
citizens and did not have many rights.  Their predicament was,
however, well-deserved: They did not care about the correctness
of programs, and they did not even grasp what it \emph{means} for
a program to be correct! A few bright people changed the
situation. In 1961 McCarthy~\cite{mccarthy1961} published the
first article concerned with the study of what programs mean. The
article is focused on handling recursive functions without side
effects, which corresponds to the style of programming used
nowadays in pure functional languages like
Haskell~\cite{haskell}. Six years later, in 1967,
Floyd~\cite{floyd1967} showed how programs with side effects and
arbitrary control flow can be handled formally. He acknowledges
that some ideas were inspired by Alan Perlis. He emphasized the
view of programs as flowgraphs (or, more precisely, flowcharts).
His method is usually known under the name ``inductive
assertions.'' The method was popularized by
Knuth~\cite{knuth1968v1}. In 1969, Hoare~\cite{hoare1969}
introduced an axiomatic method of assigning meanings to programs
that appealed more to logicians than to algorithmists.
Although the presentation style is
very different, Hoare states that ``the formal treatment of
program execution presented in this paper is clearly derived from
Floyd.'' Based on Hoare's work, Dijkstra~\cite{dijkstra1975}
introduced in 1975 yet another way of defining the meaning of
programs based on predicate transformers such as the weakest
precondition transformer. He did so in the context of a language
(without \textbf{goto}) called ``guarded commands,'' which
provided the main inspiration for the Boogie language.

It is revealing that \emph{all} the authors mentioned in the
previous paragraph received the Turing Award, although not
necessarily for closely related topics:

\begin{itemize} 
\item[1966] Alan Perlis: ``For his influence in
  the area of advanced programming techniques and compiler
  construction.'' 
\item[1971] John McCarthy: ``For his major
  contributions to the field of artificial intelligence.''
\item[1972] Edsger W.~Dijkstra: ``[For] his approach to
  programming as a high, intellectual challenge; his eloquent
  insistence and practical demonstration that programs should be
  composed correctly, not just debugged into correctness; and his
  illuminating perception of problems at the foundations of program
  design.'' 
\item[1974] Donald E.~Knuth: ``For his major
  contributions to the analysis of algorithms and the design of
  programming languages.'' 
\item[1978] Robert W.~Floyd: ``For
  having a clear influence on methodologies for the creation of
  efficient and reliable software, and for helping to found the
  following important subfields of computer science: the theory of
  parsing, the semantics of programming languages, automatic
  program verification, automatic program synthesis, and analysis
  of algorithms.'' 
\item[1980] C.~A.~R.~Hoare: ``For his
  fundamental contributions to the definition and design of
  programming languages.'' 
\end{itemize}

Most researchers now prefer to define programming languages from an
operational point of view. Such definitions (1)~tend to be more intuitive
for programmers and (2)~correspond directly to how interpreters are
implemented. Plotkin's lecture notes~\cite{plotkin1981sos} constitute the
first coherent and comprehensive account of this approach.  Much later,
in 2004, Plotkin put together a historical account~\cite{plotkin2004} of
how his ideas on structural operational semantics crystallized.  He
points to alternative ways of handling the \textbf{goto} statement. He
cites McCarthy~\cite{mccarthy1961} as inspiring him to simplify existing
work, and credits Smullyan~\cite{smullyan1961} for the
$\frac{\text{up}}{\text{down}}$~rules. He also relates operational
semantics to denotational semantics~\cite{allison1986}.

All these developments are based on even older work. In the nineteenth
century Hilbert advocated a rigorous approach to mathematics: It should be
possible in principle to decompose any mathematical proof\index{proof} into
a finite sequence of formulas $p_1$,~$p_2$, $\ldots$, $p_n$ such that each
formula is either an axiom or is obtained from previous formulas by the
application of a simple transformation rule. (Such a sequence is a proof of
all the formulas it contains.) The set of \emph{axioms}\index{axiom} is
fixed in advance and is called \emph{theory}\index{theory}. The set of
transformation rules is also fixed in advance and is called
\emph{calculus}\index{calculus}.  Without looking at the language used to
express formulas, there is not much more that can be said about this
process. If the language is propositional logic ($\tru$, $\fls$, and
variables, connected by $\lnot$, $\lor$, and~$\land$), then we can
\emph{evaluate}\index{evaluation} a formula to $\tru$ or $\fls$ once a
\emph{valuation}\index{valuation}---assignment of values to variables---is
fixed.  A \emph{model}\index{model} is a valuation that makes all axioms
evaluate to~$\tru$.  A formula is \emph{valid}\index{validity} when it
evaluates to~$\tru$ for all models. A calculus is
\emph{sound}\index{soundness} if it produces only valid formulas starting
from valid axioms; a calculus is \emph{complete}\index{completeness} if it
can produce all the valid formulas. Even if a sound calculus is used,
anything can be derived if we start with an
\emph{inconsistent}\index{consistency} theory, one that has no model. These
observations generalize for other languages like fol, hol (\fb higher \fb
order \fb logic), and lambda calculus. The notion of evaluating formulas is
`operational,' while the calculus feels more like the axiomatic approach to
programming languages. The intimate connection between proofs and programs
is explored in a tutorial style by Wadler~\cite{wadler2000pap}.

\section{Related Work} \label{sec:intro.related}

The previous section gave (intentionally) a very narrow view of
modern research on program verification. It is now time to right
that wrong, partly. Because the field is so vast, we still do not
look at all important subfields. For example, no testing tool
inspired by theory is mentioned.

The sharpest divide is perhaps between tools mainly
informed by practice and tools mainly informed by theory. The
authors of FindBugs~\cite{findbugs}, PMD~\cite{pmd}, and
FxCop~\cite{fxcop} started by looking at patterns that appear in
bad code and then built specific checkers for each of those
patterns.  Hence, those tools are collections of checkers plus
some common functionality that makes it easy to implement and
drive such checkers.  Crystal~\cite{crystal}, Soot~\cite{soot},
and NQuery~\cite{nquery} are stand-alone platforms that make it
easy to implement specific checkers.
Rutar~et~al.~\cite{rutar2004} compare a few static analysis tools
for Java, including FindBugs and PMD, from a practical point of
view.

Tools informed by theory lag behind in impact, but promise to
offer better correctness guarantees in the future. These tools
can be roughly classified by the main techniques used by their
reasoning core.

Model checking~\cite{modcheck} led to successful hardware verifiers like
RuleBase~\cite{rulebase, beer1996}. A model checker verifies whether
certain (liveness and safety) properties hold for a certain state
graph---the `model.' The properties are written in a temporal logic, such
as LTL or CTL; the model is a Kripke structure and is represented usually
in some implicit form that tries to avoid state explosion. SPIN~\cite{spin,
spinbook} and NuSMV~\cite{nusmv, cimatti2002} are generic model checkers
and each has its input language. Hardware model checkers start by
transforming a VHDL or Verilog description into a state graph, while
software model checkers start by transforming a program written in a
language like Java into a state graph.  Software model checkers that are
clearly under active development include the open source Java
Pathfinder~\cite{jpf, visser2003} and CHESS~\cite{chess, chess_tr} (for
Win32 object code and for CIL).  SLAM~\cite{ball2004}, a commercially
successful tool developed by MSR, uses a combination of techniques,
including model checking, to verify several properties of Windows drivers.
Another noteworthy software model checker is BLAST~\cite{blast, beyer2007}
(for C\null). Bogor~\cite{bogor, robby2003} is a framework for developing
software model checkers.

The input of theorem provers~\cite{robinson2001v1,
robinson2001v2} is a logic formula. Usually, when the language is
fol, the theorem prover tries to decide automatically if the
input is valid; usually, when the language is hol the theorem
prover waits patiently to be guided through the proof. The former
are proof finders, while the later are proof checkers.  The
distinction is not clear cut: Sometimes, the steps that an
interactive prover ``checks'' are as complicated as some of the
theorems that are ``proved'' automatically. Still, in practice
the distinction is important: Automatic theorem provers tend to
be fast, while interactive theorem provers tend to be~expressive.

The most widely used hol provers are Coq~\cite{coq, coqart},
Isabelle\slash HOL~\cite{hol, isabelle, holbook}, and
PVS~\cite{pvs, owre1992}. The gist of such provers is that they (should)
rely on a very small trusted core. One way to use such theorem
provers for program verification is to do everything in their
input language. For example, Leroy~\cite{leroy2009} implemented a
compiler for a subset of C in Coq's language, formulated theorems
that capture desired properties of a C compiler, and proved them.
Since Coq comes with a converter from Coq scripts to
OCaml~\cite{remy2000} programs, the compiler is executable. (The
converter fails if non-constructive laws such as the excluded
middle are used.) Another approach is to introduce notation that
makes the Coq\slash HOL\slash Isabelle script look `very much'
like the program that is being run~\cite{marti2006}. Yet another
approach is to use hol as the target language of a VC generator
(see~\cite{bohme2008, barthe2006, berg2001}).

The first approach, that of interactively programming and proving
in the same language, is also used with ACL2~\cite{acl2,
acl2book}, whose input language is not
higher-order. 

Provers that handle only fol do not usually require programmers to interact
directly with them. SMT solvers~\cite{barrett2005}, like Z3~\cite{moura2008z3}
and CVC3~\cite{barrett2007}, are designed for program verification. The modules
of an SMT solver---a SAT solver and decision procedures for various
theories---communicate through equalities. This architecture dates back to
Nelson and Oppen~\cite{nelson1980}. The theories are axiomatizations of things
that occur frequently in programs, like arrays, bit vectors, and integers. A
specialized decision procedure can handle integers a lot more efficiently than
a generic procedure.  The extra speed comes at the cost of increasing
considerably the size of the trusted code base. To have the best of both
worlds, speed and reliability, SMT solvers may produce proofs that can be later
checked by a small program~\cite{moskal2008proofs}. However, this poses the
extra challenge of producing proof pieces from \emph{each} decision procedure
and then gluing them together.  Simplify~\cite{detlefs2005} was for many years
the solver of choice for program verification.  It has a similar architecture
(being co-developed by Nelson), but understands a language slightly different
than the standardized SMT language~\cite{barrett2010lang}.

Throughout this dissertation, the term `program verifier' is used
usually in a very restricted sense. It refers to a tool that
\begin{enumerate} \item uses an SMT solver as its reasoning core,
\item has a pipeline architecture and an intermediate language,
and \item can be used to generate warnings just like a compiler.
\end{enumerate} The pipeline architecture with an intermediate
language is typical in translators and
compilers~\cite{novillo2004}.

Many tools fit this narrow definition, including
\escjava~\cite{flanagan2002escjava}, the static verifier in
\specsharp~\cite{barnett2005spec} (for \csharp),
HAVOC~\cite{lahiri2006} (for C), VCC~\cite{cohen2009} (for C),
Krakatoa~\cite{filliatre2007why} (for Java), the Jessie plugin in
\framac~\cite{framac} (for C), and JACK~\cite{barthe2006} (for
Java). \escjava\ and JACK support JML~\cite{leavens2006jml} (the
\fb Java \fb modeling \fb language), an annotation language for
Java that has wide support, a reference manual, and even
(preliminary and partial) formal semantics~\cite{leavens2006sem}.
\framac\ supports ACSL~\cite{baudin2008}, another of the few
annotation languages with an adequate reference~manual.

The intermediate language, or at least the intermediate representation, used by
these tools is much better specified.  \escjava uses a variant of Dijkstra's
guarded commands~\cite{dijkstra1975} that has exceptions. \specsharp, HAVOC,
and VCC, which are developed by Microsoft Research, use the Boogie
language~\cite{leino2008boogie,leino2010boogie}. Krakatoa and \framac, which
are developed by INRIA, use the Why language~\cite{filliatre2007why}. The
Boogie language (and the associated verifier from Microsoft
Research) was used also as a high-level interface to SMT solvers in order to
verify algorithms and to explore encoding strategies for high-level
languages~\cite{banerjee2008,summers2010}.

Many tools have more than one reasoning engine. JACK uses Simplify and Coq.
The Why tool uses a wide range of theorem provers: Coq, PVS, Isabelle\slash
HOL, HOL~4, HOL Light, Mizar; Ergo, Simplify, CVC Lite, haRVey, Zenon,
Yices, CVC3.  SLAM uses both a model checker and the SMT solver Z3.

The Boogie tool, FreeBoogie, and the Why tool are fairly
big pieces of code that convert a program written in an intermediate language
into a formula that should be valid if and only if the program is correct.
Moore~\cite{moore2006} argues that a better solution is to explicitly write the
operational semantics, which leads to a much smaller VC generator. It is not
clear what impact this has on speed. However, the technique is most intriguing
and it would be interesting to pursue it in the context of Boogie and Why.
(Moore uses ACL2.) Other tools, like KeY~\cite{key} (for Java, using dynamic
logic~\cite{harel2000}), jStar~\cite{distefano2008jstar} (for Java, using
separation logic~\cite{reynolds2002}), and VeriFast~\cite{verifast} (for C,
using separation logic~\cite{reynolds2002}), avoid the VC generation step
because they rely on symbolic execution~\cite{king1976}.
This means, very roughly, that instead of turning the program into
a proof obligation into one giant step, they `execute' the program and they
keep track of the current possible states using formulas. At each execution
step they may use the help of a reasoning engine. (Note that at this level of
abstraction and hand-waving there is not much difference between symbolic
execution and abstract interpretation~\cite{cousot1977}.)

With so many tools, it is surprising that they do not have a more serious
impact in practice. On this subject one can only speculate. It is probably
true that a closer collaboration between theoreticians and practitioners
would ameliorate the situation. But it is also true that researchers still
have much work to do on \emph{known} problems\index{problem!open}. The
expressivity of specification languages is a family of such problems. For
example, it is still considered a research challenge to annotate
particularly simple patterns, like the Iterator pattern~\cite{haack2009} or
the Composite pattern~\cite{summers2010}. Speed is another problem. As any
(non-Eclipse) computer user will tell you, the number of users of a program
tends to decrease with its average response time.  Since program verifiers
are intended to be used similarly with compilers, their speed is naturally
compared with that of compilers. Right now, the response times of program
verifiers are higher and have a greater variance.

It is, of course, bothersome that in today's state of affairs it
is hard to annotate the Iterator pattern properly in many
verification methodologies. (Otherwise it would not have been the
official challenge of SAVCBS~2006.) But we should not expect
practitioners to annotate such things \emph{at all}, just as we
should not expect them to state the invariant $0\le i\le n$ on
the omnipresent \textbf{for} loop. In general, usable
verification tools should embody knowledge that is common to
people in a specific domain. There is some work on tackling
typical exercises in an introductory programming
course~\cite{leino2007sum} and there is also work in encoding
traditional mathematical knowledge in
verifiers~\cite{buchberger2006}.

A choice that may have seemed unusual is the use of the name
`Roy--Warshall,' instead of the more standard `Floyd--Warshall,' for the
algorithm in Figure~\ref{lst:fw}. There are many algorithms with the shape 
\begin{equation} 
  \hbox{ \textbf{for} $k$ 
    \textbf{do for} $i$, $j$ \textbf{do} $a_{ij} 
      \gets a_{ij} \circ (a_{ik} \bullet a_{kj})$} 
  \index{Roy--Warshall algorithm}
\end{equation} 
and they are sometimes known under the name ``the $kij$ algorithm.''
Roy~\cite{roy1959} and Warshall~\cite{warshall1962tbm} noticed
independently that $\circ = \lor$ and $\bullet = \land$ solves the
transitive closure problem. Floyd~\cite{floyd1962} noticed that $\circ =
\min$ and $\bullet = +$ solves the all pairs shortest paths problem. (He
did this five years before assigning meanings to programs.) Although less
clear, other earlier algorithms are instantiations of the $kij$ schema.
Kleene~\cite{kleene1951} proved that every finite automaton corresponds to
a regular expression. His proof is constructive and is now known as
Kleene's algorithm. The Gauss--Jordan method for solving systems of linear
equations is another example.  Pratt~\cite{pratt1989} discusses the $kij$
algorithm in general terms.

\section{A Guided Tour} \label{sec:intro.tour}

The content of this dissertation touches on issues related to software
engineering, compilers, programming languages, algorithm analysis, and
theorem provers. It is often abstract and theoretical, but sometimes
descends into implementation details.  Therefore, there are plenty of
opportunities for a reader to find something enjoyable. Of course, there
are even more opportunities for a reader to find something not enjoyable,
especially if the reader has a strong preference for one sub-field of
computer science. Fortunately, the dissertation does not stem from one big
contribution, but rather from many smaller ones that are related and yet
can be understood independently. The dependencies between chapters appear
in Figure~\ref{fig:deps}.  Chapter~\ref{ch:boogie} is very important, as it
is required reading for all subsequent chapters.

\begin{figure}\centering
\begin{tikzpicture}[y=\baselineskip]
  \foreach \y/\n/\l in {
    0/boogie/The Core Boogie,
    1/design/Design Overview,
    2/passive/Optimal Passive Form,
    3/spwp/Strongest Postcondition versus Weakest Precondition,
    4/ev/Edit and Verify,
    5/reachability/Semantic Reachability Analysis}
    \node[draw,circle,inner sep=1pt,thick] (\n) at (0,-\y) [label=right:\hbox to 26em{\l\leaders\hbox to 1em{\hfil.\hfil}\hfil \pageref{ch:\n}}] {};
  \foreach \i/\j in {boogie/design,boogie/passive,boogie/spwp,spwp/ev,spwp/reachability}
    \draw[arr,thick] (\i) to [out=180,in=110] (\j);
\end{tikzpicture}
\caption{Dependencies between chapters}\label{fig:deps}
\end{figure}

The scope of the dissertation reflects the broad interests of the
author. Often, broad texts say nothing about a lot, and are quite
boring. The focus on a very specific part of a very specific type
of static analysis tools and the focus on a very small subset of
the Boogie language are meant to ensure enough depth.

My hope is that the dissertation will bring together at
least two people from different sub-fields of computer science to
work together on a common problem related to program verifiers.
\medskip

Chapter~\ref{ch:boogie} presents the syntax and the operational
semantics of a subset of the Boogie language, which is used
throughout subsequent chapters.

Chapter~\ref{ch:design} presents \emph{what} FreeBoogie does. Its
architecture is sketched, including the interfaces between major
components. The advantages and disadvantages of each important
design decision are discussed, so that others who endeavor in
similar tasks avail of our experience and avoid repeating our
mistakes.

Chapter~\ref{ch:passive} presents \emph{how} FreeBoogie replaces
assignments by equivalent assumptions. The algorithm's complexity
is analyzed in detail and its goals are clearly defined. A
natural variant of the problem is proved to be NP-hard.

Chapter~\ref{ch:spwp} presents \emph{how} FreeBoogie generates a
prover query from a Boogie program, using either a weakest
precondition transformer or a strongest postcondition
transformer. In the process, we see how four types of assigning
meanings to Boogie programs relate to each other: operational
semantics, Hoare triples, weakest preconditions, and strongest
postconditions. (Chapters~\ref{ch:passive} and~\ref{ch:spwp} are
based on~\cite{grigore2009spu}.)

Chapter~\ref{ch:ev} presents \emph{how} FreeBoogie exploits the
incremental nature of writing code and annotations in order to
improve its response times. The generic idea, of exploiting what
is known from previous runs of the verifier, is refined and then
proved correct. (This chapter is based on~\cite{grigore2007ev}.)

Chapter~\ref{ch:reachability} presents \emph{what} FreeBoogie does to
protect developers from making silly mistakes that render verification
useless: Inconsistencies make everything provable.  An algorithm that
efficiently searches for inconsistencies is given and is analyzed.
Experimental results show that it is practical.  (This chapter is based
on~\cite{janota2007}.)

Chapter~\ref{ch:conclusions} concludes the dissertation.
Appendix~\ref{ch:notation} summarizes the notation used
throughout the dissertation. The reader is advised to browse that
appendix before continuing with the next chapter.

Software engineers and practitioners are likely to enjoy most
Chapter~\ref{ch:design}. People with a formal background, like
logicians and type theorists, are likely to enjoy most
Chapters~\ref{ch:spwp} and~\ref{ch:ev}. Algorithmists are likely
to enjoy most Chapter~\ref{ch:passive},~\ref{ch:ev}
and~\ref{ch:reachability}. Developers of program verification
tools are likely to enjoy most Chapters~\ref{ch:ev}
and~\ref{ch:reachability}.
\chapter{The Core Boogie}
\label{ch:boogie}

\chquote{The boogie bass is defined as a left-hand rhythmic
pattern, developed from boogie-woogie piano styles, that is
played on the bass string.}{Frederick M. Noad~\cite{noad2002}}

\noindent From now on almost all programs are written in the Boogie
language, and almost all are written in a subset---the core
Boogie---defined in this section. What is `core' is relative to
the topics discussed in this dissertation.

Figure~\ref{lst:first-boogie} shows an implementation of sequential search.
After the counter~$i$ is initialized in line~2 the control goes
nondeterministically to \emph{both}\index{nondeterminism} labels
$b$~and~$d$.  An execution gets \emph{stuck}\index{stuck execution} if it
hits an assumption that does not hold. Since the conditions on the lines 3
and~5 are complementary, exactly one of the two executions will continue.
(In general, the \textbf{if} statement of full Boogie may be desugared into
a \textbf{goto} statement that targets two \textbf{assume} statements with
complementary conditions.) The \textbf{return} statement is reached only if
$i \ge \mathit{vl}$ or $v[i]=u$.

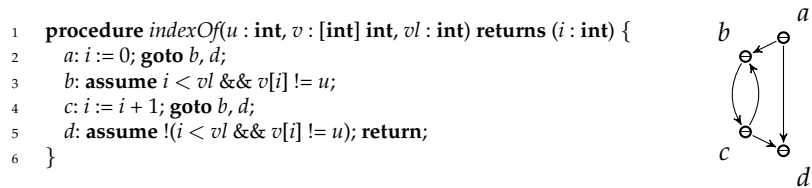
\begin{figure}
\centering
\begin{tabular}{cc}
\begin{boogie}[boxpos=c]
procedure indexOf(u : int, v : [int] int, vl : int) returns (i : int) {
  a: i := 0; goto b, d;
  b: assume i < vl && v[i] != u;
  c: i := i + 1; goto b, d;
  d: assume !(i < vl && v[i] != u); return;
}
\end{boogie}&
\hspace{5mm}
\begin{tikzpicture}[scale=.5,baseline=0.7cm]
  \oonode (A) at (1,3)  [label=60:$a$] {};
  \oonode (B) at (0,2.5)  [label=135:$b$] {};
  \oonode (C) at (0,0.5)  [label=-135:$c$] {};
  \oonode (D) at (1,0)  [label=-60:$d$] {};

  \draw[arr] (A) -- (B);
  \draw[arr] (A) -- (D);
  \draw[arr] (B) to [bend right=30] (C);
  \draw[arr] (C) to [bend right=30] (B);
  \draw[arr] (C) -- (D);
\end{tikzpicture}
\\
\end{tabular}
\caption{Boogie program with one loop and its flowgraph}
\label{lst:first-boogie}
\end{figure}

The high-level constructs of full Boogie (such as \textbf{if}
statements, \textbf{while} statements, and \textbf{call}
statements) can always be desugared into the core that is
formalized here. The concrete grammar of core statements
appears in Figure~\ref{grm:boogie-stmts}. The statement
\textbf{assert}~$p$ means ``when you reach this point, check that
$p$ holds.'' The statement \textbf{assume}~$p$ means ``continue
to look only at executions that satisfy~$p$.''

\startgrammar
  statement\is label\? (assignment \| assumption \| assertion \| jump) \b; \\
  label\is \i{id} \b{:} \\
  assignment\is id \b{:=} expression \\
  assumption\is \b{assume} expression \\
  assertion\is \b{assert} expression \\
  jump\is (\b{goto} id-tuple \b{}) \| \b{return} \\
  id-tuple\is \i{id} (\b, \i{id})\* \\
\stopgrammar{Syntax of core Boogie statements}{grm:boogie-stmts}

The type system of full Boogie is rich, featuring polymorphic
maps, bit vectors, and user-defined types, among others. Its
expression language is similarly rich. Unlike in the case
of statements, the VC generator implementation (henceforth
known as FreeBoogie) does not desugar types and expressions
into simpler ones. How they are treated, however, is not
novel. For the sake of the presentation, only a few types
and expressions are retained in the core, the ones in
Figure~\ref{grm:exprs-and-types}.

\startgrammar
  type\is primitive-type \| map-type \\
  primitive-type\is \b{int} \| \b{bool} \\
  map-type\is \b[ primitive-type \b] type \\
  expression\is \b( quantifier \i{id} \b: type \b{::} expression \b) \\
  expression\is unary-operator expression \| expression binary-operator expression \\
  expression\is expression \b[ expression \b] \\
  expression\is \i{id} \| literal \| \b( expression \b) \\
  quantifier\is \b{forall} \| \b{exists} \\
  unary-operator\is \b! \| \b- \\
  binary-operator\is \b+ \| \b- \| {\boldmath$<$} \| \b{==} \| \b{\&\&} \| {\boldmath$||$} \\
  literal\is \b{true} \| \b{false} \| \b0 \| \b1 \| \b2 \| \ldots \\
\stopgrammar{Syntax of core Boogie types and expressions}{grm:exprs-and-types}

The syntax for the overall structure of core Boogie programs
appears in Figure~\ref{grm:boogie}. Statements are preceded
by variable declarations, in which the variables representing
the input and the output are singled out. The keyword
\textbf{procedure} is retained from full Boogie, where a program
may contain more than one procedure and where there is also a
\textbf{call} statement. Procedures are not included in the core,
because there is nothing novel related to procedure calls in
this dissertation. The mandatory \textbf{return} statement at
the end of each body reduces the number of special cases that
need to be discussed later. FreeBoogie inserts such a statement
automatically during parsing, so the user does not have to end
all procedures with \textbf{return}.

\startgrammar
  program\is \b{procedure} \i{id} \b( arguments\? \b{) returns (} results\? \b) body \\
  arguments, results\is \i{id} \b: type (\b, \i{id} \b: type)\* \\
  body\is variable-declaration\* statement\* \b{return} \b; \\
  variable-declaration\is \b{var} \i{id} \b: type \b; \\
\stopgrammar{Structure of a (core) Boogie program}{grm:boogie}

Typechecking core Boogie is straightforward:
Figure~\ref{fig:boogie-typing} shows a representative sample of
the typing rules. The judgment $\vdash p$ means that the program
fragment~$p$ is well-typed and the judgment $\vdash p:t$ means
that the program fragment~$p$ is well-typed \emph{and} has the
type~$t$. The customary environment is omitted because it is
fixed---it consists of all the variable typings appearing at the
beginning of the program. The rules that are missing are similar
to the ones given: For example, there is a rule that says that
the expression appearing in an assumption must have the type
\textbf{bool} (analogous to rule [asrt]).

\begin{figure}
\centering
\begin{displaymath}
\begin{array}{ccc}
  \frac
    {\exists t,\;  \vdash v : t \;\land\; \vdash e : t}
    { \vdash v \boldsymbol{:=} e}
    \;\text{[asgn]}
  &&
  \frac{\vdash e : \mathbf{bool}}{\vdash \mathbf{assert}\;e}
    \;\text{[asrt]}
  \\ \\
  \frac
    {\vdash e : \mathbf{bool} \quad \vdash f : \mathbf{bool}}
    {\vdash e \:\boldsymbol{\&\&}\: f : \mathbf{bool}}
    \;\text{[bool]}
  &&
  \frac
    {\vdash e : \mathbf{int} \quad \vdash f : \mathbf{int}}
    {\vdash e \:\boldsymbol{+}\: f : \mathbf{int}}
    \;\text{[arith]}
  \\ \\
  \frac
    {\vdash e : \mathbf{int} \quad \vdash f : \mathbf{int}}
    {\vdash e \:\boldsymbol{<}\: f : \mathbf{bool}}
    \;\text{[comp]}
  &&
  \frac
    {\exists t,\; \vdash e : t \;\land\; \vdash f : t}
    {\vdash e \:\boldsymbol{==}\: f : \mathbf{bool}}
    \;\text{[eq]}
  \\ \\
  \frac{}{\vdash\mathbf{true}:\mathbf{bool}} \;\text{[lit-bool]}
  &&
  \frac{}{\vdash\boldsymbol{0}:\mathbf{int}} \;\text{[lit-int]}
\end{array}
\end{displaymath}
\caption{Typing rules for core Boogie}\label{fig:boogie-typing}
\end{figure}

\paragraph{Operational Semantics}\index{operational semantics}

The set of all variables is denoted by~\textit{Variable}\index{variable}.
It can be thought of as the set of all identifiers or as the set of all
strings. The set of all values is denoted by~\textit{Value}\index{value}
and it contains the set of booleans~$\mathbb{B}=\{\tru,\fls\}$.
\emph{Stores}\index{store} assign values to variables.
\begin{align}
\mathit{Store} &= \mathit{Variable} \to \mathit{Value} \\
\sigma &\in\mathit{Store}
\end{align}
\emph{Expressions}\index{expression} assign values to a stores. Boolean
expressions, called \emph{predicates}\index{predicate}, define sets of
stores.
\begin{align}
\mathit{Expression} &= \mathit{Store} \to \mathit{Value} \\
\mathit{Predicate} &= \mathit{Store} \to \mathbb{B} \\
p,q,r &\in \mathit{Predicate}
\end{align}

According to the syntax, the program is a list of statements,
which means that we can assign counters $0$,~$1$, $2$,~$\ldots$
to them. To simplify the presentation, we will assume that all
labels are counters (in the proper range). The \emph{state} of
a Boogie program is either the special \textit{error} state
or a pair~$\langle\sigma,c\rangle$ of a store~$\sigma$ and a
counter~$c$ of the statement about to be executed. The following
rules define a relation~$\leadsto$ on states, thus giving an
operational semantics for the core of the Boogie language.
\begin{equation}
\frac
  {p\;\sigma}
  {\langle\sigma,c:(\mathbf{assume}/\mathbf{assert}\;p)\rangle \leadsto 
    \langle\sigma,c+1\rangle}
  \label{eq:assume-assert-ok-opsem}
\end{equation}
\begin{equation}
\frac
  {\lnot(p\;\sigma)}
  {\langle\sigma,c:(\mathbf{assert}\;p)\rangle \leadsto \mathit{error}}
  \label{eq:assert-nok-opsem}
\end{equation}
\begin{equation}
\frac
  {}
  {\langle\sigma,c:(v\mathtt{:=}e)\rangle \leadsto 
    \langle(v\gets e)\;\sigma,c+1\rangle}
  \label{eq:assign-opsem}
\end{equation}
\begin{equation}
\frac
  {c' \in \ell}
  {\langle\sigma,c:(\mathbf{goto}\;\ell)\rangle \leadsto 
    \langle\sigma,c'\rangle}
  \label{eq:goto-opsem}
\end{equation}
A rule $\frac{h}{\langle\sigma,\,c:\mathscr{P}\rangle\leadsto
s}$ means that the program may evolve from
state~$\langle\sigma,c\rangle$ to the state~$s$ if the
hypothesis~$h$ holds and if the counter~$c$ corresponds to a
statement that matches the pattern~$\mathscr{P}$. For example,
rule~\eqref{eq:goto-opsem} says that $\langle\sigma,c\rangle$
may evolve into $\langle\sigma,c'\rangle$ if the statement
at counter~$c$ matches the pattern \textbf{goto}~$\ell$
and~$c'\in\ell$. The notation $p\;\sigma$ stands for the
application of function~$p$ to the argument~$\sigma$. Space (that
is, the function application operator) is left associative, has
the highest precedence, and, as is customary, it is omitted if
there is only one parenthesized argument that is not followed by
another function application. The notation $(v \gets e)$ used
in rule~\eqref{eq:assign-opsem} stands for a store transformer,
defined as follows.
\begin{align}
(v \gets e) &: \mathit{Store}\to\mathit{Store} \\
(v \gets e)\;\sigma\;w &=
  \begin{cases}
  \sigma\;w& \text{if $v\ne w$}\\
  e\;\sigma& \text{if $v=w$}
  \end{cases}
\end{align}

\begin{definition}\index{execution}
An \emph{execution} of a core Boogie program is a sequence
$s_0$,~$s_1$, \dots,~$s_n$ of states such that $s_{k-1}\leadsto
s_k$ for all $k\in1.\,.\>n$ and $s_0=\langle\sigma_0,0\rangle$
for some arbitrary \emph{initial store}~$\sigma_0$.
\label{def:execution}
\end{definition}

\begin{remark}
All executions are finite. This is an unusual simplification, but one that can
be made as long as we do not investigate whether programs terminate.
\end{remark}

The sources of nondeterminism\index{nondeterminism} in core Boogie are (1)~the
initial store~$\sigma_0$ and (2)~the \textbf{goto}
rule~\eqref{eq:goto-opsem} which allows multiple successors.

We say that an execution $s_0$,~$s_1$, \dots,~$s_n$ \emph{goes
wrong} if $s_n=\mathit{error}$. This can happen only if the last
statement that was executed was an assertion. We say that that
assertion was \emph{violated}\index{violated assertion}.

\goodbreak
\begin{definition}\index{correctness}
A core Boogie program is \emph{correct} when none of its executions
goes~wrong.
\label{def:correctness}
\end{definition}

Boogie does not facilitate reasoning about termination.
Throughout the dissertation the term ``correct'' will usually
mean what is traditionally referred to as ``partially correct.''
Because we are not interested in termination, we do not
distinguish between executions that reach an assumption that
is not satisfied and executions that reach a \textbf{return}
statement: both get stuck\index{stuck execution}.

Let us look at an example. If we turn the labels of
the program in Figure~\ref{lst:first-boogie} (on
page~\pageref{lst:first-boogie}) into counters we obtain
\begin{boogie}[numbers=none] 
0:$\quad$ i := 0; 
1:$\quad$   goto 2, 5; 
2:$\quad$   assume i < vl && v[i] != u; 
3:$\quad$   i := i + 1; 
4:$\quad$   goto 2, 5; 
5:$\quad$   assume !(i < vl && v[i] != u); 
6:$\quad$   return;
\end{boogie}
for which one possible execution is {\def\<{\langle}
\def\>{\rangle} \def\@{\displaybreak[0]\\}
\begin{align}
  &\<\sigma,0\> \@
  &\<(i\gets0)\;\sigma,1\>\@
  &\<(i\gets0)\;\sigma,2\>\@
  &\big(i<\mathit{vl}\land v[i]\ne u\big)\;\big((i\gets0)\;\sigma\big)  &(*)\@
  &\<(i\gets 0)\;\sigma,3\>\@
  &\<(i\gets i+1)\;((i\gets0)\;\sigma),4\>\@
  &\<(i\gets i+1)\;((i\gets0)\;\sigma),5\>\@
  &\lnot\big(i<\mathit{vl}\land v[i]\ne u\big)\;\big((i\gets i+1)\;((i\gets0)\;\sigma)\big)  &(*)\@
  &\<(i\gets i+1)\;((i\gets0)\;\sigma),6\>
\end{align}}
(The store~$\sigma$ is arbitrary but fixed.) The lines marked with~$(*)$ are not states but rather
conditions that are assumed to hold. In order to evaluate those
conditions we need to look inside the predicates. Every $n$ary
function $f:(\mathit{Value}\to)^n\mathit{Value}$ has a
corresponding function~$f'$ on expressions:
\begin{equation}
\begin{aligned}
f' &: (\mathit{Expression}\to)^n\mathit{Expression} \\
f'\;e_1\ldots e_n\;\sigma &= f\;(e_1\;\sigma)\ldots(e_n\;\sigma)
\end{aligned}
\label{eq:abuse-value-ops}
\end{equation}
In particular, boolean functions have corresponding
predicate combinators. By notation abuse, we write
$\land$,~$\lor$,~$\ldots$ 
between booleans as well as between
predicates. Also, the boolean constants $\tru$~and~$\fls$ are 
boolean functions with arity~$0$, so we shall abuse them too
and write $\tru$~and~$\fls$ for constant predicates. Similarly,
we will lift the other operators. An example evaluation of a
predicate~follows.
\begin{align*}
    &\big(i<\mathit{vl}\land v[i]\ne u\big)\;\big((i\gets0)\;\sigma\big) \\
=\quad &(i<\mathit{vl})\;\big((i\gets0)\;\sigma\big)
    \land (v[i]\ne u)\;\big((i\gets0)\;\sigma\big)\\
=\quad &i\;\big((i\gets0)\;\sigma\big) < \mathit{vl}\;\big((i\gets0)\;\sigma\big) \land\ldots
\end{align*}
A predicate that consists of a variable~$v$ is
evaluated by reading the variable's value from the store; a
predicate that consists of a constant~$c$ evaluates to that
(lifted) constant. (For example, the predicate~$0$ is lifted to
the integer~$0$.)
\begin{align}
  v\;\sigma&=\sigma\;v\\
  c\;\sigma&=c
\end{align}
The evaluation continues as follows.
\begin{align}
    &i\;\big((i\gets0)\;\sigma\big) < \mathit{vl}\;\big((i\gets0)\;\sigma\big)\\
=\quad &(i\gets0)\;\sigma\;i < (i\gets0)\;\sigma\;\mathit{vl} \\
=\quad &0<\sigma\;\mathit{vl}
\end{align}
And we conclude that for the previous example execution
the initial store~$\sigma$ must satisfy~$0<(\sigma\;\mathit{vl})$.

A special case of~\eqref{eq:abuse-value-ops} is
equality. Equality between values is a function
$=\;:\mathit{Value}\to\mathit{Value}\to\mathbb{B}$ that has a
corresponding expression combinator $=\;:\mathit{Expression}\to
\mathit{Expression}\to \mathit{Predicate}$. For example, we write
$v=e$ for a predicate that defines the set of stores in which the
variable~$v$ and the expression~$e$ evaluate to the same value.

We say that predicate~$p$ is \emph{valid} and we write $|p|$ when
it holds for all stores.
\begin{equation}\index{validity!of predicates}
|p|=(\forall\sigma,\;p\;\sigma) \label{eq:valid_notation}
\end{equation}
We will often need to say that two predicates delimit exactly
the same set of stores, so we introduce a shorthand notation for
it, which will also be useful when defining (syntactically) new
predicates.
\begin{equation}\index{equivalence of predicates}
(p\equiv q) = |p=q|
\end{equation}

Finally, we shall abuse notation and write $(v \gets e)\;p$ for
predicates~$p$, based on the fact that every store transformer
$f:\mathit{Store}\to\mathit{Store}$ has a corresponding
expression transformer~$f'$:
\begin{align}
f' &: \mathit{Expression}\to\mathit{Expression} \\
f'\;e\;\sigma &= e\;(f\;\sigma)
\end{align}

This concludes the presentation of the core of
Boogie that is used often in the next chapters. 

\paragraph{Prior Work}

The semantics given here for core Boogie correspond to the trace semantics
of Leino~\cite{leino2008boogie}. Simple statements are treated
essentially as input--output relations~\cite[Chapter~6]{hesselink1992}.
Unusually, this chapter does not define explicitly any way to compose
statements. Instead, a program counter and \textbf{goto} statements allow
any control flowgraph, in the style of Floyd~\cite{floyd1967}.

The full Boogie language~\cite{leino2010boogie} is more high-level and more
user-friendly.

\chapter{Design Overview}
\label{ch:design}

\chquote{One of the best indications that a program is the
result of the activity of design is the existence of a document
describing that design.} {Jim Waldo~\cite{waldo2006}}

\noindent This chapter is rather dense. A cursory read will show
how the later, more theoretical chapters fit together in the
context of a real program; a careful read will serve as a guide
to FreeBoogie's source code for new developers.

\section{An Example Run}

\bc
\begin{boogie}
type T;
procedure indexOf(x : T, a : [int] T, al : int) returns (i : int) {
  i := 0;
  while (i < al && a[i] != x) { i := i + 1; }
}
\end{boogie}
\ec{A high-level Boogie version of Figure~\ref{lst:first-boogie}}
{lst:boogie-indexof}

The best way to understand how FreeBoogie works is to run it on a few
examples and ask it to dump its data structures at intermediate stages.
Figure~\ref{lst:boogie-indexof} shows a Boogie program suitable for a first
run. Notice that the language is not restricted to the core defined in
Chapter~\ref{ch:boogie}. To peek at FreeBoogie's internals use the~command
\shell{fb --dump-intermediate-stages=log example.bpl} assuming that you
wrote the content of Figure~\ref{lst:first-boogie} in the file
\texttt{example.bpl} and that FreeBoogie is correctly installed on your
system.  This will create a directory named \texttt{log}.  The output of
each processing phase of FreeBoogie appears in a subdirectory
of~\texttt{log}.

Such transformation phases include desugaring the \textbf{while} statement,
desugaring the \textbf{if} statement, cutting cycles, eliminating
assignments (Chapter~\ref{ch:passive}), computing the VC
(Chapter~\ref{ch:spwp}). Let us look briefly at the state of FreeBoogie
after \textbf{while} and \textbf{if} statements are desugared. To do so we
could look at the pretty-printed Boogie code but we can also look at the
flowgraph that is dumped by FreeBoogie in the GraphViz~\cite{ellson2001}
format. A command like \shell{dot log/freeboogie.IfDesugarer/*.dot -Tpng >
fig3.2.png} produces Figure~\ref{fig:first-boogie}.  Such drawings of the
internal data structures are very helpful in understanding FreeBoogie and
in debugging it. For example, in Figure~\ref{fig:first-boogie} we see that FreeBoogie introduces labels prefixed by {\tt \$\$} and a string like
\texttt{if} or \texttt{while}, which hints to the origin of the label. 

\begin{figure}
\includegraphics[width=\textwidth]{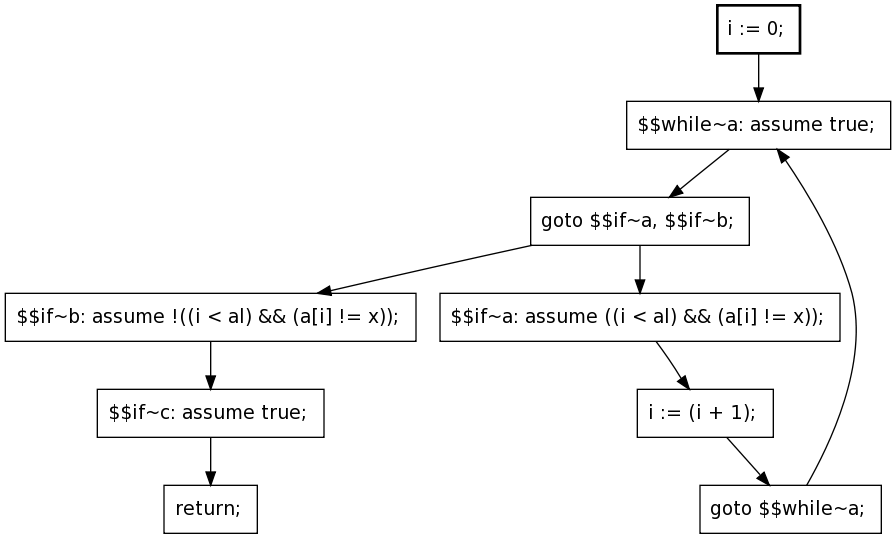}
\caption{Flowgraph of a desugared version of Figure~\ref{lst:first-boogie}}
\label{fig:first-boogie}
\end{figure}

The flowgraph is an example of \emph{auxiliary} information\index{auxiliary
information} that FreeBoogie computes after each transformation. The other
pieces of auxiliary information are the symbol table and the types. The
\emph{symbol table}\index{symbol table} is a one-to-many bidirectional map
between identifier definitions and identifier uses. The \emph{types} are
associated with expressions (and subexpressions).

To see the query that is sent to the theorem prover you must
run a different command, this time shown in abbreviated form:
\shell{fb -lf=example.log -ll=info -lc=prover example.bpl} The
\fb log \fb file \texttt{example.log} will contain everything sent to
the prover. FreeBoogie prints \shell{OK: indexOf at
example.bpl:2:11} indicating that the program is correct.

\section{Pipeline}
\label{sec:pipeline}

Figure~\ref{fig:architecture} shows that FreeBoogie has a
pipeline architecture. The \colorbox{green!50!white}{light green}
color stands for the Boogie AST (\fb abstract \fb syntax \fb
tree); the \colorbox{blue!50!black}{\textcolor{white}{dark blue}}
color stands for the SMT AST\null.

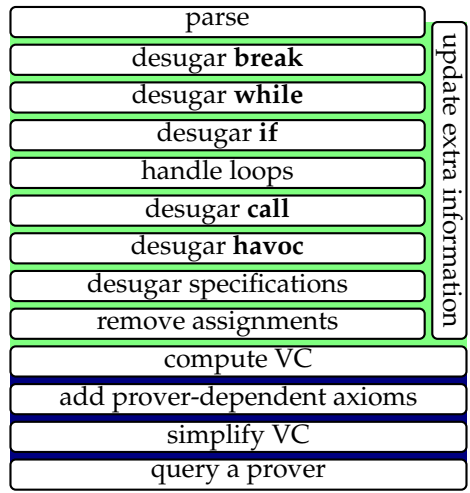
\begin{figure}
  \centering
  \begin{tikzpicture}[scale=1]
  \tikzstyle phaseS=[thick,rounded corners=2pt,draw=black,fill=white];
  \fill[color=green!50!white] (0,-0.25) rectangle +(6,-4.75);
  \fill[color=blue!50!black] (0,-4.75) rectangle +(6,-1.5);
  \foreach \py/\px/\ptext in {
    0/5.45/parse,
    1/5.45/desugar \textbf{break},
    2/5.45/desugar \textbf{while},
    3/5.45/desugar \textbf{if},
    4/5.45/handle loops,
    5/5.45/desugar \textbf{call},
    6/5.45/desugar \textbf{havoc},
    7/5.45/desugar specifications,
    8/5.45/remove assignments,
    9/6/compute VC,
    10/6/add prover-dependent axioms,
    11/6/simplify VC,
    12/6/query a prover}
  {
    \draw[phaseS]
     [yscale=.5](0,-\py-0.1) rectangle node {\ptext}  +(\px,-.8);
  }
  \draw[phaseS] (5.55,-0.25) rectangle node[rotate=-90] 
    {update extra information} (6,-4.45);

  \end{tikzpicture}
  \caption{FreeBoogie architecture.}
  \label{fig:architecture}
\end{figure}

Horizontal boxes, except for the first one (parse) and the last
one (query a prover), represent \emph{transformations}. Depending
on the type of input and on the type of output there are three
types of transformations: Boogie to Boogie, Boogie to SMT, and
SMT to SMT\null. For brevity, we say `Boogie transformations'
instead of `Boogie to Boogie transformations', and `SMT
transformations' instead of `SMT to SMT transformations'. All
these transformations are designed not to miss bugs, at the cost
of possible false positives.

\begin{definition}\index{soundness}
A Boogie transformation is \emph{sound} when it produces only
incorrect Boogie programs from incorrect Boogie programs. A
Boogie to SMT transformation is \emph{sound} when it produces
only invalid formulas from incorrect Boogie programs. An SMT
transformation is \emph{sound} when it produces only invalid
formulas from invalid formulas.
\label{def:sound-transform}
\end{definition}

\begin{remark}
This definition is in a way formal, but in a way it is not.
It makes use of the concept of `correct Boogie program'
and we only have semantics for \emph{core} Boogie programs
(Chapter~\ref{ch:boogie}). For example, we can say precisely what
it means for the assignment removal transformation to be sound,
because both the input and the output of that transformation are
core Boogie programs; however, we can only informally describe
the preceding transformations as sound.
\end{remark}

The symmetric notion is that of completeness.

\begin{definition}\index{completeness}
A Boogie transformation is \emph{complete} when it produces
only correct Boogie programs from correct Boogie programs. A
Boogie to SMT transformation is \emph{complete} when it produces
only valid formulas from correct Boogie programs. An SMT
transformation is \emph{complete} when it produces only valid
formulas from valid formulas.
\label{def:complete-transform}
\end{definition}

All transformations in FreeBoogie are sound; all transformations
in FreeBoogie are complete, except loop handling.

Full Boogie would not be user friendly without high-level
constructs like \textbf{while} statements and \textbf{break}
statements. Many phases in FreeBoogie perform syntactic
desugarings of these constructs. The desugaring is sometimes
local, in the sense that it can be done without keeping track of
an environment, and sometimes it is not. For example, to desugar
the \textbf{break} statement we must keep track of the enclosing
\textbf{while} and \textbf{if} statements; but the desugaring of
a \textbf{havoc} statement does not depend on any surrounding
code.

The most important transformation in FreeBoogie is the
transition from Boogie to SMT\null. The role of the preceding Boogie
transformations is to simplify the program to a form on which the
VC is easily computed; the role of subsequent SMT transformations
is to bring the VC to a form that is easily handled by an SMT solver.

The order of Boogie transformations depends on constraints
such as the following. The Boogie to SMT transformation
(`compute VC' in Figure~\ref{fig:architecture}) only handles
the \textbf{assert}, \textbf{assume}, and \textbf{goto}
statements. The Boogie transformation that removes assignments
only handles acyclic flowgraphs. Hence, the flowgraph must
first be transformed into an acyclic one (`handle loops' in
Figure~\ref{fig:architecture}).

The VC uses concepts such as arrays, which may or may not be
known to the prover. In the latter case, axioms that describe
the concept must be added to the VC\null. Finally, the VC is
simplified so that the communication with the prover is
more efficient.

The source code of FreeBoogie, written in Java~6, contains
four packages, which are in one-to-one correspondence with the
following sections.

\begin{itemize}
\item \textit{freeboogie.ast}: data structures to represent
  Boogie programs
\item \textit{freeboogie.tc}: computing auxiliary information
  from a Boogie AST
\item \textit{freeboogie.vcgen}: the Boogie transformations and
  the Boogie to SMT transformation
\item \textit{freeboogie.backend}: 
  the SMT transformations, 
  data structures to represent SMT formulas,
  communication with SMT solvers
\end{itemize}

\section{The Abstract Syntax Tree and its Visitors}
\label{sec:design.ast}

The Boogie AST data structures are described using a compact
notation. A subset, corresponding to core Boogie, appears in
Figure~\ref{fig:boogie-absgrm}. AstGen (a helper tool) reads this
description and a code template to produce Java classes. The
approach has advantages and disadvantages. The generated classes
are very similar to each other because they come from the same
template. This means that it is easy to learn their interface.
It also means that it is easier to change all the classes in a
consistent way by changing the template. The compact description
in Figure~\ref{fig:boogie-absgrm} is easier to read than the
corresponding $12$~Java classes. The overall structure of the AST
is easier to grasp. It is also easier to modify, since it takes
far less time to change one or two lines than one or two Java
classes. However, the programmer needs to learn a new language
(the one used in Figure~\ref{fig:boogie-absgrm}) and IDEs are
usually confused by code generators.

\begin{figure}\centering\footnotesize
\begin{verbatim}
Program = Signature! sig, Body! body;
Signature = String! name, [list] VariableDecl args, [list] VariableDecl results;
Body = [list] VariableDecl vars, Block! block;
VariableDecl = String! name, Type! type;
Block = [list] Statement statements; 
Type = enum(Ptype: BOOL, INT) ptype,
Statement :> AssertAssumeStmt, AssignmentStmt, GotoStmt;
AssertAssumeStmt = enum(StmtType: ASSERT, ASSUME) type, 
  [list] Identifier typeArgs, Expr! expr;
AssignmentStmt = Identifier lhs, Expr rhs;
GotoStmt = [list] String successors;
Identifier = String! id, 
\end{verbatim}
\caption{The abstract grammar of core Boogie}\label{fig:boogie-absgrm}
\end{figure}

Another consequence of this approach, which might be seen as a
disadvantage, is that there is no way to add specific code to
specific classes: We are forced to implement operations over the
AST using the visitor pattern~\cite{gamma1995}. In passing, note
that if the target language would have been \csharp, then one
\emph{could} add specific code to specific classes by using
partial classes.

\subsection{The Abstract Grammar Language}

AstGen reads a description of an abstract grammar and a template.
Therefore it understands two languages---the AstGen abstract
grammar language and the AstGen template language. This section
describes the AstGen abstract grammar language, which was already
used in Figure~\ref{fig:boogie-absgrm}. The syntax of the AstGen
abstract grammar language appears in Figure~\ref{grm:astgen}.
(Note that a formal language is used to describe the concrete
syntax of a language that is used to describe the abstract
grammar of a language whose concrete syntax was studied in
Chapter~\ref{ch:boogie} using the same formal language that
we use here in Figure~\ref{grm:astgen}: There is some opportunity
for confusion.)

\startgrammar
  grammar\is rule\*\\
  rule\is composition \| inheritance \| specification\\
  composition\is class \b= members\? \b;\\
  inheritance\is class {\boldmath$:>$} classes\? \b;\\
  specification\is class \b: \i{text} \P\\
  class\is \i{id}\\
  members\is member (\b, members)\*\\
  classes\is class (\b, classes)\*\\
  member\is tags type \b!\? name\\
  tags\is tag (\b, tags)\*\\
  type\is \i{id} \| \b{enum} \b( \i{id} \b: \i{id} (\b, \i{id})\* \b)\\
  name\is \i{id}\\
  tag\is \b[ \i{id} \b]\\
\stopgrammar{The syntax of the abstract grammar language}{grm:astgen}

The abstract grammar is described by a list of rules. Each
rule starts with the name of the class to which it pertains. A
composition rule continues with an equal sign (\textbf{=}) and a
list of members. An inheritance rule continues with a supertype
sign ({\boldmath$:>$}) and a list of subclasses. A specification rule
continues with a colon (\textbf{:}) and some arbitrary text. The
end of composition and inheritance rules is marked by a semicolon
(\textbf{;}) and the end of a specification rule is marked by the
end of the line (depicted as \P\ in Figure~\ref{grm:astgen}).
That is (almost) all.

To illustrate why this notation is beneficial, suppose that
initially the data structures for the Boogie AST contained only
public fields.
\begin{jml}
public class Program {
  public Signature sig;
  public Body body;
}
\end{jml}
This Java code is obviously not much longer and indeed very similar to the
first line in Figure~\ref{fig:boogie-absgrm}.  But it has a number or
problems.  First, we probably want the \textit{Program} class to be final.
Without using AstGen we must go and add the keyword \textbf{final} in each
class: \textit{Program}, \textit{Signature}, \textit{Body},
\textit{VariableDecl}, $\ldots$ With AstGen, we only need to add that
keyword in the template. Another problem is that there is no constructor.
Again, adding constructors is a repetitive job if we must do it in each and
every class. Finally, FreeBoogie's data structures are immutable. More
precisely, the members are private and final, they are set by the
constructor, and accessor methods only allow them to be read. Again, making
these changes in all classes is a repetitive job.  In summary, the main
advantage of using AstGen is that we separate the concern of defining the
shape of the abstract grammar from lower-level concerns such as whether we
allow subclassing or not, whether we allow mutations or not.

Specification rules, which refer to classes, and tags, which refer to
members (see Figure~\ref{grm:astgen}), allow for a little non-uniformity in
the generated code. The bang sign (\textbf{!}) is a shorthand for the tag
[nonnull]. As we will see, the code template may contain parts that are
used or not by AstGen depending on whether a tag is present. In particular,
FreeBoogie's code template says ``\textbf{assert}~$x\ne\mathbf{null}$''
when member~$x$ has the tag [nonnull].  The only other tag used in
Figure~\ref{fig:boogie-absgrm} is [list], which will be discussed briefly
in Section~\ref{sec:design.immutability}.

A specification rule associates some arbitrary text to a class.
Templates then instruct AstGen where in the output to insert
the arbitrary text. For example, in FreeBoogie the arbitrary
text is always a side-effect free Java boolean expression.
FreeBoogie's code template inserts these boolean expressions in
Java \textbf{assert} statements within constructors. In general,
the intended use of specification rules is to give object or
class invariants. However, there is nothing in AstGen to enforce
this use. Hence, specification rules could be abused to insert,
say, custom comments in the header of generated classes.

\subsection{AstGen Templates}

Figure~\ref{fig:astgen-template} illustrates the main
characteristics of an AstGen template and, at the same time,
gives some details on how the Boogie AST data structures
are implemented. The language for templates is influenced
by \TeX~\cite{knuth1986tex}. Macros start with a backslash
(\macro{}) and may take arguments. Some macros are primitive and
some are defined using \macro{def}. Before cataloging primitive
macros, let us analyze the high-level structure of the template
in Figure~\ref{fig:astgen-template}.

\begin{figure}\centering\footnotesize
\begin{verbatim}
\def{smt}{\if_primitive{\Membertype}{\MemberType}}
\def{mt}{\if_tagged{list}{ImmutableList<}{}\smt\if_tagged{list}{>}{}}
\def{mtn}{\mt \memberName}
\def{mtn_list}{\members[,]{\mtn}}

\classes{\file{\ClassName.java}
/* ... package specification and some imports ... */
public \if_terminal{final}{abstract} class \ClassName extends \BaseName {
\if_terminal{
  \members{private final \mtn;}
  private \ClassName(\mtn_list) {
    \members{this.\memberName = \memberName;}
    checkInvariant();
  }
  public static \ClassName mk(\mtn_list) {
    return new \ClassName(\members[,]{\memberName});
  }
  public void checkInvariant() {
    assert location != null;
    \members{\if_tagged{nonnull|list}{assert \memberName != null;}{}}
    \invariants{assert \inv;}
  }
  \members{public \mtn() { return \memberName; }}
  @Override public <R> R eval(Evaluator<R> evaluator) { 
    return evaluator.eval(this); 
  }
}{
  \selfmembers{public abstract \mtn();}
}} }
\end{verbatim}
\caption{Excerpt from the AstGen template for Boogie AST classes}
\label{fig:astgen-template}
\end{figure}

The first four lines define macros that are used later. AstGen
then sees the (primitive) \macro{classes} macro and processes its
argument once for each class in the abstract grammar. Terminal
classes, which are those without subclasses, have private fields, a private constructor, a static factory method \textit{mk}, a
method \textit{checkInvariant}, accessors for getting the values
of the fields, and a method \textit{eval}, which is typically
called \textit{accept} in most presentations of the visitor
pattern. Non-terminal classes only have abstract accessors for
getting the values of the fields.

The primitive macros can be grouped in four categories:
(1)~output selection, (2)~data, (3)~test, and (4)~iteration.

The macro \macro{file}\{\thinspace\textit{f}\} globally directs
the output from now on to the file~\textit{f}.

The data macros do not take any parameter. They expand to the
name of the current class (\macro{className}), the name of the
base class of the current class (\macro{baseName}), the type
of the current member (\macro{memberType}), the name of the
current member (\macro{memberName}), the name of the current
enumeration (\macro{enumName}), the current enumeration value
(\macro{valueName}), the current invariant (\macro{inv}).
The `current' class\slash member\slash enumeration\slash
value\slash invariant is determined by the enclosing iteration
macros. All data macros except \macro{inv} are made of two
words and they come in four case conventions (camelCase,
PascalCase, lower\_case, and UPPER\_CASE): The output is
formatted accordingly.

The test macros have the shape
\macro{if}\textit{condition}\{\textit{yes}\}\{\textit{no}\}.
If the condition holds then the \textit{yes} part is processed
and the \textit{no} part is ignored; if the condition does not
hold then the \textit{no} part is processed and the \textit{yes}
part is ignored. The braces in the \textit{yes} and \textit{no}
parts must be balanced. The condition \texttt{\_primitive}
holds when the type of the current member does not appear on
the left hand side of a composition rule. (In particular,
it holds for members whose type is an enumeration.) The
condition \texttt{\_enum} holds when the type of the current
member is an enumeration. The condition \texttt{\_terminal}
holds when the current class has no subclass. The condition
\texttt{\_tagged$\{$\textrm{\textit{tag\_expression}}$\}$}
is more interesting. A tag expression may contain tag names,
logical-and (\texttt{\&}), logical-or (\texttt{|}), and
parentheses. A tag name evaluates to~$\tru$ when the current
member has that tag.

The iteration macros have the shape
\macro{}\textit{macro}[\textit{separator}]\{\textit{argument}\}.
The argument is processed repeatedly and the optional separator
is copied to the output between two passes over the argument. The
macro \macro{classes} processes its argument for each class (and
hence each pass has a `current class'). The macro \macro{members}
processes its argument for each member of the current class,
including inherited members (and hence each pass has a `current
member'). The macro \macro{selfmembers} processes its argument
for each member of the current class, excluding inherited
members (and hence each pass has a `current member'). The macro
\macro{invariants} processes its argument for each invariant of
the current class, which appear in specification rules (and hence
each pass has a `current invariant'). The macro \macro{enums}
processes its argument for each enumeration used as a type in the
current class (and hence each pass has a `current enumeration').
The macro \macro{values} processes its argument for each value
of the current enumeration (and hence each pass has a `current
enumeration value').

It is an error for a macro~$x$ to appear in a context where
there should be a current~$y$, but there is none. For example,
it is an error for the macro \macro{enumName} to appear in a
context where there is no current enumeration. In other words,
the macro \macro{enumName} cannot appear outside of the argument
of \macro{enums}, which is the only macro that sets a current
enumeration.

\subsection{Visitors}
\label{sec:visitors}

The visitor pattern is widely used to implement compilers. It can
be seen as a workaround to a limitation of most object-oriented
languages. A reference~$u$ has the static type~$C_u$ when the
declaration of variable~$u$ is \jmlCode|$C_u$ u;| a reference~$u$
has the dynamic type~$C'_u$ when it points to an object whose
type is~$C'_u$; an object has the type~$C'_u$ when it was created
by the statement \jmlCode|new $C'_u$($\cdots$)|. The method call
\jmlCode|u.m(v)| is resolved based on the dynamic type~$C'_u$
and on the static type~$C_v$. In other words, the code that will
be executed is in a method named~$m$ that takes an argument
of type~$C_v$ (or a supertype of~$C_v$) and is defined in the
class~$C'_u$ (or a supertype of~$C'_u$). There is no way to
do the dispatch based on the dynamic type of two (or more)
references.

However, it is possible to do the dispatch based on the
dynamic types of~$n$ references \emph{one by one}, at the
cost of writing extra code. Say the references $u_1$,~$u_2$,
\dots,~$u_n$ have static types $C_1$,~$C_2$, \dots,~$C_n$ and
dynamic types $C'_1$,~$C'_2$, \dots,~$C'_n$. The initial call
\jmlCode|$u_1$.$m_1$($u_2$,$\ldots$,$u_n$)| will execute a method
\jmlCode|$m_1$($C_2,\ldots,C_n$)| from the class~$C'_1$, because
all the (proper) subclasses of $C_1$ implement such a method. The
body of all these methods will be identical: It will contain the
call \jmlCode|$u_2$.$m_2$(this, $u_3,\ldots,u_n$)|. Each subclass
of~$C_2$, including $C'_2$, is expected to have a \emph{set}
of methods \jmlCode|$m_2$($\mathscr{C}_1,C_3,C_4,\ldots,C_n$)|
for all possible (proper) subclasses~$\mathscr{C}_1$
of~$C_1$. The static type of \jmlCode|this| in the call
\jmlCode|$u_2$.$m_2$(this, $u_3,\ldots,u_n$)| was~$C'_1$, so the
method with $\mathscr{C}_1=C'_1$ will be chosen out of the whole
set. In general, all subclasses of~$C_k$ must implement a set of
methods \jmlCode|$m_k$.($\mathscr{C}_1,\ldots,\mathscr{C}_{k-1},
C_{k+1},\ldots,C_n$)|, for all subclasses~$\mathscr{C}_1$
of~$C_1$, all subclasses~$\mathscr{C}_2$ of~$C_2$, and so on. The
methods~$m_n$ will do the actual work. Let us estimate the number
of methods that only forward calls and were referred to in the
beginning of the paragraph as `extra code'. If the number of
possible types for $u_1$,~$u_2$, \dots,~$u_n$ is, respectively,
$w_1$,~$w_2$, \dots,~$w_n$ then the number of methods $m_k$ is
$\prod_{i\le k} w_i$. There are therefore $\sum_{k<n}\prod_{i\le
k} w_i$ methods whose only purpose is forwarding and $\prod_{i\le
n} w_i$ methods that do something interesting.

As it is traditionally presented, the visitor pattern is the case
$n=2$ with $C_1$ being the root of the AST class hierarchy and
$C_2$ being the root of the visitors class hierarchy. There is
exactly one forwarding method per AST class (and their total
number is $w_1$ with the previous notation). In this guise,
the visitor pattern can be seen as a way of grouping together
the code that achieves one conceptual operation. For example,
pretty printing an AST can be done by implementing a method
\textit{prettyPrint} in each AST class, but can also be done by
putting all the pretty printing code into one visitor called
\textit{PrettyPrinter}. (Note that AstGen makes it hard to use
the former approach, with a specific \textit{prettyPrint} method
in each class.)

FreeBoogie uses the traditional visitor pattern and
the root of the visitors' class hierarchy is the class
\jmlCode|Evaluator<R>|. The root of the Boogie AST
class hierarchy is the class \jmlCode|Ast|. A subclass
of \jmlCode|Evaluator<R>| is like a function of type
$\mathit{Ast}\to R$, in the sense that it associates a
value of type~$R$ (possibly \jmlCode|null|) to an AST
node. For example, the type checker is a subclass of
\jmlCode|Evaluator<Type>|. The base class \jmlCode|Evaluator|
declares one \jmlCode|eval($\mathscr{A}$)| method for each AST
class~$\mathscr{A}$. These are the methods called~$m_2$ in the
previous discussion of the general visitor pattern. These methods
are not only declared, but they are also implemented, so that
subclasses explicitly handle only the relevant types of AST
nodes. For all the other AST node types, the default behavior
implemented in \jmlCode|Evaluator| is to recursively evaluate all
children and to cache the results. Because the \jmlCode|eval|
methods of \jmlCode|Evaluator| are so similar, they are generated
from an AstGen~template.

An important type of evaluator is a transformer: The class
\jmlCode|Transformer| extends \jmlCode|Evaluator<Ast>|. The
main functionality implemented in \jmlCode|Transformer|, path
copying, is illustrated in Figure~\ref{fig:path-copying}. Empty
nodes (\tikz[baseline=-.5ex] \node[fgdraw,circle] {}; and
\tikz[baseline=-.5ex] \node[fgdraw]{};) represent AST nodes
that exist on the heap before a transformer~$T$ acts; filled
nodes (\tikz[baseline=-.5ex] \node[fgdraw,fgfill,circle]{};
and \tikz[baseline=-.5ex] \node[fgdraw,fgfill]{};) represent
AST nodes created by the transformer~$T$. Because the
transformer~$T$ is interested only in rectangle nodes,
it overrides only the \jmlCode|eval| method that takes
rectangles as parameters. That overriden method is responsible
for creating the filled rectangle (\tikz[baseline=-.5ex]
\node[fgdraw,fgfill]{};). All the other filled nodes
(\tikz[baseline=-.5ex] \node[fgdraw,fgfill,circle]{};) are
created by \textit{Transformer}, and need not be of any concern
to the particular transformer~$T$.

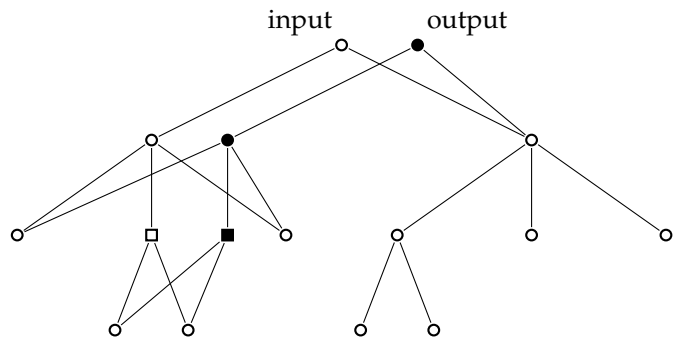
\begin{figure}\centering
\begin{tikzpicture}
  [yscale=1.8,
  level distance=7mm,
  level/.style={sibling distance=5cm/(#1^1.5)}]
\tikzset{tree/.style={fgdraw,circle}}
\node[tree] (A) {}
  child {node[tree] (B) {}
    child {node[tree] (B1) {}}
    child {node[tree,rectangle] (C) {}
      child {node[tree] (C1) {}}
      child {node[tree] (C2) {}}
    }
    child {node[tree] (B3) {}}
  }
  child {node[tree] (A2) {}
    child {node[tree] {}
      child {node[tree] {}}
      child {node[tree] {}}
    }
    child {node[tree] {}}
    child {node[tree] {}}
  };
\path
  (A) +(1cm,0) node[tree,fgfill] (A') {}
  (B) +(1cm,0) node[tree,fgfill] (B') {}
  (C) +(1cm,0) node[tree,fgfill,rectangle] (C') {};
\draw (A') -- (B') -- (C');
\draw (A') -- (A2);
\draw (B') -- (B1); \draw (B') -- (B3);
\draw (C') -- (C1); \draw (C') -- (C2);
\node[above left] at (A) {input};
\node[above right] at (A') {output};
\end{tikzpicture}
\caption{Path copying}\label{fig:path-copying}\index{path copying}
\end{figure}

The input and the output of a transformer usually share a large
number of nodes. Since \textit{Evaluator} caches the information
that various evaluators associate with AST nodes, there is no
need to repeat the computation of that auxiliary information for
the shared parts. For example, most of the type information is
already in the cache of the type checker.

Sometimes a transformer wants to `see' AST nodes of type~$A$
even if it computes no value for them. A typical example is
a pretty printer. In such cases a transformer may override
\jmlCode|eval(A)| and return \textbf{null}. A nicer solution is
to override \jmlCode|see(A)|, whose return type is \textbf{void}.
If both \jmlCode|eval(A)| and \jmlCode|see(A)| are overriden,
then the former will be called by the traversal code in
\textit{Transformer}.

\subsection{Immutability}
\label{sec:design.immutability}

In Java programming, it is unusual to constrain data structures
to be immutable. Since the resulting code may look awkward to
many programmers, there better be some good reasons for this
design decision. In fact, awkward code, such as copying all but
one of the fields in a new object instead of doing a simple
assignment, is only one of the apparent problems.

\bc
\begin{jml}
public class Renamer extends Transformer {
  @Override public Identifier eval(Identifier identifier) {
    if (!identifier.id().equals("u")) return identifier;
    else return Identifier.mk("v");
  }
}
\end{jml}
\ec{Changing all occurrences of variable~$u$ into variable~$v$}
{lst:example-transformer}

Immutability implies path copying, which is a
potential performance problem. Consider the task of
changing all occurrences of the variable~$u$ into
variable~$v$, which is achieved by the transformer in
Figure~\ref{lst:example-transformer}. Suppose an AST with
height~$h$ and $n$~nodes contains exactly one occurrence of
variable~$u$. If the class \textit{Identifier} would be mutable,
one assignment would be enough to achieve the substitution;
since the class \textit{Identifier} is immutable, about $h$ new
nodes must be created and initialized. However, if there are
two occurrences of variable~$u$, they share some ancestors,
meaning that less than about $2h$ new AST nodes must be created
and initialized. Even more, if we take into account the tree
traversal, then both implementations, with a mutable AST and with
an immutable AST, take $\Theta(n)$ time. In other words, there is
no asymptotic slowdown.

A Boogie block contains a \emph{list} of statements (see
Figure~\ref{fig:boogie-absgrm}). Such lists should be immutable,
but there are no immutable lists in the Java API (\fb application
\fb programming \fb interface), only immutable \emph{views} of
lists. Immutable collections can be implemented such that
immutability is enforced statically by the compiler or such
that immutability is enforced by runtime checks. Unfortunately,
the former is incompatible with implementing Java API
interfaces~\cite{javaCollectFaq}. For example, in order to use
the iteration statement
\jmlCode|for (T x : xs)|,
one must implement the interface \textit{Iterable} that
contains the method \textit{remove}. Obviously, calls to
the \textit{remove} method are not prevented statically by
the compiler. FreeBoogie uses the \textit{ImmutableList}
class from the Guava~\cite{guava-libraries}
library, which follows the approach with runtime checks.
(Figure~\ref{fig:astgen-template} shows that the
\textit{ImmutableList} is used whenever the [list] tag appears in
the abstract grammar.)

However, the advantages of immutability outweigh its
disadvantages.

First, immutability enables \textit{Evaluator} to cache the
results of previous computations, because only immutable data
structures can be used as keys in maps. A particular evaluator,
such as the type-checker, need not mention anywhere in its
implementation that caching is used. Yet, if the type-checker
is invoked twice on the same AST fragment, then the second call
will return immediately. This leads to cleaner code also because
AST transformers need not bother with updating the auxiliary
information---recomputing it is cheap. These advantages are
discussed further in Chapter~\ref{ch:ev}.

Second, immutability makes the code easier to understand, because
it frees the programmer from thinking about aliasing of AST nodes.
In Java, any mutation of \jmlCode|u.f| must be done only after
thinking how it will affect code that uses possible aliases of~$u$.
Because the AST is a central data structure in FreeBoogie, there
is a lot of potential aliasing that must be considered whenever
a mutation is done. It is much simpler to forbid mutations altogether.

Still, there are situations when the programmer must think about
aliasing of AST data structures. It is natural to think of an
AST reference as \emph{being} a piece of a Boogie program,
even if, strictly speaking, it only \emph{represents} a piece
of a Boogie program. To maintain this useful illusion the
programmer must ensure that no sharing occurs within one version
of the AST\null. More precisely, there should never be more
than one reference-path between two AST nodes. (There is a
\emph{reference-edge}~$u\to v$ from the object referred by~$u$
to the object referred by~$v$ when \jmlCode|u.f==v| for some
field~$f$.) For example, if the expression~$x+y$ appears multiple
times in a Boogie program, then the corresponding AST also
appears multiple times, instead of being shared. In practice,
this means that the programmer must occasionally clone pieces
of the AST when implementing transformers. (The \textit{clone}
method is implemented in the code template for AST~classes.)

\section{Auxiliary Information}

The package \textit{freeboogie.tc} derives extra information from
a Boogie AST---types, a symbol table, and a flowgraph.

The AST constructed by the parser is type-checked in order to
catch simple mistakes in the input. As a safeguard against bugs in FreeBoogie,
the AST is type-checked after each transformation. A side-effect
of type-checking is that the type of each expression is known.

The symbol table\index{symbol table} helps in navigating the AST\null. It
consists of one-to-many bidirectional maps that link identifier
declarations to places where the identifiers are used. The only such map
that is relevant to core Boogie is the one that links variable
declarations, including those in arguments, to uses of variables. The other
maps, relevant to full Boogie, link procedure declarations to procedure
calls, type declarations to uses of user defined types, function
declarations to uses of (uninterpreted) functions, and type variables to
uses of type variables. (Type variables are similar to generics in Java.)
All these maps are in the class~\textit{SymbolTable}.

Another bidirectional map is built by
\textit{ImplementationChecker}: In full Boogie a \emph{procedure}
may have zero, one, or multiple \emph{implementations}. (In core
Boogie, the whole program is one implementation.)

Finally, it is sometimes convenient to view one implementation as
a flowgraph whose nodes are statements. Such a flowgraph is built
by \textit{FlowGraphMaker}. Formally, a flowgraph is defined as
follows.

\begin{definition}\index{flowgraph}
A \emph{flowgraph} is a directed graph with a distinguished
\emph{initial node} from which all nodes are reachable.
\end{definition}

It seems natural that a flowgraph has an initial node, because
there is usually one entry point to a program. It seems less
natural that all nodes must be reachable, which means that
there is no obviously dead code. The reason for this standard
restriction is rather technical: It simplifies the study of
flowgraph properties. However, it does complicate slightly the
definition of what it means for a flowgraph to correspond to a
core Boogie program. A few terminology conventions will help. In
Chapter~\ref{ch:boogie} we noticed that we can attach counters
to statements because they are in a list. For concreteness, let
us use the counters $1$,~$2$ \dots, $n$, in this order, when the
list of statements has length~$n$. Each counter~$x$ in the range
$0.\,.\>n$ has an associated statement, named statement~$x$.
Statement~$0$ is the sentinel statement \textbf{assume true},
which is prepended for convenience. Label~$x$ is the label that
precedes statement~$x$, if there is one.

\begin{remark}
The sentinel statement~$0$ is not introduced by the FreeBoogie
implementation. It is merely a device that will simplify the
subsequent presentation, especially some proofs.
\end{remark}

The flowgraph of a Boogie program is constructed, conceptually,
in two phases.

\begin{definition}\index{pseudo-flowgraph}
The \emph{pseudo-flowgraph of a core Boogie program} with
$n$~statements has as nodes statement~$1$ up to statement~$n$
and the sentinel statement~$0$. It has an edge~$0\to1$
(from statement~$0$ to statement~$1$) and has edges~$x\to
y$ when (a)~statement~$x$ is \textbf{goto} label~$y$, or
(b)~statement~$x$ is \emph{not} \textbf{goto} and label~$y$ is
the successor of label~$x$.
\end{definition}

\begin{remark}
Compare with
\eqref{eq:assume-assert-ok-opsem}--\eqref{eq:goto-opsem}.
This definition roughly says that there is an edge where the
operational semantics rules would allow a transition \emph{if we
ignore the upper parts of the rules}.
\end{remark}

\begin{definition}
The \emph{flowgraph of a core Boogie program} is a graph that has
node~$0$ as its initial node. Its nodes~$V$ are those nodes of
the pseudo-flowgraph that are reachable from node~$0$ and that
are not \textbf{goto}~statements. It has an edge~$m\to n$ when
there is a path~$m\leadsto n$ in the pseudo-flowgraph that is
disjoint from~$V$, except at endpoints.
\label{def:boogie-flowgraph}
\end{definition}

\begin{example}
Figure~\ref{fig:boogie-flowgraph} shows a core Boogie program,
its pseudo-flowgraph, and its flowgraph. (Label~$k$ is $L_k$.)
Chapter~\ref{ch:reachability} discusses \emph{semantically}
unreachable nodes of the flowgraph, such as node~$4$ in this
example.
\end{example}

\begin{proposition}
The only nodes that have no outgoing edges in a flowgraph of a
core Boogie program are those that correspond to \textbf{return}
statements.
\end{proposition}

\begin{figure}
\centering
\begin{tabular}{ccc}
\begin{boogie}[boxpos=c]
procedure dead(x : int) returns () {
  $L_1$: assume x > 0;
  $L_2$: goto $L_4$, $L_6$;
  $L_3$: assume true;
  $L_4$: assume x < 0;
  $L_5$: return;
  $L_6$: assume true;
  $L_7$: return;
}
\end{boogie}&
\hspace{5mm}
\begin{tikzpicture}[scale=.5,baseline=0.75cm]
  \foreach \n/\x/\y in {0/1/4, 1/1/3, 3/3/4, 4/0/1, 5/0/0, 6/2/1, 7/2/0}
    \oonode (\n) at (\x,\y) [label=left:$\n$] {};
  \node[fgdraw] (2) at (1,2)  [label=left:$2$] {};

  \foreach \i/\j in {0/1, 1/2, 2/4, 2/6, 4/5, 6/7}
    \draw[arr] (\i) -- (\j);
\end{tikzpicture}&
\hspace{5mm}
\begin{tikzpicture}[scale=.5,baseline=0.5cm]
  \foreach \n/\x/\y in {0/1/3, 1/1/2, 4/0/1, 5/0/0, 6/2/1, 7/2/0}
    \oonode (\n) at (\x,\y) [label=left:$\n$] {};
  \foreach \i/\j in {0/1, 1/4, 1/6, 4/5, 6/7}
    \draw[arr] (\i) -- (\j);
\end{tikzpicture}
\\
\end{tabular}
\caption{Flowgraph of a core Boogie program}\label{fig:boogie-flowgraph}
\end{figure}
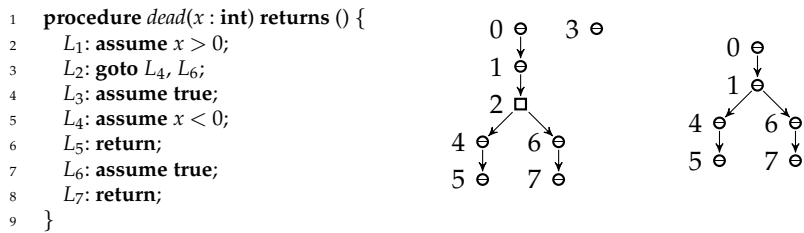

All auxiliary information is available through
\textit{TcInterface}, which is an implementation of the Facade
pattern.

\section{Verification Condition Generation}

The package \textit{freeboogie.vcgen} consists of Boogie
transformers and Boogie to SMT transformers. The facade of this
package is the class \textit{VcGenerator}.

Most Boogie transformers are responsible for small AST
modifications such as desugaring\index{desugaring} an \textbf{if} statement
into \textbf{assume} and \textbf{goto} statements. For speed,
it would be better to cluster many such simple transformers
into one, but the code is easier to maintain if they are kept
separate. A few helper classes are used by multiple Boogie
transformers: \textit{CommandDesugarer} is used as a base class
by transformers that change statements into lists of statements;
\textit{ReadWriteSetFinder} is an evaluator that associates with
each statement two sets---the set of variables that are read and
the set of variables that are written.

Boogie transformations do not update the auxiliary information
while they are building new AST nodes. Instead, at the very end,
they recompute all auxiliary information, and caches ensure that
no computation is repeated. This way, bugs that produce untypable
Boogie programs get caught at run-time. (Type information is
auxiliary information, so type-checking is repeated.)

The Boogie to SMT transformation is done by the
class \textit{WeakestPrecondition} or by the class
\textit{StrongestPostcondition}, depending on the command line
options. The theory behind these two classes is presented in
Chapter~\ref{ch:spwp}.

\section{The Prover Backend}
\label{sec:design.backend}

The package \textit{freeboogie.backend} contains (1)~SMT data
structures and (2)~code to communicate with provers. The design
is inspired by the sorted multi-prover backend in \escjava.

\subsection{Data Structures and Sort-Checking}
\label{sec:design.ds}

The main data structure is a rooted ordered tree whose nodes are labeled by
strings. Each node has a \emph{sort}\index{sort}, and there are
sort-checking rules, which say what combinations of sorts and labels are
valid. In effect, sorts are types---the only reason a different name is
used is to distinguish SMT sorts from Boogie types. In \escjava it is
impossible to construct a tree that has sort errors: Programs that try to
construct invalid terms fail Java type-checking. Such a strong static
guarantee is appealing, but increases the backend size significantly.  For
example, instead of a single factory method with the signature
  \jmlCode|SmtTree mk(String label, ImmutableList<SmtTree> children)|
there is a plethora of methods with various argument and return
types, such as the method with the signature
  \jmlCode|SmtFormula mkEq(SmtTerm left, SmtTerm right)|,
where both classes \textit{SmtFormula} and \textit{SmtTerm} are subclasses
of \textit{SmtTree}. Because of the size, the backend is hard to
adapt to changes.

FreeBoogie opts for a small backend, so that it easy to understand and
maintain. If an ill-sorted term is built, then most SMT solvers complain
and the problem is found. To help in finding the source of the
problem, the backend has built-in dynamic checks that should point to the
offending code, before the SMT term is shipped to the solver.

\begin{remark}
This is an instance of choosing dynamic checks over static checks, because
the latter involve too much work. The code is still organized in a way that
should allow static verification. It is the encoding in Java types that was
judged too complicated for its benefits.
\end{remark}

Before calling \jmlCode|mk(label, children)| the label must have
been defined. For example, after the call
  \jmlCode|def("eq", new Sort[]{Sort.TERM, Sort.TERM}, Sort.FORMULA)|
it is possible to call
  \jmlCode|mk("eq", children)|.
This second call will check (using Java assertions) that there are
two children and both are terms, and will mark the constructed
SMT tree as being a formula. All defined labels are grouped in
stack frames, such that the call \jmlCode|popDef()| discards
all definitions done after the corresponding call \jmlCode|pushDef()|.
Such grouping is useful because some labels refer to constructs
built into SMT solvers and other labels refer to uninterpreted
functions that are defined by the Boogie program. When
FreeBoogie moves from one input file to another it forgets about
labels corresponding to functions while not forgetting about
labels corresponding to solver built-ins by using the stack
mechanism.

The methods \textit{def}, \textit{mk}, \textit{pushDef}, and
\textit{popDef} are all defined in the class \textit{TreeBuilder}.
For convenience, the functions \textit{def} and \textit{mk} are
overloaded.

Let us first look at the method~\textit{mk}. It comes in three
varieties:
\begin{align}
&\text{\jmlCode|mk("and", children)|}\\
&\text{\jmlCode|mk("eq", $t_1$, $t_2$)|}\\
&\text{\jmlCode|mk("literal_int", new FbInteger(3))|}
\end{align}
The first form takes a list of children as the second
argument. When the number of children is fixed, as is the
case for the label \texttt{eq}, it is convenient to hide the
building of the list behind a helper overload. The second form
can be used when the number of children is one, two, or three.
The third form is special. Strictly speaking, the constants
$1$,~$2$, $3$,~$\ldots$ are distinct functions
that take no argument. This suggests that they should each
be defined separately, which is clearly a very bad idea from
the point of view of performance. So, instead of defining
labels \texttt{1},~\texttt{2}, \texttt{3},~\dots, we define
the meta-label \texttt{literal\_int}. A \emph{meta-label} has
an associated Java type (in this case \textit{FbInteger}) and
it is equivalent to multiple labels, one for each value of the
associated Java type. In other words, the meta-label 
\texttt{literal\_int} and the value \jmlCode|new FbInteger(3)|
determine the label, and there is no child.

Now let us look at the method~\textit{def}. It comes in three
varieties:
\begin{align}
&\text{\jmlCode|def("and", Sort.FORMULA, Sort.FORMULA)|}\\
&\text{\jmlCode|def("eq", new Sort[]\{Sort.TERM, Sort.TERM\}, Sort.FORMULA)|}\\
&\text{\jmlCode|def("literal_int", FbInteger.class, Sort.INT)|}
\end{align}
The order of the arguments is: label, sort of arguments, sort
of result. The example for the first form says that tree nodes
labeled with \texttt{and} may have any number of children, all of which
must be formulas, and the tree itself is a formula. The example
for the second form says that tree nodes labeled \texttt{eq} have
two children, the first one is a term, the second one is a term,
and the tree itself is a formula. The example for the third form
says that the meta-label \texttt{literal\_int} together with a
value of type \textit{FbInteger} constitutes a label, and trees
labeled in this way are integers.

(As a side note, \textit{FbInteger} is used because Boogie allows
arbitrarily large integers and has bit vector operations. No class
in the standard Java library supports both.)

The sorts include \textit{FORMULA} and $\mathit{TERM}:>\mathit{INT},
\mathit{BOOL}$. In some places, such as the first argument of a quantifier,
only variables are allowed. Those require the a sort of the form
\textit{VARx}, which is a subsort of some sort~$x$. Other sorts are easy to
add.

Any SMT trees $s$ and~$t$ have the property that \jmlCode|s.equals(t)|
implies \jmlCode|s==t|. This is implemented by maintaining a global set of
all SMT trees that were created, a technique sometimes known by the name
\emph{hash-consing}~\cite{filliatre2006hash}.

\subsection{The Translation of Boogie Expressions}

The methods \textit{mk} provide one way of building trees; the
method~\textit{of} provides another way of building trees.
For example, the class \textit{StrongestPostcondition} uses
the methods~\textit{mk} to connect formulas (using the labels
\texttt{and}, \texttt{implies}) and uses the method~\textit{of}
to obtain the formulas corresponding to individual assertions and
assumptions.

The method~\textit{of} converts from Boogie expressions
to SMT formulas. The actual work is done in two
classes---\textit{TermOfExpr} and \textit{FormulaOfExpr}. The
translation is almost one-to-one. Each Boogie operator has a
corresponding interpreted symbol in the SMT language; each
function declared in a (full) Boogie program behaves similarly
to an uninterpreted function symbol in the SMT language.
There is, however, an important deviation from the one-to-one
correspondence. As we have seen, the SMT language distinguishes
between terms and formulas. Roughly speaking these correspond,
respectively, to non-boolean Boogie expression and to boolean
Boogie expressions. For example, in the SMT language the
operands of logical-and must be formulas and in Boogie the
operands of logical-and must be booleans. On the other hand, an
SMT uninterpreted symbol is always a term, while in Boogie a
function may return a boolean. Also, in SMT the arguments of an
uninterpreted symbol must be terms, while in Boogie a function
might take booleans as arguments. Because of these reasons, a
one-to-one translation may produce ill-formed SMT trees, which
fail sort-checking.

In SMT the constants \textbf{true} and \textbf{false} are a
formulas. If we introduce two corresponding uninterpreted terms,
\textit{trueTerm} and \textit{falseTerm}, we can then try to fix
the ill-formed SMT tree using the following two rules.
\begin{enumerate}
\item If a term~$\tau$ appears where a formula is expected then we 
  replace the term by \smtCode|(= $\tau$ trueTerm)|. This compares
  for equality $\tau$ and \textit{trueTerm}.
\item If a formula~$\varphi$ appears where a term is expected
  then we replace the formula by 
  \smtCode|(ite $\varphi$ trueTerm falseTerm)|.
  This expression evaluates to \textit{trueTerm} for all 
  models in which $\varphi$~evaluates to~$\tru$;
  it evaluates to \textit{falseTerm} for all models
  in which $\varphi$~evaluates to~$\fls$.
\end{enumerate}
This, however, is not exactly what FreeBoogie does.
Simplify is an old but still competitive prover whose language is
similar to the SMT language. One difference is that in Simplify
a term never contains a formula. In particular, there is no
\textbf{ite}, so the rule~2 from above cannot be used.

To clarify these ideas, let us turn to an example.
\begin{boogie}
function f(x : bool) returns (bool);
axiom (forall x : bool :: x == f(x));
procedure p() returns () { assert f(true) != f(false); }
\end{boogie}
The assertion should hold. 

The following is what FreeBoogie sends to Simplify.
\begin{smt}
(BG_PUSH (NEQ trueTerm falseTerm))
(BG_PUSH (FORALL (xTerm) (EQ xTerm (f xTerm))))
(NOT (IFF (EQ trueTerm (f trueTerm) (EQ trueTerm (f falseTerm)))))
\end{smt}

In Simplify's language, interpreted symbols are written
in \textsc{capital} letters and their names are usually
self-explanatory. The command \textbf{BG\_PUSH} communicates a
hypothesis to the prover. Line~3 is a query. The Boogie constants
\textbf{true} and \textbf{false} that appear as arguments
of the function~$f$ were translated into terms directly,
without an intermediate application of rule~2. The comparison
between booleans~\boogieCode|$\bullet$ != $\bullet$|, which appears
in the assertion, was translated to \smtCode|(NOT (IFF $\bullet$ $\bullet$))|. 
Because \textbf{IFF} expects formulas as arguments and because
\smtCode|(f $\ldots$)| is a term, rule~1 was applied, which is
why \textbf{EQ} appears in the query.

Assuming hypothesis~2, the query is equivalent with hypothesis~1.
\begin{align}
&\phantom{\;=\;}
  \text{\smtCode|(NOT (IFF (EQ trueTerm (f trueTerm) (EQ trueTerm (f falseTerm)))))|}\\
&=\text{\smtCode|(NOT (IFF TRUE (EQ trueTerm (f falseTerm)))))|}
  \label{eq:term-tricks1}\\
&=\text{\smtCode|(NEQ trueTerm (f falseTerm))|}\\
&=\text{\smtCode|(NEQ trueTerm falseTerm)|}
  \label{eq:term-tricks2}
\end{align}
Equation~\eqref{eq:term-tricks1} follows by setting
$\mathit{xTerm}:=\mathit{trueTerm}$ in hypothesis~2;
equation~\eqref{eq:term-tricks2}, which is the same as hypothesis~1, follows
by setting $\mathit{xTerm}:=\mathit{falseTerm}$ in hypothesis~2.  Without
the two hypothesis, the query is not valid. For example, it could be that
$\mathit{trueTerm}=\mathit{falseTerm}$. In general, some extra hypotheses
are required.  If not enough hypotheses are introduced, then FreeBoogie is
incomplete, but should still be sound. Whenever \textit{TermOfExpr} or
\textit{FormulaOfExpr} produce an SMT tree, they may attach to it extra
hypotheses.  These are SMT trees themselves, and may have further
hypotheses attached. All hypotheses are collected and sent to the prover
before the query.

\subsection{Talking to the Prover}

The class \textit{Prover} defines the interface that is used
by the package \textit{freeboogie.vcgen} to talk to the
prover. It is a thin interface, consisting of the methods
\textit{assume}, \textit{retract}, \textit{push}, \textit{pop},
and \textit{isValid}.

The real prover does not have to have the notion of an
assumption (also known as hypothesis), but a class that extends
\textit{Prover} should take advantage of all facilities of a real
prover. For example, if Simplify is used as a prover, then a
sequence of calls \jmlCode|assume(h)|, \jmlCode|isValid($q_1$)|,
\jmlCode|isValid($q_2$)| may result in one of the following two
strings being sent to the prover:
\begin{align}
&(\mathbf{IMPLIES}\;h\;q_1)\;(\mathbf{IMPLIES}\;h\;q_2)\\
&(\mathbf{BG\_PUSH}\;h)\;q_1\;q_2
\end{align}
Both are OK, but the second is better, if only because
$h$ is communicated once.

Similarly, a class that extends \textit{Prover} may choose to
treat certain SMT tree labels specially to take advantage of
other facilities of the real prover.

\section{Other Generated Code}

The Boogie parser resides in the package
\textit{freeboogie.parser} and is generated by ANTLR
(\textbf{an}other \fb tool for \fb language \fb
recognition); the command line parser resides in the package
\textit{freeboogie.cli} and is generated by CLOPS (\fb command
\fb line~\textbf{op}tion\textbf{s}).

\section{Related Work}

The main goal of the previous sections is to anchor the
subsequent theoretical chapters in a concrete program,
FreeBoogie. A secondary goal is to serve as a guide to the code
and to make explicit the early design choices. This section is
for the reader who wants to understand the design in detail, but
feels that the previous sections are too shallow.

\subsection{Related Tools}

FreeBoogie is a Java clone of the Boogie
tool~\cite{barnett2005boogie} from Microsoft Research. The
internals differ but the input and the output interfaces
are the same. The input is a program written in the Boogie
language~\cite{leino2008boogie,leino2010boogie}; the output is a
formula written in the SMT language, or in a similar language.

The Why language~\cite{filliatre2007why} is another intermediate
verification language. Its syntax is similar to that of OCaml.  The
accompanying tool targets SMT solvers, but also proof assistants such as
Coq and Isabelle.

The Boogie tool is part of the \specsharp program
verifier~\cite{barnett2005spec}, which verifies programs written in a
superset of \csharp; the Why tool is part of \framac~\cite{framac}, which
verifies C programs.

\subsection{Design, Pipeline, and Correctness}

Because the input and the output are programs written in
well-defined languages, FreeBoogie is a compiler. The standard
text on compiler design is the Dragon Book~\cite{aho2007}. The
pipeline architecture is common to most compilers.

One of the oldest problems studied by the formal methods
community is the correctness of compilers. The early approaches
were focused on proving that a compiler is correct (see, for
example, \cite{moore1989cc}). The idea is to show that the
semantics of the input program are in a certain relationship
with the semantics of the output program. The relationship
usually ensures that the two programs have the same observable
behavior. Although there is recent research in the same
vein~\cite{leroy2009}, there are also attempts to go around the
problem and avoid the full verification of the compiler. One of
these alternative approaches is \emph{translation validation},
proposed in 1998 by Amir Pnueli and others~\cite{pnueli1998tv}.
The idea is to check the result of each particular compilation,
instead of proving that the compiler works for all possible
inputs. To make it even easier, instead of checking the whole
compilation, the idea can be applied to each compilation phase.
The technique was used to find bugs in GCC (the \fb GNU \fb
C \fb compiler)~\cite{necula2000tv}. A related technique,
\emph{credible compilation}~\cite{rinard1999credible}, consists
of modifying the compiler to output a proof for each compilation.
It is then possible to check the proof of equivalence using a
small trusted proof checker. In 2004, Benton~\cite{benton2004}
introduced a way of doing equivalence proofs that seems to
fit well with the Boogie language.

The problem of correctness of a program verifier has certain peculiarities. The
`observable behavior' of a Boogie program, which is not executable, is whether
it is correct or not, according to Definition~\ref{def:correctness} (on
page~\pageref{def:correctness}). It follows that two Boogie programs are
equivalent if they are both correct or both incorrect. As with normal
compilers, a transformation of a Boogie program is correct if the output is
equivalent to the input. Section~\ref{sec:pipeline} says that some
transformations in FreeBoogie are incorrect and it identifies soundness and
completeness as weaker guarantees.

The Dragon Book analyzes the problem of correctness from an
algorithmic point of view, but it does not address the problem
of correctness of an implementation; the Dragon Book discusses
the overall pipeline architecture of a compiler but does not
go into details of code organization. The Design Patterns
book~\cite{gamma1995} partly covers this area. For example, the
visitor pattern, the facade pattern, and the factory method
pattern, which were used to explain FreeBoogie's design, are all
presented in that book. Of these three patterns, only the visitor
pattern was briefly analyzed in Section~\ref{sec:visitors}.

``[A] \emph{facade} provide[s] a unified interface to a set
of interfaces in a subsystem. Facade defines a higher-level
interface that makes the subsystem easier to use.'' In this
dissertation a more specific meaning is used---a class (or
interface) that provides almost all the services of a package.
The non-facade classes are still visible from outside, in case
they are needed, but their use is discouraged.

A \emph{factory method} creates an object and returns it. The Design Patterns
book emphasizes that the method call may be dynamically dispatched to
subclasses. While this aspect is used by the backend, the factory methods for
Boogie AST are static.  Such methods lead to code that is easier to read
because of their descriptive names and because type inference for generics
works better for methods than for constructors in Java~6. Such methods also
make possible the implementation of more sophisticated creation patterns, such
as singleton and~hash-consing.

\subsection{The Expression Problem}

The visitor pattern is a way of achieving multiple dynamic
dispatch, but it is also a partial solution to the
expression problem. The term was coined by Philip Wadler in an
email~\cite{wadler1998ep} from 1998: ``The \emph{Expression
Problem} is a new name for an old problem. The goal is to define
a datatype by cases, where one can add new cases to the datatype
and new functions over the datatype, without recompiling existing
code, and while retaining type safety.'' Since most languages
have modular compilation, the restriction on recompilation
usually means that one is allowed to add new modules but not to
modify existing ones. Wadler continues: ``One can think of cases
as rows and functions as columns in a table. In a functional
language the rows are fixed [$\ldots$] but it is easy to add new
columns. In an object-oriented language, the columns are fixed
[$\ldots$] but it is easy to add new rows.''

The visitor pattern is a way of easily adding new columns in
an object-oriented language. In compilers, functionality tends
to evolve more than the AST. This is one reason why functional
languages are well-suited for implementing compilers and it is
why most compilers written in object oriented languages use the
visitor pattern. However, once the visitor pattern is used, it
becomes hard to modify the AST.

FreeBoogie is written in an object oriented language and uses
the visitor pattern. To add new functionality, one implements a
new evaluator or a new transformer, both being visitors. The
code implementing the new functionality goes in a new class so,
in Wadler's terminology, it is easy to add columns. If a new
type of node must be added to the AST then one new class must
be added to the AST data structures. So far, the existing code
was not touched and no recompilation is necessary. The next
step, however, is to add a methods to \textit{Evaluator}, the
root of the hierarchy of visitors. In Wadler's terminology, it
is hard to add rows. What is problematic in practice is not the
recompilation time, but the fact that the AST data structures
and the base classes for visitors must be kept in sync. Because
the AST data structures, the class \textit{Evaluator}, and the
class \textit{Transformer} are generated by AstGen from the same
description of the abstract grammar they are \emph{automatically}
in sync. In other words, AstGen can be seen as a patch that
brings the visitor pattern closer to the ideal solution for the
expression problem. (But recompilation is still necessary when
new AST nodes are added.)

Another approach to the expression problem is to
modify the programming language to support multiple
dispatch~\cite{chambers1994mm, clifton2006}.

\subsection{Provers}

Many languages understood by theorem provers (such as the SMT
language, the Simplify language, and the PVS prover language)
are based on S-expressions. Like XML~\cite{bray2006xml} (the
e\textbf{x}tensible \fb markup \fb language), S-expressions
provide a syntax that is easy to parse. Briefly, they are a fully
parenthesized preorder print of an AST\null. Where XML says
\texttt{<tag>$\,\cdots$</tag>}, S-expressions say \texttt{(tag
\dots)}. (There is no equivalent for XML attributes.) One year
before writing about the semantics of programs, John McCarthy
presented~\cite{mccarthy1960} in 1960 how S-expressions are used
in the programming system LISP (\textbf{lis}t \fb processing).
``S stands for symbolic.''

Internally, FreeBoogie uses a tree of strings to represent symbolic
expressions. To save memory and to speed up otherwise expensive structural
comparisons FreeBoogie uses hash-consing\index{hash-consing}.  The idea was
described by Andrei Ershov~\cite{ershov1958} in 1957, and one year later an
English translation was available.  Before creating a tree node, a global
hash table is searched to see if a structurally similar node already
exists. The creation of nodes takes a constant amount of time on average,
structural comparison of tree nodes is done then by a simple pointer
comparison, and much memory is saved because duplication of information is
avoided. A simple implementation of hash-consing is easy, but offering
proper library support is not trivial---a solution~\cite{filliatre2006hash}
for OCaml was published in~2006.

The tree of strings was chosen as the main data structure of
\textit{freeboogie.backend} because it is a natural representation of the
SMT language. The SMT community~\cite{barrett2010lib} produced a language,
a command language, theories, benchmarks; it also organizes an annual
competition between SMT solvers. The language defines, on top of
S-expressions, the meaning of about twenty keywords and about a dozen
`attributes'.  The theories are the `standard library' of the SMT language.
Each defines the meaning of a set of symbols. For example, the theory Ints,
defines the semantics of the symbols \texttt{0}, \texttt{1},
\texttt{\char`\~}, \texttt{-}, \texttt{+}, \texttt{*}, \texttt{<=},
\texttt{<}, \texttt{>=}, and \texttt{>}.  The benchmarks include
hand-crafted queries, random queries, queries produced from hardware
verification tasks, and queries produced from software verification tasks.
A random sample of the benchmarks is used in the annual SMT competition.
SMT solvers under active development with good support for quantifiers
include Z3~\cite{moura2008z3}, and CVC3~\cite{barrett2007}.

Since the SMT language is supported by several provers that compete each
year on software verification tasks, it makes a natural target for
FreeBoogie.  A better way to interact with provers is through an API, which
is unfortunately not standardized.  Z3 has a C API, a .NET API, and an
OCaml API; CVC3 has a C API and a C++ API\null, and is open source.

\subsection{Code Generators}

The Boogie AST is generated by AstGen; the Boogie parser is
generated by ANTLR; the command-line parser is generated by
CLOPS.

ANTLR~\cite{parr1995} is probably the most widely used parser
generator for Java. The grammar is $\mathrm{LL}(k)$, but
backtracking can be optionally activated. FreeBoogie sticks to
$\mathrm{LL}(k)$ so that the generated parser is as fast as
possible. ANTLR can construct AST trees. This facility is not
used. It is not easy to convince ANTLR to generate different
data structures for the AST\null. For example, it is not easy to
convince ANTLR to generate immutable ASTs. Also, it is not easy
to convince ANTLR to generate the root of the visitors' hierarchy
in sync with the AST data~structures.

CLOPS~\cite{janota2009clops} is a Java parser generator
specialized for command lines. FreeBoogie is one of its first
users.

\chapter{Optimal Passive Form}
\label{ch:passive}

\chquote{Active Evil is better than Passive Good.}{William Blake}

\noindent Flanagan and Saxe~\cite{flanagan2001passive} explain
how passivation avoids the exponential explosion of VCs. This
chapter (1)~gives a precise definition for passivation and
(2)~proceeds to examine its algorithmic difficulty.

\section{Background}
\label{sec:passive.background}

\subsection{VC Generation with Assignments}
\label{sec:boogiesem-wpsp}

The best way to understand why passivation is such a good idea,
especially from the point of view of performance, is to compare
it to a VC generation method that skips passivation. This means
that we have to look ahead at the next stage of the pipeline and
see how it would have to work if passivation is not performed.
For this reason, the methods presented in this \emph{subsection}
will only be fully justified in the next chapter.

\subsubsection{Weakest Precondition for Core Boogie}\index{weakest precondition}
\label{sec:pass-wpcb}

The weakest precondition transformer for the
three statements that can appear in a flowgraph
(see Definition~\ref{def:boogie-flowgraph} on
page~\pageref{def:boogie-flowgraph}) is
\begin{align}
\wpp{(\text{\textbf{assume}~$p$})}{b} &\,\equiv\, p\limp b\\
\wpp{(\text{\textbf{assert}~$p$})}{b} &\,\equiv\, p\land b  
  \label{eq:wp-assert}\\
\wpp{(u := e)}{b} &\,\equiv\, (u \gets e)\>b  \label{eq:wp-assign}
\end{align}
Here $b$ and $p$ are predicates as defined in
Chapter~\ref{ch:boogie}, where the predicate transformer
$(u \gets e)$ is also introduced. Symbolically, the effect of $(u
\gets e)$ is to substitute the expression~$e$ for each occurrence
of the variable~$u$.

Each node~$x$ of the flowgraph gets a precondition~$a_x$ and a
postcondition~$b_x$ according to the rules
\begin{align}
b_x &\equiv \bigwedge_{x\to y} a_y \label{eq:post-from-pre} \\
a_x &\equiv \wpp{x}{b_x}
\end{align}
The query sent to the prover is the predicate~$a_0$.

Figure~\ref{fig:diamonds-seq} shows a program that leads to a
very big query. Let us denote by $q_k(u)$ the postcondition
of the statement \textbf{assert}~$p_k(u)$. Then $q_n(u)\equiv\tru$
and $q_{n-1}(u)\equiv p_n(e_n(u))\land p_n(f_n(u))$. In general,
all $q_k$s will be conjunctions of function compositions,
so let us use a simplified notation just here, while we
evaluate the memory usage of the weakest precondition method.
\begin{equation}
q_{n-1}=\{p_ne_n,p_nf_n\}
\end{equation}
Here $pe$ stands for the function composition $p\circ
e$, and $\{p, q, r\}$ stands for the conjunction $p\land q\land
r$. In general, the rule is
\begin{equation}
q_{k-1}=\{p_ke_k,p_kf_k,q_ke_k,q_kf_k\}. \label{eq:q-rule}
\end{equation}
For example,
\begin{equation}
q_{n-2}=\{p_{n-1}e_{n-1},p_{n-1}f_{n-1},p_ne_ne_{n-1},p_nf_ne_{n-1},
  p_ne_nf_{n-1},p_nf_nf_{n-1}\}
\end{equation}

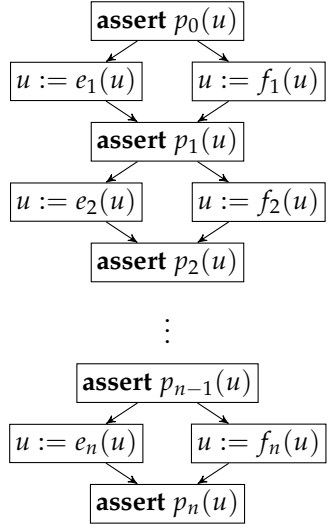
\begin{figure}\centering
\begin{tikzpicture}[scale=0.8]
  \newcount\vertpos\vertpos=9
  \def\n#1#2(#3,#4)#5{\node[draw,rectangle,inner sep=2pt] (#1#2) at (#3,#4) {#5};}
  \def\lr#1{\advance\vertpos by-1
    \n e#1(0,\the\vertpos){$u:=e_{#1}(u)$}
    \n f#1(3,\the\vertpos){$u:=f_{#1}(u)$}}
  \def\m#1{\advance\vertpos by-1 \n p#1(1.5,\the\vertpos)%
    {\textbf{assert}~$p_{#1}(u)$}}
  
  \m0\lr1\m1\lr2\m2\advance\vertpos by-1
  \m{n-1}\lr{n}\m{n}
  \node at (1.5,3) {$\vdots$};

  \draw[arr] (p0) -- (e1); \draw[arr] (p0) -- (f1);
  \foreach \x in {1, 2, n} {
    \draw[arr] (e\x) -- (p\x); \draw[arr] (f\x) -- (p\x);
  }
  \draw[arr] (p1) -- (e2); \draw[arr] (p1) -- (f2);
  \draw[arr] (pn-1) -- (en); \draw[arr] (pn-1) -- (fn);
\end{tikzpicture}
\caption{Exploding diamonds}
\label{fig:diamonds-seq}
\end{figure}

All these expressions are represented using SMT trees, which
are really dags because of hash-consing. In evaluating the
memory space necessary to represent a predicate we must take
into account subexpressions that appear more than once. One
such subexpression is~$u$, which appears quite often. Is there
any other sharing in, say, the predicate $q_{n-2}$? Well, it
contains $p_{n-1}e_{n-1}$ and $p_ne_ne_{n-1}$, which both end in
$e_{n-1}$, so $q_{n-2}$ contains $e_{n-1}$ at least twice.

In general, if we regard $e_k$,~$f_k$, and~$p_k$ as (distinct)
letters from some alphabet and $q_k$ as a set of strings, then
we are interested in the size of the trie that represents the
reversed strings of~$q_k$. Figure~\ref{fig:q-rule-trie} shows
the trie for~$q_{k-1}$. (This trie is, roughly, an upside-down
depiction of the SMT tree, except that SMT tree nodes correspond
to edges in the trie.) If we denote by $|q_k|$ the number of
edges in the trie necessary to represent~$q_k$, then
\begin{align}
|q_{n-1}| &= 4, \\
|q_{k-1}| &= 2 |q_k| + 4.
\end{align}
So $|q_0|=2^{n+2}-4$. Of course, the SMT tree for the whole query
will also contain $p_0$,~$u$, and a node for~$\land$. In any
case, the size of the query is~$\Theta(2^n)$. We will see later
that if passivation is applied first, then the size of the query
is~$\Theta(n)$. This is one reason why we passivate programs.

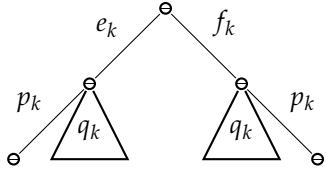
\begin{figure}\centering
\begin{tikzpicture}[auto]
  \foreach \x/\y in {0/0,1/1,2/2,3/1,4/0}
    \oonode (\x) at (\x,\y) {};
  \draw (0) to node {$p_k$} (1);
  \draw (1) to node {$e_k$} (2);
  \draw (2) to node {$f_k$} (3);
  \draw (3) to node {$p_k$} (4);
  \draw[thick] (1)--(0.5,0)--(1.5,0)--(1); \node at (1,0) [label=above:$q_k$] {};
  \draw[thick] (3)--(3.5,0)--(2.5,0)--(3); \node at (3,0) [label=above:$q_k$] {};
\end{tikzpicture}
\caption{The trie constructed by \eqref{eq:q-rule}}
\label{fig:q-rule-trie}
\end{figure}

Flanagan and Saxe~\cite[Section~4]{flanagan2001passive} identify sequences
of assignments like $x:=x+x$ as another source of exponential explosion.
Such programs are problematic only in the absence of hash-consing.

\subsubsection{Strongest Postcondition of Core Boogie}
\index{strongest postcondition}

The strongest postcondition is computed as follows.
\begin{align}
\spp{(\text{\textbf{assert}/\!\textbf{assume}~$p$})}{a} &\equiv a\land p \\
\spp{(u:=e)}{a} &\equiv \exists v, \bigl((u\gets v)\> a\bigr) \land 
  \bigl(u=(u\gets v)\> e\bigr)
\label{eq:sp-of-assgn}
\end{align}
Again, each node~$y$ gets a precondition~$a_y$ and a postcondition~$b_y$.
\begin{align}
a_y &\equiv
  \begin{cases}
  \tru & \text{if the node~$y$ is initial}\\
  \bigvee_{x\to y} b_x & \text{otherwise}
  \end{cases} \\
b_y &\equiv \spp{y}{a_y}
\end{align}

Given a predicate~$f$ and a sub-predicate~$e$ such that $f\equiv p\;e$, we
say that $e$~occurs in a positive position in~$f$ when $|p\;\fls|\limp
|p\;\tru|$; similarly, $e$ occurs in a negative position in~$f$ when
$|p\;\tru|\limp |p\;\fls|$.  For validity queries, SMT solvers cope well
with $\forall$ in positive positions and with $\exists$ in negative
positions, but not the other way.  The quantifier in~\eqref{eq:sp-of-assgn}
will be on the left hand side of an implication when it is sent to the
prover (see Chapter~\ref{ch:spwp}), so it is not a problem.

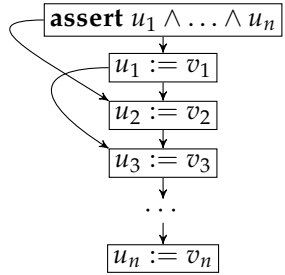
\begin{figure}\centering
\begin{tikzpicture}[y=-1.5\baselineskip,inner sep=2pt]
  \tikzstyle{r}=[draw,rectangle,inner sep=2pt];
  \node[r] (0) at (0,0) {$\mathbf{assert}\;u_1\land\ldots\land u_n$};
  \node[r] (1) at (0,1) {$u_1:=v_1$};
  \node[r] (2) at (0,2) {$u_2:=v_2$};
  \node[r] (3) at (0,3) {$u_3:=v_3$};
  \node[r] (n) at (0,5) {$u_n:=v_n$} ;
  \node (d) at (0,4) {$\cdots$};
  \foreach \i/\j in {0/2,1/3}
    \draw[arr] (\i.west) .. controls +(-1,0) and +(-1,-1) .. (\j.north west);
  \foreach \i/\j in {0/1, 1/2, 2/3, 3/d, d/n}
    \draw[arr] (\i) to (\j);
\end{tikzpicture}
\caption{Exploding predicates}\label{fig:sp-explosion}
\end{figure}

As it turns out, in the presence of hash-consing there is no exponential
explosion on the example in Figure~\ref{fig:diamonds-seq}. There is,
however, an exponential explosion for programs like the one in
Figure~\ref{fig:sp-explosion}. In this example, a control flow path goes
up--down through all statements, and for each assignment there is also an
edge that skips it. The postcondition of the last assignment includes $2^n$
conjunctions of the form $w_1\land\ldots\land w_n$, where
$w_k\in\{u_k,v_k\}$. Even with sharing, these conjunctions require
$\Omega(2^n)$ space.

\subsection{A Few Concepts from Computational Complexity}
\label{sec:cc-primer}

One of the main results of this chapter is that a certain
variation of passivation is NP-complete. This section is a brief
reminder of what NP-complete means.

Very roughly, The RAM (\fb random \fb access \fb memory) model of
computation consists of a processor and a memory whose words each
have a fixed number of bits and are indexed by integers that fit
in a word. The processor can execute binary operations between
two memory words and store the result in a (not necessarily
distinct) third word in one time step. Operations include modular
arithmetic and bitwise logic. The processor can also (1)~request
one word of input and (2)~produce one word of output.

\begin{definition}\index{problem!computational}
A \emph{computational problem} is a relation $P\subseteq
L\times L$, where $L$ is the set of finite sequences of
words~$w_1$,~$w_2$, \dots,~$w_n$, for all~$n$.
\end{definition}

We say that~$y$ is a \emph{valid solution}\index{valid solution} (or
\emph{valid output}) for the \emph{problem instance} (or \emph{input})~$x$
when~$xPy$. An \emph{algorithm} is the finite set of instructions that the
processor is hardwired to carry out. We say that an algorithm is a
\emph{solution} to a problem when it produces valid output for every input.
(For simplicity, let us assume that there is a valid output for every
input.) We write $|w|$ for the length~$n$ of a sequence $w_1$,~$w_2$,
\dots,~$w_n$.

\begin{definition}\index{complexity}
The \emph{time complexity} of an algorithm is a function
$t:L\to\Z_+$ that associates with each input the number of
steps executed by the algorithm before producing the output.
The \emph{space complexity} of an algorithm is a function
$s:L\to\Z_+$ that associates with each input the number of memory
words touched by the algorithm before producing the output.
\end{definition}

The \emph{worst case}\index{worst case} complexity~$T(n)$ is the maximum
complexity~$t(x)$ over inputs~$x$ of size~$n$. Similarly, one
can define an \emph{average complexity}, if probabilities are
associated to inputs.

\begin{definition}
A computational problem is in the complexity class~P when it
has a solution whose worst case time complexity has a polynomial
upper~bound.
\end{definition}

\begin{definition}\index{problem!decision}
A \emph{decision problem} is a set $D\subseteq L$.
\end{definition}

Equivalently, a decision problem is a computational problem
whose answer is 0~or~1, that is, a problem with at most two valid
outputs.

\begin{definition}\label{def:np}
A decision problem~$D$ is in the complexity class~NP when
there exists a subset $R\subseteq L\times L$ with the properties:
\begin{enumerate}
\item There is a polynomial~$p$ such that $|y|\le p(|x|)$
  whenever $xRy$.
\item For all $x\in D$ there is an $y$ such that $xRy$. There
  is no such~$y$ for~$x\notin D$.
\item The problem of deciding whether $xRy$ is in~P.
\end{enumerate}
\end{definition}

The element~$y$ is a \emph{proof} that $x\in
D$; a solution to problem~3 in the above definition is a
\emph{verification procedure}.

\begin{remark}
Definition~\ref{def:np} says that a problem is in NP when it may be solved in
polynomial time by a nondeterministic machine. If a problem is solved in
polynomial time by a nondeterministic machine, then we may record the choices
made as the proof~$y$. In the other direction, a nondeterministic machine may
try all possible proofs~$y$ in polynomial time, because their length is bounded
by a~polynomial.
\end{remark}

This definition applies to decision problems, but we will
also use it with respect to optimization problems.

\begin{definition}\index{problem!optimization}
Given a partial cost function $c:L\times L\rightharpoonup \mathbb{R}_+$, an
\emph{optimization problem} asks for an algorithm that for each input~$x$
produces an output~$y$, such that $c(x,y)$~is minimized. We say that $y$~is
a \emph{feasible} solution for input~$y$ if $c(x,y)$~is defined.
\end{definition}

Each optimization problem has an associated decision problem that
asks whether there exists a solution of cost $\le k$,
for some constant~$k$. We say that an optimization problem is in
NP when its associated decision problem is in NP\null. Note that
if we have a solution to the associated decision problem, then
we can find the cost of an optimal solution with a logarithmic
overhead by using binary search. This is an example of reducing a
problem to another.

\begin{definition}\index{oracle}
An \emph{oracle} for problem~$P$ solves any instance of~$P$ in
constant time and constant space.
\end{definition}

\begin{definition}\index{reduction}
A decision problem~$P$ \emph{reduces} to a decision problem~$Q$
when there is a polynomial time solution for~$P$ that uses
an oracle for~$Q$ as a subroutine. A decision problem~$P$
\emph{transforms} to a decision problem~$Q$ when there is a
function~$f:L\to L$ computable in polynomial time such that $x\in
P$ if and only if $f(x)\in Q$, for all inputs~$x$.
\end{definition}

\begin{remark}
If problem~$P$ transforms to problem~$Q$, then problem~$P$ also reduces to
problem~$Q$.
\end{remark}

\begin{definition}\index{NP-complete}\index{NP-hard}
A problem is \emph{NP-complete} when it is in~NP and all other
problems in~NP transform to it. A problem is \emph{NP-hard} when
all problems in~NP reduce to it.
\end{definition}

As before, we say that an optimization problem is
NP-complete\slash NP-hard when its associated decision problem is
NP-complete\slash NP-hard.

No polynomial solution is known for any NP-complete problem. The
only way to handle large instances of NP-complete problems in
practice is to settle for approximations or heuristics.

\subsection{Algorithms and Data Structures Reminder}
\label{sec:passive-algo-back}

One of the proofs given later is by transformation from the mins
problem (\fb maximum \fb independent \fb node \fb set). Then a
conjecture is supported by a heuristic argument that makes use of
the lis problem (\fb longest \fb increasing \fb subsequence) and
the maximum bipartite matching problem. This section serves as a
brief reminder of what these problems are and what are the best
algorithms known for them.

\subsubsection{Maximum Independent Node Set}

\begin{definition}\index{adjacent nodes}\index{independent nodes}
Two nodes of a graph are \emph{adjacent} when they are the
endpoints of some edge. A set of nodes is \emph{independent} when
its elements are pairwise non-adjacent.
\end{definition}

\begin{problem}[maximum independent node set]
\index{problem!maximum independent node set}
Given is a graph. Find a largest independent set of nodes.
\end{problem}

The associated decision problem asks whether there is
an independent set of size~$\ge k$.

{
  \def\style#1{\ifx#1x\fnode\else\enode\fi}
  \def\picture#1#2#3{
    \begin{tikzpicture}[baseline=-.5ex,scale=.3]
      \style#1 (0) at (0,0) {};
      \style#2 (1) at (1,0) {};
      \style#3 (2) at (2,0) {};
      \draw (0)--(1); \draw (1)--(2);
    \end{tikzpicture}
  }

\begin{example}
The graph \picture ... has one independent set of size~$0$
(\picture ...) three independent sets of size~$1$
(\picture x.., \picture .x., \picture ..x) and one of size~$2$
(\picture x.x). The latter is a solution to this instance of
the mins problem. There are also three sets of nodes that
are not independent (\picture xx., \picture .xx, \picture xxx).
\end{example}}

This problem is known to be NP-complete.

\subsubsection{Longest Increasing Subsequence}

\begin{definition}\index{subsequence}
A \emph{subsequence} is obtained from a sequence $w_1$,~$w_2$,
\dots,~$w_n$ by removing some of its elements and maintaining
the relative order of the others.
\end{definition}

\begin{problem}[lis]\index{problem!longest increasing subsequence}
Given is a sequence of integers. Find a subsequence of it that
has maximum length.
\end{problem}

The best known solution for this problem works in $\Theta(n
\lg\lg n)$~time. It uses a data structure that maps integers to
integers and supports the operations $\mathit{update}(k,v)$ that
binds the key~$k$ to the value~$v$ and $\mathit{predecessor}(k)$
that returns the value last bound to the largest key that
is~$\le k$. For example, after $\mathit{update}(2,1)$,
$\mathit{update}(1,2)$, $\mathit{update}(2,3)$, the query
$\mathit{predecessor}(1)$ returns~$2$ and the queries
$\mathit{predecessor}(2)$ and $\mathit{predecessor}(3)$
return~$3$. For simplicity, let us assume $v\in\R_+$, which
means that $\mathit{predecessor}(k)$ returns~$0$ if the value
of~$k$ filters out all previous updates. The \textit{update}
and \textit{predecessor} operations are supported by binary
search trees in $\Theta(\lg n)$~time and by van~Emde Boas trees
in $\Theta(\lg\lg n)$~time. Figure~\ref{fig:lis-alg} shows an
algorithm that uses such a data structure to solve the lis
problem in $\Theta(n\lg\lg n)$~time. Strictly speaking, the
algorithm only returns the length of a longest subsequence, but
it can be easily modified to return the subsequence.

\begin{figure}\centering
\leavevmode\vbox{
\begin{alg}
\^  $\proc{lis}(w_1,w_2,\ldots,w_n)$
\=  $r:=0$
\=  ~for~ $i$ ~from~ $1$ ~up to~ $n$
\+    $l:=1+\mathit{predecessor}(w_i-1)$
\=    $r:=\max(r,l)$
\=    $\mathit{update}(w_i,l)$
\-  ~return~ $r$
\end{alg}}
\caption{A solution for lis}\label{fig:lis-alg}
\end{figure}

\subsubsection{Maximum Bipartite Matching}
\label{sec:passive.bipartite}

\begin{definition}\index{bipartite graph}
A graph is \emph{bipartite} when its nodes can be partitioned
into two subsets, the left nodes and the right nodes, such that
all edges are between a left and a right node.
\end{definition}

\begin{definition}\index{matching}
Given a graph, a \emph{matching} is a subset of pairwise
non-adjacent edges. Two edges are \emph{adjacent} if they share
(at least) a node.
\end{definition}

\begin{problem}[maximum bipartite matching]
\index{problem!maximum bipartite matching}
Given is a bipartite graph. Find a largest matching.
\end{problem}

The old Hungarian algorithm~\cite{kuhn1955} solves the problem in
$O(n^3)$~time.  Other algorithms solve the problem in $O\bigl(\min(m\sqrt
n, n^{2.38})\bigr)$~time. Here, $m$~is the number of edges and $n$~is the
number of nodes.

\subsection{Examples, Terminology, and Notations}

The running example in Figure~\ref{fig:boogie-example} is the
simplest interesting case for passivation. Variable~$v$ is
the variable of interest. With respect to it, the nodes are
\emph{read-only}~(\tikz[baseline=-.5ex] \ronode() at (0,0){};),
\emph{write-only}~(\tikz[baseline=-.5ex] \wonode() at (0,0){};),
or \emph{read-write}~(\tikz[baseline=-.5ex] \rwnode() at
(0,0){};). A statement that does not involve the variable of
interest will be drawn as a white circle~(\tikz[baseline=-.5ex]
\oonode() at (0,0){};). There is an easy mnemonic rule for these
notations: Since execution flows downward and reading precedes
writing, the upper half corresponds to reading and the lower half
to writing. The terms \emph{read node}\index{read node} (for \tikz[baseline=-.5ex]
\ronode() at (0,0){}; or~\tikz[baseline=-.5ex] \rwnode() at
(0,0){};) and \emph{write node}\index{write node} (for \tikz[baseline=-.5ex]
\wonode() at (0,0){}; or~\tikz[baseline=-.5ex] \rwnode() at
(0,0){};) will also be used. In Figure~\ref{fig:boogie-example},
nodes~1,~2, 3, and~4 are read nodes, while nodes 0 and~3 are
write nodes.

\begin{figure}\centering
\begin{minipage}{0.5\linewidth}\centering
\begin{boogie}
procedure makeEven(u : int) returns (v : int) {
  $L_0$:  v := u; goto $L_1$, $L_2$;
  $L_1$:  assume !even(v); goto $L_3$;
  $L_2$:  assume even(v); goto $L_4$;
  $L_3$:  v := v + 1;
  $L_4$:  assert even(v); return;
}
\end{boogie}
\end{minipage}
\begin{minipage}{0.3\linewidth}\centering
\begin{tikzpicture}
  [scale=0.5]
  \wonode (A) at (1,3) [label=above:$0$] {};
  \ronode (B) at (0,2) [label=left:$1$] {};
  \ronode (C) at (2,2) [label=right:$2$] {};
  \rwnode (D) at (0,1) [label=left:$3$] {};
  \ronode (E) at (1,0) [label=below:$4$] {};
  \foreach \x/\y in {A/B, B/D, D/E, A/C, C/E}
    \draw[arr] (\x)--(\y);
\end{tikzpicture}
\end{minipage}
\caption{An example of a Boogie program}
\label{fig:boogie-example}
\end{figure}
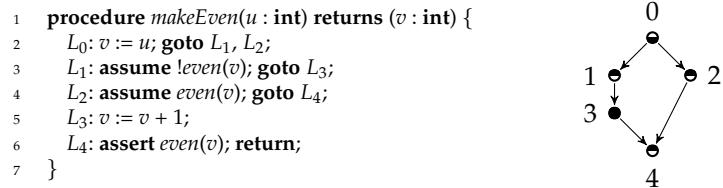

Figure~\ref{fig:ex-vo-vs-co} shows a flowgraph with a slightly
different drawing convention: This time the arrows are missing
and edges are understood to go downward. This convention is
not limiting because all inputs to the passivation phase are
dags. All nodes of this flowgraph read and write the variable of
interest.

Figure~\ref{fig:ex-iv-vs-dv} shows another flowgraph, using
an even more complicated drawing convention: A dotted edge
\begin{tikzpicture}[baseline=-.5ex, scale=0.5] \enode(A) at (0,0)
[label=left:$x$]{}; \enode(B) at (1,0) [label=right:$y$]{};
\draw[densely dotted] (A)--(B); \end{tikzpicture} stands
for a read-only node (\tikz[baseline=-.5ex]\ronode (A) at
(0,0){};) that has (normal) incoming edges from node~$x$
and from node~$y$. So Figure~\ref{fig:ex-iv-vs-dv} depicts
a flowgraph with $7$~nodes and $8$~edges. Similarly,
the third example in Figure~\ref{fig:ex-special-family}
represents a flowgraph with $14$~nodes ($9$~write-only and
$5$~read-only) and $18$~edges. Without the dotted edge
convention, Figure~\ref{fig:ex-special-family} would be much
bigger and harder to grasp.

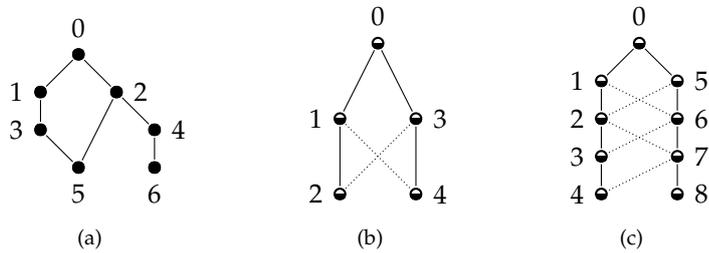
\begin{figure}\centering
\subfigure[]{\label{fig:ex-vo-vs-co}
\begin{tikzpicture}
  \rwnode (A) at (0.5,1.5) [label=above:$0$] {};
  \rwnode (B) at (0,1) [label=left:$1$] {};
  \rwnode (C) at (1,1) [label=right:$2$] {};
  \rwnode (D) at (0,0.5) [label=left:$3$] {};
  \rwnode (E) at (1.5,0.5) [label=right:$4$] {};
  \rwnode (F) at (0.5,0) [label=below:$5$] {};
  \rwnode (G) at (1.5,0) [label=below:$6$] {};
  \foreach \x/\y in {A/B,B/D,D/F,A/C,C/F,C/E,E/G}
    \draw (\x)--(\y);
\end{tikzpicture}}\hfil
\subfigure[]{\label{fig:ex-iv-vs-dv}
\begin{tikzpicture}
  \wonode (A) at (0.5,2) [label=above:$0$] {};
  \wonode (B) at (0,1) [label=left:$1$] {};
  \wonode (C) at (0,0) [label=left:$2$] {};
  \wonode (D) at (1,1) [label=right:$3$] {};
  \wonode (E) at (1,0) [label=right:$4$] {};
  \draw (A)--(B)--(C);
  \draw (A)--(D)--(E);
  \draw[densely dotted] (B)--(E);
  \draw[densely dotted] (C)--(D);
\end{tikzpicture}}\hfil
\subfigure[]{\label{fig:ex-special-family}
\begin{tikzpicture}[scale=0.5]
  \foreach \n/\x/\y/\p/\t in {
    0/1/4/above/,
    1/0/3/left/,
    2/0/2/left/,
    3/0/1/left/,
    4/0/0/left/,
    5/2/3/right/,
    6/2/2/right/,
    7/2/1/right/,
    8/2/0/right/}
    \wonode (\n) at (\x,\y) [label=\p:$\n\t$] {};
  \foreach \x/\y in {
    0/1,1/2,2/3,3/4,
    0/5,5/6,6/7,7/8}
    \draw (\x)--(\y);
  \foreach \x/\y in {
    5/2,
    6/3,6/1,
    7/4,7/2}
    \draw[densely dotted] (\x)--(\y);
\end{tikzpicture}}
\caption{Various interesting special cases}
\label{fig:special-cases}
\end{figure}

\section{The Definition of Passive Form}
\label{sec:passive}

The main purpose of this section is to give a precise formulation
for the problem of finding a good passive form.

\begin{example}\index{version of variable}
A passive form of the program in Figure~\ref{fig:boogie-example}
appears in Figure~\ref{fig:passive-form}. It is obtained by
introducing \emph{versions}~$0$ and~$1$ of the variable~$v$ and
by inserting the copy statement~$5$.
\label{ex:pass}
\end{example}

In general, a \emph{copy statement}\index{copy statement} (or \emph{copy
node}) has the shape $v_i:=v_j$ and is drawn as a filled
square~(\tikz[baseline=-.5ex]\cnode at (0,0){};). A copy node is
a read-write node.

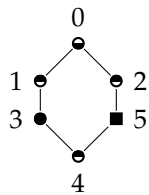
\begin{figure}\centering
\begin{minipage}{0.5\linewidth}\centering
\begin{boogie}
procedure makeEven(u : int) returns (v : int) {
  $L_0$:  $v_0$ := u; goto $L_1$, $L_2$;
  $L_1$:  assume !even($v_0$); goto $L_3$;
  $L_2$:  assume even($v_0$); goto $L_5$;
  $L_3$:  $v_1$ := $v_0$ + 1;
  $L_4$:  assert even($v_1$); return;
  $L_5$:  $v_1$ := $v_0$; goto $L_4$;
}
\end{boogie}
\end{minipage}
\begin{minipage}{0.3\linewidth}\centering
\begin{tikzpicture}
  [scale=0.5]
  \wonode (A) at (1,3) [label=above:$0$] {};
  \ronode (B) at (0,2) [label=left:$1$] {};
  \ronode (C) at (2,2) [label=right:$2$] {};
  \rwnode (D) at (0,1) [label=left:$3$] {};
  \ronode (E) at (1,0) [label=below:$4$] {};
  \cnode (F) at (2,1) [label=right:$5$] {};
  \foreach \x/\y in {A/B, B/D, D/E, A/C, C/F, F/E}
    \draw (\x)--(\y);
\end{tikzpicture}
\end{minipage}
\caption{A passive form for the program in Figure~\ref{fig:boogie-example}}
\label{fig:passive-form}
\end{figure}

Example~\ref{ex:pass} contains only one variable and this is
the case we shall analyze in detail. Multiple variables do not
introduce any new complications: A program can be made passive
with respect to each of its variables in turn, while considering
the others to be constants.

\begin{definition}\index{passive program}
A program is \textsl{passive} when (1)~on all execution paths each variable
is written at most once, and (2)~no statement reads and writes the same
variable. 
\label{def:passive}
\end{definition}

\begin{remark}
Obtaining a passive form is similar to automatically deriving
a functional equivalent of a loop-less imperative program. The
remaining assignments, which are eventually transformed into
assumptions, can be seen as \textbf{let} bindings.
\end{remark}

The passive form of a program~$G'$ is an equivalent program~$G$
that is passive. If there are execution paths in program~$G'$
that write to variable~$v$ multiple times then those writes must
be changed to write to distinct variables in program~$G$. We
denote the variables of program~$G$ by $v_i$ where $i$ is some
integer and say that $v_i$ is version~$i$ (of variable~$v$).
To maintain the semantics, each read from variable~$v$ must
be replaced by a read from the latest written version. As
Example~\ref{ex:pass} illustrates, it is necessary sometimes
to alter the structure of the program, and in current
approaches~\cite{flanagan2001passive, barnett2005wpu} this is always
done by inserting copy statements of the form $v_i:=v_j$. The
following definition makes these notions precise.

\begin{definition}\index{passive form}
A program~$G$ whose flowgraph has nodes~$V$ is a 
\emph{passive form} of the program~$G'$ whose flowgraph has nodes~$V'$
when

\begin{enumerate}
\item program~$G$ is passive and
\item there exists a mapping $c:V'\to V$, a \emph{write-version}
  function\index{write-version function} $w:V\to\Z$, 
  and a \emph{read-version} function\index{read-version function}
  $r:V\to\Z$ such that
  \begin{enumerate}
  \item 
    \emph{statement structure is preserved}: the statement~$c(x)$
    is obtained from statement~$x$ by replacing each occurrence
    of variable~$v$ by some variable~$v_i$,
  \item 
    \emph{the new nodes are copy statements}: all statements in
    $C=V-c(V')$ have the form $v_i:=v_j$,
  \item 
    \emph{the flow structure is preserved}: there is an
    edge $x\to y$ in program~$G'$ if and only if there is
    a path~$c(x)\stackrel{Q}{\leadsto}c(y)$ in program~$G$
    that uses only copy nodes as intermediate nodes, that is,
    $\bigl(Q-\{x,y\}\bigr) \subseteq C$; also, all copy
    nodes have at least one outgoing edge,
  \item 
    \emph{the initial node is preserved}: $c(0) = 0$,
  \item
    \emph{writes and reads are confined}: each statement~$x$ in
    program~$G$ may only read version~$r(x)$ and may only write
    version~$w(x)$,
  \item
    \emph{the read version is always the latest written version}:
    $w(x)=r(y)$ if $x$~is a write node, $y$~is a read node,
    and there is a path~$x\leadsto y$ in~$G$ that contains no
    intermediate write nodes.
  \end{enumerate}
\end{enumerate}
\label{def:passive-form}
\end{definition}

We say that functions $c$, $w$, and $r$ \emph{witness}\index{witness} that
program~$G$ is a passive form of program~$G'$.

\begin{remark}
The definition is fairly straightforward but it is important
to have it written down. The definition naturally leads to two
notions of optimality (which have been previously missed) and to
a straightforward algorithm that is better than some previous
solutions.
\end{remark}

\begin{example}
If a program starts with \textbf{goto}~$l_1,\ldots,l_n$ and
continues with $l_k:\>S;\>\textbf{return}$ for all~$k$, then $S$
is a passive form of it. Thus, the passive form may be smaller
than the original program, although we will usually keep the
mapping~$c$ injective.
\end{example}

It is easy to see that the passive form of a program~$G$ is
correct if and only if program~$G$ is correct. The reason is
that every read of a version by the passive form reads the same
value as the corresponding read of the variable in program~$G$,
and every write of a version by the passive form writes the same
value as the corresponding write to the variable in program~$G$.

\subsection{Types of Passive Forms}
\label{sec:passive-types}

The requirement that program~$G$ is passive can be expressed as
a constraint on the write-version function~$w$.

\begin{proposition}
If the write-version function~$w$ is a witness that program~$G$
is a passive form, then the existence of a path $x\leadsto y$
where both $x$ and~$y$ are write nodes implies that $w(x)\ne
w(y)$.
\label{prop:distinct}
\end{proposition}

A write node~$x$ writes to version~$w(x)$. According to
Definition~\ref{def:passive}, the same version must not be
written by any other node on an execution path that includes
node~$x$.

Proposition~\ref{prop:distinct} suggests that a ``passive form''
may be called more explicitly ``distinct-version passive form.''
\index{passive form!distinct-version}
Definition~\ref{def:passive-form} is very general. The passive
forms obtained by previous approaches all satisfy a stronger
definition that corresponds to the intuition that versions
increase in time.

\begin{definition}\index{passive form!increasing-version}
An \emph{increasing-version passive form} is a passive form
witnessed by a write-version function~$w$ with the property that
$w(x)<w(y)$ whenever there is a path~$x\leadsto y$ from a write
node~$x$ to another write node~$y$.
\label{def:increasing-version}
\end{definition}

Programs have multiple passive forms, some better than others.
There are two natural notions of optimality.

\begin{definition}\index{passive form!version-optimal}
A passive form is \emph{version-optimal} when it uses as few
variable versions as possible.
\label{def:version-optimal}
\end{definition}

\begin{definition}\index{passive form!copy optimal}
A passive form is \emph{copy-optimal} when it uses as few copy
nodes as possible.
\label{def:copy-optimal}
\end{definition}

We say that \emph{a version~$i$ is used by a passive form} when
there is a write node whose write version is~$i$. This convention
simplifies the later analysis, but it requires some explanation.
Why is it OK to ignore the read version of read nodes? The read
version either is written by the program, in which case it was
already counted as the write version of another node, or it is
not written by the program. All versions that are not written
are essentially equivalent. In other words, by ignoring the read
versions of read nodes we undercount by one, if the uninitialized
variable is read by the original program. That is a small price
to pay for a reduction in the number of special cases that we
must consider in the following analysis.

Let us summarize the information in
Definitions~\ref{def:passive-form},~\ref{def:increasing-version},
\ref{def:version-optimal}, and~\ref{def:copy-optimal}. Each
program~$G$ determines a set of (distinct-version) passive forms
according to Definition~\ref{def:passive-form}; each program
determines a set of increasing-version passive forms according to
Definition~\ref{def:increasing-version}. Each increasing-version
passive form is a distinct-version passive form. Each passive
form has two associated costs---the number of versions and the
number of copy nodes. Therefore, for each program there are four
(not necessarily distinct) interesting classes of passive forms:

\begin{enumerate}
\item the version-optimal increasing-version passive forms,
\item the copy-optimal increasing-version passive forms,
\item the version-optimal distinct-version passive forms, and
\item the copy-optimal distinct-version passive forms.
\end{enumerate}

One interesting problem is whether these four classes always
overlap. In other words, is there always a passive form that is
optimal by all measures and also follows the intuitive rule that
versions do not decrease in time? But, before addressing this
question, let us first look at how good a version-optimal passive
form can be.

\begin{lemma}
The number of versions of any passive form is greater or equal to
the number of write nodes on any execution path in the original
program.
\label{lemma:min-versions}
\end{lemma}

\begin{proof} 
Consider an execution path~$P$ in the original program and its
subset of write nodes~$Q\subseteq P$. For two nodes $x,y\in Q$,
we can assume that there is a path~$x\leadsto y$. According
to Proposition~\ref{prop:distinct}, their write versions are
distinct. Hence, $|w(Q)|=|Q|$. 
\end{proof}

We can now show two facts about how different types of passive
forms relate to each other---Proposition~\ref{prop:iv=>(vo!=co)}
and Proposition~\ref{prop:coDv<coIv}.

\begin{proposition}
There are programs for which an increasing-version passive
form cannot be both version-optimal and copy-optimal.
\label{prop:iv=>(vo!=co)}
\end{proposition}

\begin{proof} 
One such program has the flowgraph depicted in
Figure~\ref{fig:ex-vo-vs-co} on page~\pageref{fig:ex-vo-vs-co}.
A version-optimal passive form appears in
Figure~\ref{fig:vo-vs-co-sol1} and a copy-optimal passive form
appears in Figure~\ref{fig:vo-vs-co-sol2}. In these drawings,
each node~$x$ is labeled with~$x\sqatop{r(x)}{w(x)}$.
For example, the label $0\sqatop{-1}{0}$ is put on
node~$0$, which has read version~$-1$ and write version~$0$.

It is easy to see that these solutions are optimal. According
to Lemma~\ref{lemma:min-versions} the number of versions in any
passive form is $\ge4$ and Figure~\ref{fig:vo-vs-co-sol1} uses
$4$~versions. Figure~\ref{fig:vo-vs-co-sol2} uses $0$~copy nodes.

It remains to show that there is no optimal increasing-version
passive form for this program that uses $4$~versions and $0$~copy
nodes. We do this by showing that all increasing-version passive
forms that have no copy node must use $\ge5$ versions. According
to Definition~\ref{def:increasing-version}, $w(0)<w(1)<w(3)$
and $w(2)<w(4)<w(6)$. Because there are no copy operations, the
Definition~\ref{def:passive-form} gives us $w(3)=r(5)=w(2)$.
Hence, the $5$ used versions $w(0)$,~$w(1)$, $w(3)=w(2)$, $w(4)$,
and~$w(6)$ are distinct.
\end{proof}

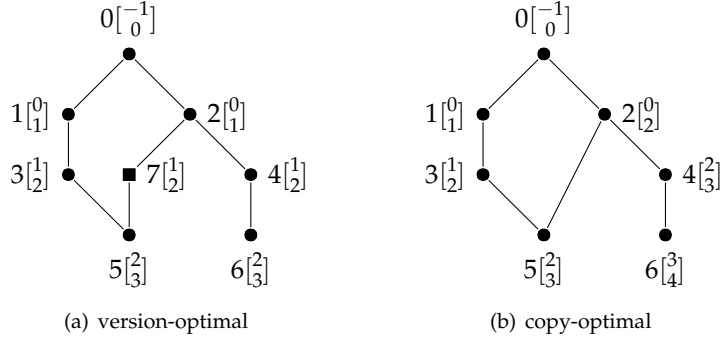
\begin{figure}\centering
\subfigure[version-optimal]{\label{fig:vo-vs-co-sol1}
\begin{tikzpicture}[scale=.8]
  \foreach \n/\x/\y/\p/\r/\w in {
    0/1/3/above/-1/0, 
    1/0/2/left/0/1,
    2/2/2/right/0/1,
    3/0/1/left/1/2,
    4/3/1/right/1/2,
    5/1/0/below/2/3,
    6/3/0/below/2/3}
  \rwnode (\n) at (\x,\y) [label=\p:$\n\sqatop{\r}{\w}$] {};
  \cnode (7) at (1,1) [label=right:$7\sqatop{1}{2}$] {};
  \foreach \x/\y in {
    0/1,1/3,3/5,
    0/2,2/7,7/5,
    2/4,4/6}
  \draw (\x)--(\y);
\end{tikzpicture}}\hfil
\subfigure[copy-optimal]{\label{fig:vo-vs-co-sol2}
\begin{tikzpicture}[scale=.8]
  \foreach \n/\x/\y/\p/\r/\w in {
    0/1/3/above/-1/0, 
    1/0/2/left/0/1,
    2/2/2/right/0/2,
    3/0/1/left/1/2,
    4/3/1/right/2/3,
    5/1/0/below/2/3,
    6/3/0/below/3/4}
  \rwnode (\n) at (\x,\y) [label=\p:{$\n\sqatop{\r}{\w}$}] {};
  \foreach \x/\y in {
    0/1,1/3,3/5,
    0/2,2/5,
    2/4,4/6}
  \draw (\x)--(\y);
\end{tikzpicture}}
\caption{Optimal increasing-version passive forms for Figure~\ref{fig:ex-vo-vs-co}}
\label{fig:vo-vs-co-sol}
\end{figure}

The second fact justifies the \emph{unintuitive} choice to allow
versions to decrease in time.

\begin{proposition}
There are programs for which copy-optimal distinct-version
passive forms use strictly fewer copy nodes than copy-optimal
increasing-version passive forms.
\label{prop:coDv<coIv}
\end{proposition}

\begin{proof}
One such program has the flowgraph depicted in
Figure~\ref{fig:ex-iv-vs-dv} on page~\pageref{fig:ex-iv-vs-dv}.
A distinct-version passive form with no copy nodes appears
in Figure~\ref{fig:iv-vs-dv-sol}. (Note that the two nodes
represented earlier with dotted lines are now explicit.)

It remains to show that there is no increasing-version
passive form than uses no copy node. Suppose that
there is one. Then, $w(1)=w(4)$ and $w(2)=w(3)$. But
Definition~\ref{def:increasing-version} tells us that $w(1)<w(2)$
and $w(3)<w(4)$. Contradiction.
\end{proof}

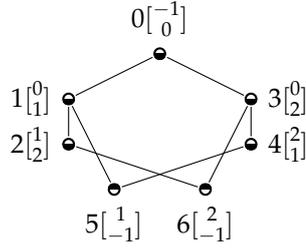
\begin{figure}\centering
\begin{tikzpicture}[scale=.6]
  \foreach \n/\x/\y/\p/\r/\w in {
    0/2/3/above/-1/0,
    1/0/2/left/0/1,
    2/0/1/left/1/2,
    3/4/2/right/0/2,
    4/4/1/right/2/1}
    \wonode (\n) at (\x,\y) [label=\p:$\n\sqatop{\r}{\w}$] {};
  \foreach \n/\x/\r in {
    5/1/1,
    6/3/2}
    \ronode (\n) at (\x,0) [label=below:$\n\sqatop{\r}{-1}$] {};
  \foreach \f/\t in {0/1,1/5,1/2,2/6,0/3,3/6,3/4,4/5}
    \draw (\f)--(\t);
\end{tikzpicture}
\caption{A copy-optimal passive form for Figure~\ref{fig:ex-iv-vs-dv}}
\label{fig:iv-vs-dv-sol}
\end{figure}

\section{The Version-Optimal Passive Form}
\label{sec:version-optimal}

\subsection{The Algorithm}

Proposition~\ref{prop:coDv<coIv} says that it is beneficial
to allow versions to decrease in time if we are looking for a
copy-optimal passive form. This section shows that this is not
the case if the objective is a version-optimal passive form. The
proof is constructive---an algorithm that finds a version-optimal
increasing-version passive form that is as good as any passive
form can be.

The algorithm from Figure~\ref{fig:pass-algo} is \emph{the
best practical solution for the passivation problem}.
It is very similar to the algorithm of Flanagan and
Saxe~\cite{flanagan2001passive}, being changed so that it
always finds a version-optimal passive form. The method
$\mathit{VersionOptimalPassiveForm}(G)$ modifies the program~$G$
into one of its passive forms. (The real implementation in
FreeBoogie transforms the program without mutating it, as was
explained in Section~\ref{sec:visitors}.)

\begin{figure}\centering
\leavevmode\vbox{
\begin{alg}
\^  $\proc{Read}(y)$  \comment memoized
\=  $r:=-1$ 
\=  ~for~ ~each~ predecessor $x$ of $y$
\+    $r:=\max\bigl(r,\mathit{Write}(x)\bigr)$
\0  ~return~ $r$
\end{alg}
\bigskip
\begin{alg}
\^  $\proc{Write}(y)$  \comment memoized
\=  $r:=\mathit{Read}(y)$
\=  ~if~ $y$ is a write node 
\+    $r:=r+1$
\0  ~return~ $r$
\end{alg}
\bigskip
\begin{alg}
\^  $\proc{VersionOptimalPassiveForm}(G)$
\=  ~for~ ~each~ node $x$ of $G$
\+    substitute reads of $v$ by reads from version $\mathit{Read}(x)$
\=    and writes of $v$ by writes to version $\mathit{Write}(x)$
\0  ~for~ ~each~ edge $x\to y$ of $G$
\+    ~if~ $\mathit{Write}(x)\ne\mathit{Read}(y)$
\+      create a new node $z:=\big(v_{\mathit{Read}(y)}:=v_{\mathit{Write}(x)}\big)$
\=      remove the edge $x\to y$
\=      add edges $x\to z$ and $z\to y$
\end{alg}}
\caption{A practical algorithm for finding a passive form}
\label{fig:pass-algo}
\end{figure}

Let us see, slowly, why this algorithm is correct. The
mapping~$c:V'\to V$ from nodes in the original program to nodes
in its passive form is $c(x)=x$. The method $\mathit{Read}(y)$
computes the read-version~$r(y)$, but only for nodes~$y$
that exist in the original graph; similarly, the method
$\mathit{Write}(y)$ computes the write-version~$w(y)$, but only
for nodes~$y$ that exist in the original graph. With a different
notation, the read-version function and the write-version
function are defined to be
\begin{align}
r(y) &= 
  \begin{cases}
  \max_{x\to y} w(x) & \text{if $y$ has predecessors}\\
  -1 & \text{otherwise}
  \end{cases}, \label{eq:rvo} \\
w(y) &= r(y) + [\text{node~$y$ is a write node}]. \label{eq:wvo}
\end{align}
If we unfold~$r$,
\begin{equation}
w(y) = [\text{$y$ writes}] +
  \begin{cases}
  \max_{x\to y} w(x) & \text{if $y\ne0$} \\
  -1 & \text{if $y=0$}
  \end{cases},
\end{equation}
it becomes clear that $w(y)+1$ is the maximum number
of write nodes on a path~$0\leadsto y$. It follows that $r(y)$
and~$w(y)$ computed by \eqref{eq:rvo} and~\eqref{eq:wvo} are in
the range~$-1\dts n-1$, where $n$~is the maximum number of write
nodes on an execution path. Moreover, $w(y)$ can be~$-1$ only if
$y$~is not a write node. So far we know that for write nodes~$y$
in the original graph, $w(y)$~is in the range~$0\dts n-1$.

The output graph also has copy nodes~$z$ created on line~6 of
\textit{VersionOptimalPassiveForm}. Each such node corresponds
to an edge~$x\to y$ in the original graph and has $w(z)=r(y)$.
Hence, the version written by copy nodes is also in the
range~$0\dts n-1$. (If $r(y)=-1$ then $w(x')=-1$ for all
predecessors~$x'$ of~$y$, including~$x$, so there would be no
copy node created.)

It follows immediately from Lemma~\ref{lemma:min-versions} that
no passive form has a number of versions~$<n$, so we proved that
if the output of \textit{VersionOptimalPassiveForm} is indeed a
passive form then it must be version-optimal. Most of the conditions
imposed by Definition~\ref{def:passive-form} on page~\pageref{def:passive-form}
are easy to check.
\begin{enumerate}
\item \emph{Passive:}
  The write version increases at every write node, so it cannot
  be the same as the write version of any preceding node.
\item
  \begin{enumerate}
  \item \emph{Statement structure is preserved}:
    The statement structure is only modified by lines~2--3.
  \item \emph{The new nodes are copy statements}:
    New nodes are only created on line~6.
  \item \emph{The flow structure is preserved}:
    Edges are modified only on lines~7--8. Replacing~$x\to y$
    by $x\to z\to y$ maintains paths, because node~$z$ has no
    other adjacent edges.
  \item \emph{The initial node is preserved}:
    We have $c(x)=x$ for all nodes~$x$ in the original graph.
  \item \emph{Writes and reads are confined}:
    For existing nodes, see lines~2--3. For each copy
    node we can simply choose the read version and the
    write version appropriately.
  \item \emph{The read version is always the latest written version}:
    For each edge~$x\to y$ we have $w(x)=r(y)$. This is true
    even for edges adjacent to copy nodes. Also $w(x)=w(y)$ if
    node~$y$ is not a write node. If there is a path $x_0\to
    x_1\to\cdots\to x_n\to x_{n+1}$, and $x_1,\ldots,x_n$ are not
    write nodes then $w(x_0)=w(x_1)=\cdots=w(x_n)=r(x_{n+1})$.
  \end{enumerate}
\end{enumerate}

We proved the following theorem.

\begin{theorem}
The algorithm in Figure~\ref{fig:pass-algo} constructs a
version-optimal (distinct-version) passive form, which is also an
increasing-version passive form.
\label{th:pass-algo}
\end{theorem}

The algorithm is fast. The first loop of \textit{VersionOptimalPassiveForm}
(lines~2--3) is executed once for every node~$x$ in~$V$; the second loop
(line~5) is executed $|E|$~times, where $E$~is the set of edges in the
flowgraph. The substitutions performed on lines 2 and~3 can be done in time
proportional to the AST size~$|x|$ of the statement~$x$. Because of
memoization, the bodies of \textit{Read} and \textit{Write} are executed at
most $|V|$~times each. One way to memoize the two methods (that is, to
cache their results) is to have two integer fields in the data-structure
representing a flowgraph node. A sentinel value, say~$-2$, would represent
`not yet computed', while any other value would mean that the computation
was already done. These two memory cells per node are the only significant
memory used by the algorithm, apart from the memory used to represent the
input and the output.

\begin{proposition}
The algorithm in Figure~\ref{fig:pass-algo} on
page~\pageref{fig:pass-algo} uses $\Theta(|E|+\sum_{x\in V}|x|)$
time and $\Theta(|V|)$~temporary space.
\label{prop:pass-algo-complexity}
\end{proposition}

If the data-structure for representing a flowgraph
node is immutable, then memoization can be implemented
using a hashtable, in which case the time bound in
Proposition~\ref{prop:pass-algo-complexity} is true with high
probability, but not always. Also, the constant hidden by the
space bound is bigger.

The input occupies $\Theta(|E|+\sum_{x\in V}|x|)$ space and
only $O(|E|+\sum_{x\in V}|x|)$ extra space can be allocated in
$\Theta(|E|+\sum_{x\in V}|x|)$~time, so the output must occupy,
asymptotically, the same amount of space as the input. In fact,
the difference in size is accounted entirely by the copy nodes.

\begin{proposition}
The algorithm in Figure~\ref{fig:pass-algo} creates $O(|E|)$~copy
nodes.
\label{prop:pass-algo-result}
\end{proposition}

\begin{example}
If we apply this algorithm to the program in
Figure~\ref{fig:diamonds-seq} we obtain the result in
Figure~\ref{fig:diamonds-passive}, which is both copy-optimal
(since it has no copy nodes) and version-optimal (by
Theorem~\ref{th:pass-algo}). It is now trivial to get
rid of assignments: Just replace each assignment~$u:=e$
by the assumption \textbf{assume}~$u=e$ (see
Figure~\ref{fig:diamonds-noasgn}). The size of the weakest
precondition of the initial node is now linear in the size of the
program. The weakest precondition consists of (1)~the expressions
present in the program, (2)~at most one new $\land$ for each
assertion, (3)~at most one new $\limp$ for each assumption, and
(4)~at most one new $\land$ for each flowgraph node with multiple
successors. Figure~\ref{fig:diamonds-vc} shows the weakest
precondition SMT tree for the case~$n=2$.
\label{ex:diamonds-passive}
\end{example}

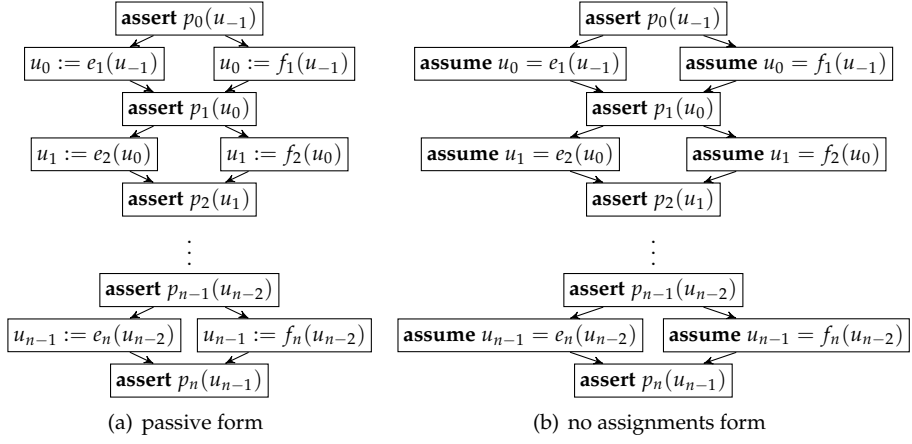
\begin{figure}\centering
\def\fig#1{
  \footnotesize
  \begin{tikzpicture}[yscale=0.6,xscale=#1]
    \newcount\vertpos\vertpos=9
    \def\n##1##2(##3,##4)##5{\node[draw,rectangle,inner sep=2pt] (##1##2) at (##3,##4) {##5};}
    \def\lr##1##2##3{\advance\vertpos by-1
      \n e##3(-2.5,\the\vertpos){\stmt{u_{##2}}{e_{##3}(u_{##1})}}
      \n f##3(2.5,\the\vertpos){\stmt{u_{##2}}{f_{##3}(u_{##1})}} }
    \def\m##1##2{\advance\vertpos by-1 \n p##2(0,\the\vertpos)%
      {\textbf{assert}~$p_{##2}(u_{##1})$}}
    
    \m{-1}0\lr{-1}01\m01\lr012\m12\advance\vertpos by-1
    \m{n-2}{n-1}\lr{n-2}{n-1}n\m{n-1}{n}
    \node at (0,3) {$\vdots$};

    \draw[arr] (p0) -- (e1); \draw[arr] (p0) -- (f1);
    \foreach \x in {1, 2, n} {
      \draw[arr] (e\x) -- (p\x); \draw[arr] (f\x) -- (p\x);
    }
    \draw[arr] (p1) -- (e2); \draw[arr] (p1) -- (f2);
    \draw[arr] (pn-1) -- (en); \draw[arr] (pn-1) -- (fn);
  \end{tikzpicture}
}
\subfigure[passive form]{
  \label{fig:diamonds-passive}
  \def\stmt#1#2{$#1:=#2$}\fig{0.5}}
\hfil
\subfigure[no assignments form]{
  \label{fig:diamonds-noasgn}
  \def\stmt#1#2{\textbf{assume}~$#1=#2$}\fig{0.7}}
\caption{Diamonds that do not explode}
\label{fig:nonexploding-diamonds}
\end{figure}

\begin{figure}\centering
\begin{tikzpicture}[xscale=2,yscale=.5]
  \footnotesize
  \node (e1) at (-2,0) {$e_1$};
  \node (e2) at (-1,0) {$e_2$};
  \node (p2) at (0,0) {$p_2$};
  \node (f2) at (1,0) {$f_2$};
  \node (f1) at (2,0) {$f_1$};
  \node (p1) at (1.5,0) {$p_1$};
  \node (p0) at (2.5,0) {$p_0$};
  \node (u1) at (0,-1) {$u_1$};
  \node (u0) at (0,-2) {$u_0$};
  \node (u-1) at (0,-3) {$u_{-1}$};
  \node (eq-2) at (-2,1) {$=$};
  \node (eq-1) at (-1,1) {$=$};
  \node (eq+1) at (1,1) {$=$};
  \node (eq+2) at (2,1) {$=$};
  \node (imp-2) at (-1,4) {$\limp$};
  \node (imp-1) at (-1,2) {$\limp$};
  \node (exp+1) at (1,2) {$\Leftarrow$};
  \node (exp+2) at (1,4) {$\Leftarrow$};
  \node (and3) at (0,3) {$\land$};
  \node (and5) at (0,5) {$\land$};
  \foreach \m/\n in {
    e1/u-1, f1/u-1, p0/u-1,
    eq-2/u0, e2/u0, f2/u0, p1/u0,
    eq-1/u1, p2/u1, eq+1/u1,
    eq-2/e1, eq-1/e2, imp-1/p2, exp+1/p2, eq+1/f2, eq+2/f1,
    imp-2/eq-2, imp-1/eq-1, exp+1/eq+1, exp+2/eq+2,
    and3/imp-1, and3/exp+1,
    imp-2/and3, exp+2/and3,
    and5/imp-2, and5/exp+2}
    \draw[-] (\m)--(\n);
  \draw[-] (and3) to [bend left=40] (p1);
  \draw[-] (and5) to [bend left=40] (p0);
  \draw[-] (eq+2) to [bend left=20] (u0);
\end{tikzpicture}
\caption{The weakest precondition of two diamonds 
  (see Figure~\ref{fig:diamonds-seq})}
\label{fig:diamonds-vc}
\end{figure}
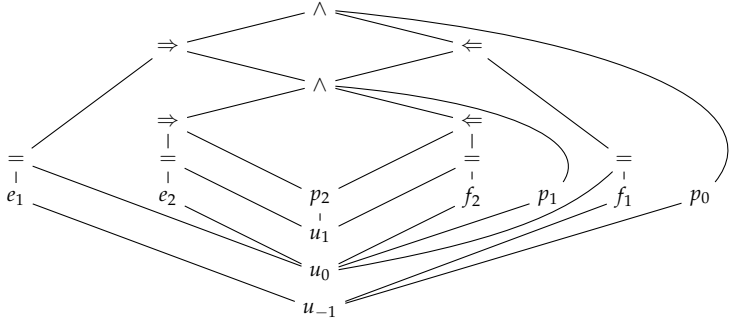

The bounds given by Propositions
\ref{prop:pass-algo-complexity}~and~\ref{prop:pass-algo-result} describe
passivation with respect to one variable. If there are $k$ variables, then
the bounds need to be multiplied by~$k$, in general. In core Boogie,
however, statements write to at most one variable. Since copy nodes for
variable~$v$ are inserted only on edges going to a statement writing to
variable~$v$, we can conclude the following that the \emph{total} number of
copy nodes is~$O(|E|)$.

The passive form computed for Example~\ref{ex:diamonds-passive} happened to be
copy-optimal. In general, this is not the case. In fact, sometimes the
algorithm introduces clearly redundant copy nodes. For example, suppose that a
read node~$x$ has three predecessors with the write-versions~$0$,~$0$, and~$1$,
respectively. The algorithm presented so far will insert two copy nodes
$v_1:=v_0$ between the predecessors with write-version~$0$ and node~$x$, but
only one is enough---copy nodes can be fused as long as the flowgraph structure
is preserved (condition~(c) of Definition~\ref{def:passive-form} on
page~\pageref{def:passive-form}). It is easy to modify the algorithm such that
copy nodes inserted before a given node~$x$ do not repeat. For example, one
could keep a hashtable that maps the triple (node~$x$, version~$i$,
version~$j$) to a copy node~$v_i:=v_j$ that gets reused if it is already in the
set.  (The technique is similar to hash-consing.)
 
\subsection{Experiments}
\label{sec:passivate-empiric}

Flanagan and Saxe~\cite{flanagan2001passive} did not try to
produce an optimal passive form. Still, their algorithm is very
similar to the one presented here, so it is interesting to see
how it compares. 

Microsoft Research released a set of Boogie
programs~\cite{boogiebench} to serve as a basis for comparing
program verifiers. The algorithm of Flanagan and Saxe,
however, does not handle the \textbf{goto} statement, which
is used extensively in the Boogie benchmark. Luckily, it
turns out that the \textbf{goto} statement is used in an
interesting way only in $16\%$ of the implementations in
the Boogie benchmark. In the other $1070$ implementations,
the flowgraphs are series--parallel, which means that they
correspond directly to programs that use only sequential
composition and nondeterministic choice for flow control. The
FreeBoogie implementation of the algorithm of Flanagan and
Saxe simply refuses to give an answer if the flowgraph is not
series--parallel.

Figure~\ref{fig:exp-version} compares the results of the two
algorithms. A variable for which Flanagan and Saxe introduce
$m$~versions when $n$~versions are enough contributes to
the disc centered at point~$(m,n)$. The more variables are
in this situation, the bigger the disc. (Its radius varies
logarithmically.) The plot summarizes the result of passivating
$1177$~variables.

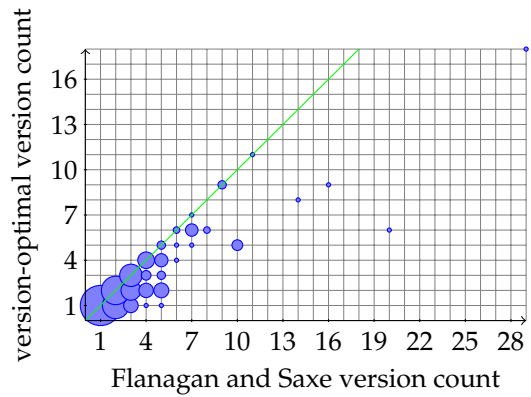
\begin{figure}\centering
  \begin{tikzpicture}[scale=0.2]
    \tikzstyle{datanode}=[circle,fill=blue!50,draw=blue]
    \draw[very thin,color=gray] (0,0) grid (29,18);
    \node[datanode,inner sep=5.356] at (1,1) {};
\node[datanode,inner sep=0.552] at (11,11) {};
\node[datanode,inner sep=3.516] at (2,1) {};
\node[datanode,inner sep=1.912] at (3,1) {};
\node[datanode,inner sep=0.552] at (4,1) {};
\node[datanode,inner sep=0.552] at (5,1) {};
\node[datanode,inner sep=0.552] at (29,18) {};
\node[datanode,inner sep=3.8] at (2,2) {};
\node[datanode,inner sep=2.628] at (3,2) {};
\node[datanode,inner sep=1.912] at (4,2) {};
\node[datanode,inner sep=1.98] at (5,2) {};
\node[datanode,inner sep=2.98] at (3,3) {};
\node[datanode,inner sep=1.284] at (4,3) {};
\node[datanode,inner sep=1.104] at (5,3) {};
\node[datanode,inner sep=2.212] at (4,4) {};
\node[datanode,inner sep=1.752] at (5,4) {};
\node[datanode,inner sep=0.552] at (6,4) {};
\node[datanode,inner sep=1.428] at (10,5) {};
\node[datanode,inner sep=1.104] at (5,5) {};
\node[datanode,inner sep=0.552] at (6,5) {};
\node[datanode,inner sep=0.552] at (7,5) {};
\node[datanode,inner sep=0.552] at (20,6) {};
\node[datanode,inner sep=0.876] at (6,6) {};
\node[datanode,inner sep=1.656] at (7,6) {};
\node[datanode,inner sep=0.876] at (8,6) {};
\node[datanode,inner sep=0.552] at (7,7) {};
\node[datanode,inner sep=0.552] at (14,8) {};
\node[datanode,inner sep=0.552] at (16,9) {};
\node[datanode,inner sep=1.104] at (9,9) {};
    \draw[green] (0,0) -- (18,18);
    \draw[->] (0,0) -- node[below,yshift=-.5cm] {Flanagan and Saxe version count} (29,0);
    \draw[->] (0,0) -- (0,18) node[sloped,above,pos=.5,yshift=.5cm] {version-optimal version count};
    \foreach \p in {1,4,...,16} {\draw (-.1,\p)--(.1,\p) node[left] {\p};}
    \foreach \p in {1,4,...,28} {\draw (\p,.1)--(\p,-.1) node[below] {\p};}
  \end{tikzpicture}
  \caption{Version count experimental comparison}
  \label{fig:exp-version}
\end{figure}

For $70\%$ of the variables in the Boogie benchmark both
algorithms say that one version is enough. For the other $30\%$
variables the algorithm of Flanagan and Saxe introduces on
average $46\%$ more versions than needed.

It is interesting to note that for randomly generated flowgraphs
the difference between the two algorithms is much bigger: Both
algorithms say that only one version is needed in less that $1\%$
of situations, while for the others the algorithm of Flanagan and
Saxe introduces on average $160\%$ more versions than needed. The
generator of random Boogie programs is part of FreeBoogie. It
works, roughly, by generating random series--parallel flowgraphs
by choosing graph grammar productions randomly and then deciding
for each node independently with some probability whether it is a
write node and whether it is a read node.

\section{The Copy-Optimal Passive Form}
\label{sec:copy-optimal}

We saw that finding a version-optimal passive form is rather
easy. In contrast, finding a copy-optimal passive form seems
rather hard. This is more evidence that the algorithm in
Section~\ref{sec:version-optimal} should be the algorithm of
choice in practice.

\subsection{The Increasing-Version Case}
\label{sec:passive.iv}

This section proves that finding a copy-optimal
increasing-version passive form is NP-complete. The proof is
by transformation from the mins problem (which is defined
in Section~\ref{sec:passive-algo-back}). Let us see, on the
example in Figure~\ref{fig:mins-to-copy}, how to find a maximum
independent set:

\begin{enumerate}
\item \emph{Transform the mins instance into a flowgraph.}
  A node~$x$ is transformed into the write only nodes
  $x_i$~and~$x_o$, the read only node~$x_r$, and edges $0\to x_o$
  and $x_o\to x_r$ and $x_i\to x_r$. (The two incoming edges of
  node~$x_r$ are drawn in Figure~\ref{fig:mins-to-copy} as
  \begin{tikzpicture}[baseline=-.5ex]
    \wonode (xo) at (0,0) [label=left:$x_o$] {};
    \wonode (xi) at (.5,0) [label=right:$x_i$] {};
    \draw[densely dotted] (xo) -- (xi);
  \end{tikzpicture}.) 
  An edge
  \begin{tikzpicture}
    [every node/.style={fgdraw},baseline=-.5ex]
    \node (x) at (0,0) [label=left:$x$] {};
    \node (y) at (.5,0) [label=right:$y$] {};
    \draw (x) -- (y);
  \end{tikzpicture}
  is transformed into edges $x_o\to y_i$ and~$y_o\to x_i$.

\item \emph{Find a copy-optimal increasing-version passive form
  of the flowgraph}. We do this by invoking an oracle.

\item \emph{Transform the passive form into a set~$I$ of nodes 
  in the original graph}. The set~$I$ contains nodes~$x$ for
  which the corresponding node~$c(x_r)$ is \emph{not} preceded by
  a copy node in the passive form.
\end{enumerate}

\begin{figure}\centering
\subfigure[mins instance]{
  \label{fig:mins-to-copy-graph}
  \begin{tikzpicture}
    [every node/.style={fgdraw},baseline=0]
    \node (A) at (-5,1) [label=above:$x$] {};
    \node (B) at (-4,1) [label=above:$y$] {};
    \node (C) at (-3,1) [label=above:$z$] {};
    \draw (A)--(B)--(C);
  \end{tikzpicture}
}\hfil
\subfigure[flowgraph]{
  \label{fig:mins-to-copy-flowgraph}
  \begin{tikzpicture}
    \wonode (S) at (1,2) {};
    \wonode (AO) at (0,1) [label=left:$x_o$] {};
    \wonode (AI) at (0,0) [label=left:$x_i$] {};
    \wonode (BO) at (1,1) [label=left:$y_o$] {};
    \wonode (BI) at (1,0) [label=left:$y_i$] {};
    \wonode (CO) at (2,1) [label=right:$z_o$] {};
    \wonode (CI) at (2,0) [label=right:$z_i$] {};
    \draw (S)--(AO)--(BI);
    \draw (S)--(BO)--(AI); \draw (BO)--(CI);
    \draw (S)--(CO)--(BI);
    \draw[densely dotted] (AO)--(AI);
    \draw[densely dotted] (BO)--(BI);
    \draw[densely dotted] (CO)--(CI);
  \end{tikzpicture}}
\caption{Transformation from mins}
\label{fig:mins-to-copy}
\end{figure}
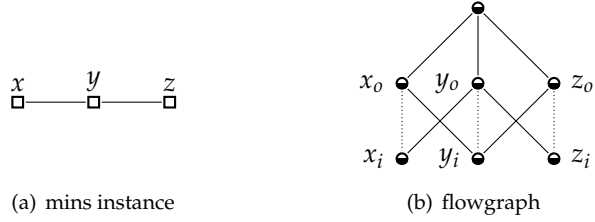

To understand why this procedure indeed finds a maximal
independent set let us focus on step~3. Clearly, it
defines a function from passive forms to sets of
nodes in the original graph. Since the definition of
a passive form (Definition~\ref{def:passive-form} on
page~\pageref{def:passive-form}) does not prevent redundant
copy nodes, all nodes~$c(x_r)$ might be preceded by copy nodes
and the corresponding set~$I$ needs not be independent. The
optimal passive form, however, will be among those passive
forms without redundant copy nodes. We say that a passive form
witnessed by a write version function~$w$ is \emph{non-redundant}
or that it \emph{has no redundant copy nodes} when it
contains exactly one copy node before each node~$c(x_r)$ with
$w(c(x_i))\ne w(c(x_o))$. It is easy to see that this implies
$r(c(x_r))=\max(w(c(x_i)),w(c(x_o)))$, and that the copy nodes
are necessary and sufficient.

The image of a non-redundant passive form is an independent node
set~$I$. Assume, by contradiction, that the set~$I$ contains
the adjacent nodes $x$~and~$y$. Then $w(c(x_i))=w(c(x_o))$
and $w(c(y_i))=w(c(y_o))$. The edges $x_o\to y_i$ and $y_o\to
x_i$ imply paths $c(x_o)\leadsto c(y_i)$ and $c(y_o)\leadsto
c(x_i)$ which in turn imply that $w(c(x_o))<w(c(y_i))$ and
$w(c(y_o))<w(c(x_i))$. We reached a contradiction, as~desired.

Conversely, every independent node set~$I$ is the image of some
increasing-version non-redundant passive form. We construct
such a non-redundant passive form as follows. Use the mapping
$c(x)=x$. For $x\in I$ set $w(x_o)=w(x_i)=2$; for $x\notin I$ set
$w(x_o)=1$ and $w(x_i)=3$. If two nodes $x$~and~$y$ are adjacent
then we must have $w(x_o)<w(y_i)$ and $w(y_o)<w(x_i)$, which are
true because at least one of the nodes $x$~and~$y$ is not in the
set~$I$.

Moreover, a non-redundant passive form with $n-k$ copy nodes
corresponds to an independent set of size~$k$, where $n$~is the
number of nodes in the original graph. Hence, the copy-optimal
increasing-version passive form corresponds to the maximum
independent node set, which suggests the following theorem. 

\begin{theorem}\index{NP-complete}
The problem of finding a copy-optimal increasing-version passive
form is NP-complete.
\label{th:coiv-is-np-complete}
\end{theorem}

To complete the proof of this theorem we must formulate the
transformation in terms of the corresponding decision problems
and we must show that the decision problem corresponding to
finding copy-optimal increasing-version passive forms is in~NP.

The decision problem corresponding to the mins problem
asks whether there is an independent set of $\ge k$ nodes.
The decision problem corresponding to the problem of
finding a copy-optimal increasing-version passive form
for the restricted family of flowgraphs illustrated in
Figure~\ref{fig:mins-to-copy-flowgraph} asks whether there is a
copy-optimal increasing-version passive form with $\le n-k$ copy
nodes. To transform one into the other we simply make sure we use
the same~$k$.

Also, the problem is in NP because we can check quickly if a
flowgraph is a passive form of another, provided we are also
given the mapping~$c$, the write function~$w$ and the read
function~$r$.

\subsection{The Distinct-Version Case}

The problem seems difficult even if we drop the
increasing-version restriction.

\begin{conjecture}\index{NP-complete}\index{problem!open}
The problem of finding a copy-optimal distinct-version passive
form is NP-complete.
\label{conj:codv-is-np-complete}
\end{conjecture}

The conjecture is that if we enlarge the search space (from
increasing-version passive forms to distinct-version passive
forms) the problem remains just as hard. In general, enlarging
the set of feasible solutions can make an optimization problem
easier or harder. As a trivial example, consider the problem of
finding the maximum of each of these sets:
\begin{align}
S_1(n) &= \{\,2\,\}\\
S_2(n) &= \{\,k\mid \text{$0<k<n$ and $k$~is prime}\,\}\\
S_3(n) &= \{\,k\mid 0<k<n\,\}
\end{align}
It is clear that $S_1(n)\subset S_2(n)\subset S_3(n)$ when $n>2$,
yet finding the maximum of the set~$S_2(n)$ is hardest. On the
other hand, if the input is restricted, then the problem cannot
become harder.

Figure~\ref{fig:ex-special-family-lists} illustrates a family
of flowgraphs for which finding a copy-optimal passive form is
equivalent to certain known classic problems. It is the example
in Figure~\ref{fig:ex-special-family} plus decreasing adjacency
lists added on the right. The flowgraphs in this family consist
of two chains of write only nodes plus some read only nodes
(depicted as dotted edges) that have one parent from the left
chain and one parent from the right chain. The adjacency lists
represent the dotted edges. For simplicity, let us also assume
that all the statements have distinct structures so that the
mapping~$c$ must be injective. Since~$c$ is injective we can be
lazy and write $x$ instead of~$c(x)$.
We shall refer to flowgraphs in this restricted class as
\emph{two-chain flowgraphs}\index{flowgraph!two-chain}.

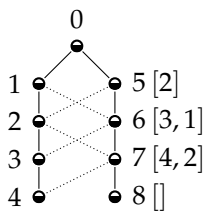
\begin{figure}\centering
\begin{tikzpicture}[scale=0.5]
  \foreach \n/\x/\y/\p/\t in {
    0/1/4/above/,
    1/0/3/left/,
    2/0/2/left/,
    3/0/1/left/,
    4/0/0/left/,
    5/2/3/right/{\>[2]},
    6/2/2/right/{\>[3,1]},
    7/2/1/right/{\>[4,2]},
    8/2/0/right/{\>[]}}
    \wonode (\n) at (\x,\y) [label=\p:$\n\t$] {};
  \foreach \x/\y in {
    0/1,1/2,2/3,3/4,
    0/5,5/6,6/7,7/8}
    \draw (\x)--(\y);
  \foreach \x/\y in {
    5/2,
    6/3,6/1,
    7/4,7/2}
    \draw[densely dotted] (\x)--(\y);
\end{tikzpicture}
\caption{A two-chain flowgraph represented with adjacency lists}
\label{fig:ex-special-family-lists}
\end{figure}

Finding a copy-optimal distinct-version passive form in such a
flowgraph is equivalent to the bipartite matching problem. A
dotted edge is selected if and only if the write versions of its
endpoints are equal. As before, we focus on non-redundant passive
forms, whose copy nodes are in one-to-one correspondence with
non-selected dotted edges. Hence, selecting a maximum of edges
corresponds to inserting a minimum of copy nodes.

{\newcommand{\dotedge}[2]{$#1\cdots#2$}

In the bipartite matching problem, two edges can be selected
simultaneously if and only if they do not share an endpoint.
This is equivalent to the condition that write versions for the
nodes in one chain are distinct. If, by way of contradiction, two
edges \dotedge{x}{z} and \dotedge{y}{z}, which share an endpoint,
are both selected, then $w(x)$ must equal~$w(y)$, which cannot
happen since nodes $x$~and~$y$ are in the same chain. Conversely,
if two edges \dotedge{x}{x'} and \dotedge{y}{y'} do not share
any endpoint then nodes $x$~and~$x'$ can be assigned one write
version and nodes $y$~and~$y'$ a different version. So they can
be both selected.}

According to Section~\ref{sec:passive-algo-back}, the best known algorithms for
the bipartite matching problem work in $O(\min(m\sqrt n, n^{2.38}))$ time,
where $m$~is the number of dotted edges and $n$~is the number of nodes.  A
lower bound would be more useful towards justifying
Conjecture~\ref{conj:codv-is-np-complete}. Unfortunately, none is known. The
equivalence to the bipartite matching problem does tell us, however, that
without finding a better algorithm for this classic problem we cannot hope to
find a copy-optimal distinct-version passive form as fast as we can find a
version-optimal one.

It is also interesting to note that solving two-chain flowgraphs
under the increasing-version restriction is equivalent to
another classic problem, that of finding the longest increasing
subsequence. Again, a dotted edge is selected if and only if
its endpoints have the same write version. This time, the write
versions in the left chain have to be increasing, and so do those
in the right chain. This is equivalent to saying that two dotted
edges can be selected if and only if they do not intersect,
endpoints included.

Now take a look at Figure~\ref{fig:ex-special-family-lists}.
The nodes in the left chain are renumbered so they are also
increasing (like their write versions). The nodes in the right
chain are annotated with the list of nodes that they can reach
through one dotted edge. Each number in those lists represents
a dotted edge. For example, the number~$3$ in the list next
to node~$6$ represents the edge~$3\cdots6$. The task is now
to select at most one number from each adjacency list such
that their sequence increases. Two numbers from the same list
cannot be simultaneously selected because they correspond
to dotted edges that intersect at their right endpoints.
Similarly, non-increasing numbers from different adjacency lists
correspond to intersecting dotted edges. See for example the
edges represented by number~$3$ in the adjacency list of node~$6$
and by number~$2$ in the adjacency list of node~$7$.

If we now concatenate the lists and then extract a longest
increasing subsequence we risk selecting more than one number
from an adjacency list. To make sure this does not happen, before
concatenation, we sort each adjacency list in decreasing order.

Note that all operations used to transform the two-chain
flowgraph into a sequence of numbers (renumbering nodes, sorting
adjacency lists, concatenating them) can be done in linear time.
Note also that computing the inverse transformation can also be
done in linear time. We conclude that the problem of finding a
copy-optimal increasing-version passive form is equivalent to
lis, in the strong sense that asymptotic bounds on the running
time apply to both.

We managed to solve the two-chain family of flowgraphs
faster in the increasing-version case than in the distinct
version case. This is a small piece of evidence supporting
Conjecture~\ref{conj:codv-is-np-complete}. However, keep in mind
that (1)~restricting the search space may make an optimization problem
harder \emph{or} easier and (2)~in a proof we should be looking
at lower bounds instead of best known upper~bounds.

\section{Conclusions}
\label{sec:passivation-pbs}

Passivation follows loop cutting and precedes the computation of
the VC\null. It increases the size of the program but its output
will never lead to exponentially large VCs.

The precise definition of passivation (in
Section~\ref{sec:passive}) enables a formal study of
its properties. The definition leads to four classes
of passive forms that are good, each in its own way
(Section~\ref{sec:passive-types}). The subsequent sections
are concerned with the complexity of finding a passive form:
Version-optimal passive forms are easy to find, copy-optimal
increasing-version passive forms are hard to find. A few open
problems remain.

\begin{problem}
What is the complexity of finding a copy-optimal distinct-version
passive form?
\end{problem}

In other words, settle Conjecture~\ref{conj:codv-is-np-complete}.

\begin{problem}
Find an approximation algorithm for the problem of finding a
copy-optimal increasing-version passive form. It should run in
$o(n^2)$~time to be~practical.
\end{problem}

\begin{problem}\index{problem!open}
Is it always possible to find a distinct-version passive form
that is both version-optimal and copy-optimal? If yes, then how?
If no, then what kind of trade-off can be achieved?
\end{problem}

These problems are clearly inspired by the practice of developing
a program verifier, yet they are very algorithmic and puzzle-like
in nature.

\section{Related Work}
\label{sec:passive.related}

Flanagan and Saxe~\cite{flanagan2001passive} observed that
programs without assignments never lead to exponential VCs.
They proposed passivation to avoid exponential explosion
altogether and implemented it in \escjava. The intermediate
language of \escjava, inspired by Dijkstra's guarded
commands~\cite{dijkstra1975}, does not have a \textbf{goto}
statement, which makes the task of finding a passive form easier.
Their method for VC generation was proved correct in Coq by
Vogels et al.~\cite{vogels2010}.

The passive form used for program verification resembles the DSA
(\fb dynamic \fb single \fb assignment) form used for a long
time in the compiler community to ease optimizations. Sadly, the
literature on compiler optimizations is largely disconnected from
the literature on program verification. For example, in~2007
Vanbroekhoven et al.~\cite{vanbroekhoven2007} identified copy
operations as the key to reducing the size of the DSA\null.
Flanagan and Saxe used copy operations six years earlier. Another
example is the treatment of arrays. In the program verification
area it is standard to model them using a pair of uninterpreted
functions, usually named \textit{update} and \textit{select}.
Once this is done, arrays are handled like all other variables
during passivation.

Barnett and Leino~\cite{barnett2005wpu} describe the VC
generation done by the Boogie tool. The
passivation stage is described informally. Since Boogie has
\textbf{goto} statements, arbitrary flowgraphs are handled.

Flanagan and Saxe~\cite{flanagan2001passive}, Vanbroekhoven
et al.~\cite{vanbroekhoven2007} as well as Barnett and
Leino~\cite{barnett2005wpu} do not give a formal definition
of what it means for a program to be a passive form of
another; Vogels at al.~\cite{vogels2010} does give such
a definition, which was developed independently from the
one used in this dissertation~\cite{grigore2009spu}.
Flanagan and Saxe~\cite{flanagan2001passive}, Barnett and
Leino~\cite{barnett2005wpu}, and Vogels et al.~\cite{vogels2010}
do not define what it means for a passive form to be optimal.
Barnett and Leino~\cite{barnett2005wpu} do discuss optimality,
briefly. Vanbroekhoven et al.~\cite{vanbroekhoven2007} employs
a postprocessing heuristic that reduces the number of copy
operations introduced by their passivation algorithm.

The algorithm presented in Section~\ref{sec:version-optimal}
should be compared with a textbook algorithm for coloring
comparability graphs~\cite{golumbic2004}. Tries, used in
Section~\ref{sec:pass-wpcb} when analyzing the size of the
VC, were introduced by Fredkin~\cite{fredkin1960trie}.
Section~\ref{sec:cc-primer} briefly enumerates a few
basic definitions from computational complexity, adapted
mostly from Goldreich~\cite{goldreich2008}. The importance
of the class of NP-complete problems was established by
Cook~\cite{cook1971}, who proved that there exist an NP-complete
problem (boolean satisfiability), and by Karp~\cite{karp1972},
who proved that $21$ other problems are NP-complete.
Section~\ref{sec:cc-primer} uses the terms `transform' and
`reduce' like Knuth~\cite{knuth1974np}.

The solution to the longest increasing subsequence was
presented in 1977 by Hunt and Szymanski~\cite{hunt1977} and it
exploits a fast data structure presented by van Emde Boas et
al.~\cite{vEB1977} the same year. An $O(m\sqrt n)$ algorithm
for bipartite matching was given by Hopcroft and Karp~\cite{hopcroft1973matching}
in 1973. Faster algorithms for dense graphs are much 
newer (see, for example, Harvey~\cite{harvey2009matching}).

\chapter{Strongest Postcondition versus Weakest Precondition}
\label{ch:spwp}

\chquote{Given a propositional theory~$T$ and a proposition~$q$, a
sufficient condition of~$q$ is one that will make~$q$ true under~$T$,
and a necessary condition of~$q$ is one that has to be true for~$q$
to be true under~$T$. In this paper, we propose a notion of
strongest necessary and weakest sufficient conditions.} {
Fangzhen Lin~\cite{lin2000}}

\noindent The main contributions of this chapter are (1)~a
comparison of the weakest precondition method to the strongest
postcondition method, (2)~links between different ways to define
the semantics of subsets of Boogie (operational semantics, Hoare
logic, weakest precondition, strongest postcondition), and (3)~an
algorithm for unsharing expressions.

\section{Hoare Logic for Core Boogie}
\label{sec:hoare-logic}

After passivation, FreeBoogie generates a VC using either a
method based on the weakest precondition predicate transformer
or on the strongest postcondition predicate transformer. The
relation between these two methods is illuminated by their
relation to a Hoare logic for Boogie.

\subsection{Relation to Operational Semantics}

The \emph{Hoare triple}\index{Hoare triple} $\hoare{p}{$x$}{r}$ means that
``if the store before executing statement~$x$ satisfies predicate~$p$ then
(1)~the execution of statement~$x$ does not go wrong and (2)~if
statement~$x$ is executed then the resulting store satisfies
predicate~$r$.'' (As elsewhere in this dissertation, by `statement' we mean
`simple statement'.)
\begin{equation}
\begin{split}
\hoare{p}{$x$}{r} &= 
  \forall \sigma,\; p\;\sigma \limp \\
  &
  \biggl(
  \forall \frac{h}{\langle\sigma,x\rangle\leadsto\mathit{error}},\; \lnot h
  \biggr) \land \biggl(
  \forall \frac{h}{\langle\sigma,x\rangle\leadsto\langle\sigma',\_\rangle},\;
    h\limp r\;\sigma'
  \biggr)
\end{split}
\label{eq:hoare-def}
\end{equation}
The quantifications over rules are perhaps a little unusual. The first one
$\forall \frac{h}{\langle\sigma,\,x\rangle\leadsto\mathit{error}}$ goes over
all rules that apply to the state~$\langle\sigma,x\rangle$ and go to the
\textit{error} state. The hypothesis~$h$ is bound by the quantifier. The
second quantification $\forall
\frac{h}{\langle\sigma,\,x\rangle\leadsto\langle\sigma',\,\_\rangle}$ is
similar, except that it ranges only over rules that go to
non-\textit{error} states. Both the hypothesis~$h$ and the store~$\sigma$
are bound by the quantifier.  For example, if statement~$x$ is
\textbf{assert}~$q$, then the first quantification is instantiated only for
rule~\eqref{eq:assert-nok-opsem} (on page~\pageref{eq:assert-nok-opsem})
and the second quantification is instantiated only for
rule~\eqref{eq:assume-assert-ok-opsem}.
\begin{equation}
\begin{split}
\hoare{p}{\textbf{assert}~$q$}{r}
  &= \forall\sigma,\; p\;\sigma \limp (q\;\sigma \land (q\;\sigma \limp r\;\sigma))\\
  &= \forall\sigma,\; p\;\sigma \limp (q\;\sigma \land r\;\sigma)\\
  &= \forall\sigma,\; \bigl(p \limp (q\land r)\bigr)\;\sigma\\
  &= |p\limp(q\land r)|
\end{split}
\end{equation}
For the \textbf{assume}
statement the first quantifier over rules evaluates to~$\tru$
since no rule for \textbf{assume} statements goes to the
\textit{error} state.
\begin{equation}
\hoare{p}{\textbf{assume}~$q$}{r} = |(p\land q)\limp r|
\end{equation}

The Hoare triple $\hoare{p}{\textbf{assert}~$q$}{r}$ holds
when the predicate $p\limp(q\land r)$ is valid; the Hoare
triple $\hoare{p}{\textbf{assume}~$q$}{r}$ holds when the
predicate $(p\land q)\limp r$ is valid; and the Hoare triple
$\hoare{p}{$v:=e$}{r}$ holds when the predicate $p\limp(v\gets
e)\;r$ is valid. There is an appealing similarity between the
predicates corresponding to \textbf{assert} and, respectively,
\textbf{assume}.

\subsection{Correctness of a Flowgraph}

A program is correct when none of its executions goes wrong (see
Definition~\ref{def:correctness} on page~\pageref{def:correctness}).  For
flowgraphs we may use the same definition, since each execution of a program
corresponds to an execution of the associated flowgraph. Therefore, one may
check that each execution ends in a non-error state and conclude that the
flowchart is correct.  The task is akin writing an exhaustive test suite.
Unlike automated testers, program verifiers rely on an alternative
characterization of correctness, first described by Floyd~\cite{floyd1967}.

An \emph{annotation}\index{annotation} of a core Boogie flowchart
associates two predicates to each node~$x$, a \emph{precondition}~$a_x$ and
a \emph{postcondition}~$b_x$.  We say that an annotation \emph{explains}
the flowchart when
\begin{enumerate}
\item the precondition~$a_0$ of the initial node is valid,
\item the triple~$\hoare{a_x}{$x$}{b_x}$ holds for all nodes~$x$, and
\item the predicate~$b_x\limp a_y$ is valid for all edges $x\to y$.
\end{enumerate}
We will refer to these as conditions 1,~2, and~3.  A flowgraph may have
zero, one, or more annotations that explain it.

\begin{theorem}\label{th:correctness}
\index{correctness}
Annotations explain only correct Boogie flowgraphs.
\end{theorem}

\begin{proof}
Suppose a Boogie program is explained by preconditions~$a_y$ and
postconditions~$b_y$. We can prove by induction on the length of executions
that $a_y\;\sigma_y$ holds whenever $\langle\sigma_y,y\rangle$ belongs to
some execution.  If $y$ is the initial statement~$0$, then $a_y\;\sigma$
holds because $a_0$~is valid (see condition~1).  Otherwise,
$\langle\sigma_y,y\rangle$ comes after some state $\langle\sigma_x,x\rangle$.
There are multiple such predecessors possible, but for all of them we know,
by induction, that $a_x\;\sigma_x$ holds. The step
$\langle\sigma_x,x\rangle\leadsto \langle\sigma_y,y\rangle$ must be
justified by rule~\eqref{eq:assume-assert-ok-opsem} or by
rule~\eqref{eq:assign-opsem} (see
page~\pageref{eq:assume-assert-ok-opsem}).  In either case, the hypothesis
of the rule is satisfied, so we conclude from~\eqref{eq:hoare-def} that
$b_x\;\sigma_y$ holds (see condition~2).  Hence, $a_x\;\sigma_x$ holds (see
condition~3).

Because $a_x\;\sigma_x$ always holds, condition~2 implies,
via~\eqref{eq:hoare-def}, that no rule leading to the \textit{error} state
applies.
\end{proof}

\section{Predicate Transformers}
\label{sec:predicate-transf}

We now discuss two methods for finding annotations. Both methods are
complete\index{completeness}, in the sense that they find explaining
annotations whenever the Boogie flowgraph is correct.

The strongest postcondition method proceeds forwards: The
precondition~$a_x$ is found before the postcondition~$b_x$ and,
when there is an edge~$x\to y$, the postcondition~$b_x$ is found
before the precondition~$a_y$. The weakest precondition method
proceeds backwards: The postcondition~$b_x$ is found before the
precondition~$a_x$ and, when there is an edge~$x\to y$, the
precondition~$a_y$ is found before the postcondition~$b_x$.

The strongest postcondition method always finds the strongest predicate
that could possibly explain the flowgraph; the weakest precondition method
always finds the weakest predicate that could possibly explain the program.
We say that predicate~$p$ is stronger than predicate~$q$ (and that
predicate~$q$ is weaker than predicate~$p$) when the predicate~$p\limp q$
is valid.

For the rest of this section it helps to think of predicates as
sets of stores. The predicate~$p\land q$ denotes the intersection
set~$p\cap q$; the predicate~$p\lor q$ denotes the union
set~$p\cup q$; the predicate~$\lnot p$ denotes the complement
set~$\bar p$; the predicate~$p\limp q$ denotes the set~$\bar
p\cup q$. The predicate $p\limp q$ is valid when $p\subseteq q$.
Weakest means largest; strongest means smallest.

\subsection{Dealing with Edges}

In the strongest postcondition method, the precondition~$a_y$
is calculated after all postconditions~$b_x$ of nodes~$x$
with edges~$x\to y$. For all nodes~$x$ we must have $b_x\limp
a_y$. The strongest predicate~$a_y$ that is implied by all
predicates~$b_x$ is their~disjunction.
\begin{equation}
a_y \equiv \bigvee_{x\to y} b_x
\label{eq:sp.a_y}
\end{equation}
The precondition of the initial node is not constrained
by condition~3 (on flowgraph edges) but by condition~1.
\begin{equation}
a_0 \equiv \tru
\label{eq:sp.a_0}
\end{equation}
Note that~\eqref{eq:sp.a_0} is not a special case
of~\eqref{eq:sp.a_y}.

In the weakest precondition method, the postcondition~$b_x$ is
calculated after all preconditions~$a_y$ of nodes~$y$ with
edges~$x\to y$. For all nodes~$y$ we must have $b_x\limp
a_y$. The weakest predicate~$b_x$ that that implies all
predicates~$a_y$ is their~conjunction.
\begin{equation}
b_x \equiv \bigwedge_{x\to y} a_y
\label{eq:wp.b_x}
\end{equation}
The postconditions of nodes with no outgoing edge, which correspond to
\textbf{return} statements, are not constrained by condition~3. In fact
they are not constrained by any of the conditions 1, 2, and 3, so for them
$b_x\equiv\tru$, the weakest possible predicate. Fortunately, that is
what~\eqref{eq:wp.b_x} reduces to for \textbf{return} nodes.

\subsection{Dealing with Statements}

\paragraph{Assumptions} If the statement~$x$ is
\textbf{assume}~$q_x$, then according to condition~2
the~predicate
\begin{equation}
(a_x \land q_x) \limp b_x
\end{equation}
must be valid.

In the strongest postcondition method we must find the strongest
predicate~$b_x$, given the predicates $a_x$~and~$q_x$. Clearly,
\begin{equation}
b_x \equiv (a_x\land q_x).
\end{equation}

In the weakest precondition method we must find the weakest
predicate~$a_x$, given the predicates $q_x$~and~$b_x$. In terms
of sets, we must find the biggest set~$a_x$ such that its
intersection with the set~$q_x$ is included in the set~$b_x$.
\begin{equation}
a_x \equiv (q_x \limp b_x)
\end{equation}
The situation is depicted in Figure~\ref{fig:wp.assume}.

\begin{figure}\centering
\begin{tikzpicture}[scale=2,fill=green!30,thick]
\def\circ#1#2{(#1,0.3) #2 circle (0.23)}
\def\ca#1{\circ{0.3}{#1}}
\def\cb#1{\circ{0.7}{#1}}
\def\rect{(0,0) rectangle (1,0.6)}
\fill[even odd rule] \rect \ca{}; \fill \cb{};
\draw \rect; \draw \ca{node{$q_x$}}; \draw \cb{node{$b_x$}};
\end{tikzpicture}
\caption{Weakest precondition of \textbf{assume} statements}
\label{fig:wp.assume}
\end{figure}
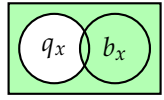

\paragraph{Assertions} If the statement~$x$ is \textbf{assert}~$q_x$, then
according to condition~2 we see that the~predicate
\begin{equation}
a_x \limp (q_x \land b_x)
\label{eq:spwp.assertcond}
\end{equation}
must be valid.

In the strongest postcondition method we must find the strongest
predicate~$b_x$, given the predicates $a_x$~and~$q_x$. No matter
what predicate~$b_x$ we choose, the predicate in~\eqref{eq:spwp.assertcond}
is invalid unless
\begin{equation}
a_x \limp q_x
\end{equation}
is valid. If this is the case, then the strongest
predicate~$b_x$ that satisfies~\eqref{eq:spwp.assertcond}~is
\begin{equation}
b_x \equiv (a_x \land q_x).
\end{equation}
In the weakest precondition method we must find the weakest
predicate~$a_x$, given predicates $q_x$~and~$b_x$. 
\begin{equation}
a_x \equiv (q_x \land b_x)
\end{equation}

\paragraph{Assignment} If statement~$x$ is $v:=e$, then according to
condition~2 we see that the predicate
\begin{equation}
a_x \limp (v \gets e)\;b_x
\label{eq:spwp.assgncond}
\end{equation}
must be valid.

In the strongest postcondition method we must find the strongest
predicate~$b_x$, given the predicate~$a_x$, the variable~$v$, and
the expression~$e$. We rewrite condition~\eqref{eq:spwp.assgncond}.
\begin{equation}
\forall\sigma,\; a_x\;\sigma \limp b_x\;((v\gets e)\;\sigma)
\label{eq:spwp.assgncond2}
\end{equation}
Note that $(v \gets e)\;\sigma=(v\gets e)\;\sigma'$
does not imply $\sigma=\sigma'$. That is, there might be multiple
stores~$\sigma$ corresponding to the same right hand side
in~\eqref{eq:spwp.assgncond2}. Because we want the set~$b_x$ to
be as small as possible we want $b_x\;\sigma$ to hold only if it
must hold, that is, only if there is some corresponding left hand
side in~\eqref{eq:spwp.assgncond2} that~holds.
\begin{equation}
b_x\;\sigma = \exists\sigma',\;a\;\sigma' \land
  \bigl( \sigma=(v\gets e)\;\sigma' \bigr)
\end{equation}

In the weakest precondition method we must find the weakest
predicate~$a_x$, given the variable~$v$, the expression~$e$,
and the predicate~$b_x$. This is trivial.
\begin{equation}
a_x \equiv (v\gets e)\;b_x
\end{equation}

\subsection{Summary}
\label{sec:spwp.summary}
\index{weakest precondition}
\index{strongest postcondition}

In the strongest postcondition method we compute
\begin{align}
a_y &\equiv 
  \begin{cases}
  \tru & \text{if $y$ is the initial statement} \\
  \bigvee_{x \to y} b_x & \text{otherwise}
  \end{cases} \label{eq:spwp.sp.pre}\\
b_x &\equiv
  \begin{cases}
  a_x \land q_x & \text{if $x$ is \textbf{assert}/\textbf{assume}~$q_x$} \\
  \lambda\sigma.\; \exists\sigma',\;a\;\sigma' \land
    \bigl( \sigma=(v\gets e)\;\sigma' \bigr)
    & \text{if $x$ is $v:=e$}
  \end{cases}
  \label{eq:spwp.sp.post}
\end{align}
and we check the validity of
\begin{equation}
\mathit{vc}_\mathit{sp} \equiv \bigwedge_{
  \genfrac{}{}{0pt}{}{\text{$x$ is an}}{\text{assertion}} }
  (a_x \limp b_x). \label{eq:spwp.sp.vc}
\end{equation}

In the weakest precondition method we compute
\begin{align}
b_x &\equiv \bigwedge_{x \to y} a_y &  \label{eq:spwp.wp.post}\\
a_x &\equiv 
  \begin{cases}
  q_x \limp b_x &\text{if $x$ is \textbf{assume}~$q_x$}\\
  q_x \land b_x &\text{if $x$ is \textbf{assert}~$q_x$}\\
  (v\gets e)\;b_x &\text{if $x$ is $v:=e$}
  \end{cases}
  \label{eq:spwp.wp.pre}
\end{align}
and we check the validity of
\begin{equation}
\mathit{vc}_\mathit{wp}\equiv a_0.  \label{eq:spwp.wp.vc}
\end{equation}

Equations \eqref{eq:spwp.sp.post}~and~\eqref{eq:spwp.wp.pre}
can be seen as defining the predicate transformer~\textit{sp} and,
respectively, the predicate transformer~\textit{wp}. Both take two
arguments, a statement and a predicate, and return a predicate.

\begin{example}
Let us look again at the program in Figure~\ref{fig:passive-form}
on page~\pageref{fig:passive-form}. The two methods described
above yield equivalent but structurally different VCs:
\begin{equation}
\begin{split}
\mathit{vc}_\mathit{sp} \equiv
  &\big((\tru\land(v_0=u)\land\lnot\mathit{even}(v_0)\land(v_1=v_0+1))\\
  &\lor(\tru\land(v_0=u)\land\mathit{even}(v_0)\land(v_1=v_0))\big)\\
  &\limp\mathit{even}(v_1)
\end{split}\label{eq:vcsp}
\end{equation}
\begin{equation}
\begin{split}
\mathit{vc}_\mathit{wp} \equiv
  &(v_0=u)\limp\\
  &\big((\lnot\mathit{even}(v_0)\limp(v_1=v_0+1)\limp(\mathit{even}(v_1)\land\tru))\\
  &\land(\mathit{even}(v_0)\limp(v_1=v_0)\limp(\mathit{even}(v_1)\land\tru))\big)
\end{split}
\end{equation}
\end{example}

\section{Replacing Assignments by Assumptions}

After passivation all assignments $v:=e$ are transformed into
assumptions \textbf{assume}~$v=e$. This section argues briefly
why this transformation is sound and complete. The key idea is
that there is a one-to-one correspondence between the executions
of the program without assignments and executions of the program
with assignments where the initial store is chosen so that all
writes are~\hbox{no-operations}.

It is complete: If the program without assignments has
an execution that goes wrong then the program with
assignments has an execution that goes wrong. Say
$\langle\sigma,x_0\rangle$,~$\langle\sigma,x_1\rangle$,
\dots,~$\langle\sigma,x_n\rangle$, \textit{error} is an execution
of the program without assignments. If statement~$x_k$ is
\textbf{assert}/\textbf{assume}~$q_k$ then $q_k\;\sigma$~holds.
In particular, if the statement corresponds to an assignment
$v:=e$ then it is \textbf{assume}~$v=e$, which means that
$\sigma\;v=e\;\sigma$. It is easy to see now that a similar
execution that goes wrong exists in the program with assignments,
because $\sigma=(v\gets e)\;\sigma$.

The proof of the converse is very similar and we omit it.
The idea is to take an execution that goes wrong in the
program with assignments and show that the final store in this
execution can serve as the (unique) store in an execution of the
assignment-free program.

The proofs above essentially rely on the existence of
executions in which no assignment changes the store. Such
executions exist because the program is passive. The
advantages of the assignment-free form are presented in
Section~\ref{sec:boogiesem-wpsp}.

\section{Verification Condition Size}
\label{sec:wpsp.vcsize}

\subsection{VCs with Sharing}

Figure~\ref{fig:sp-algo} shows \eqref{eq:spwp.sp.pre},
\eqref{eq:spwp.sp.post}, and~\eqref{eq:spwp.sp.vc} as an algorithm. It is
assumed that the program contains only \textbf{assert} and \textbf{assume}
statements.  The call $\mathit{Predicate}(x)$ returns the predicate in
statement~$x$. The calls $\mathit{Or}(\mathit{ps})$,
$\mathit{And}(\mathit{ps})$, and $\mathit{Implies}(p,q)$ build predicates
out of simpler ones. They are builders of SMT terms, which implement
hash-consing as described in Section~\ref{sec:design.backend}.  Hence, the
hash-table that implements the set~\textit{bs} may use reference equality
(see method~\textit{Pre}).

\begin{figure}\centering\leavevmode\vbox{
\begin{alg}
\^  $\proc{Pre}(y)$ \comment memoized
\0  ~if~ $y$ is initial
\+    ~return~ $\mathit{True}()$
\0  $\mathit{bs}:=\emptyset$    \comment collects all $b_x$
\=  ~for~ ~each~ predecessor $x$ of $y$
\+    insert $\mathit{Post}(x)$ in the set \textit{bs}
\0  ~return~ $\mathit{Or}(\mathit{bs})$
\end{alg}
\bigskip
\begin{alg}
\^  $\proc{Post}(x)$ \comment memoized
\0  ~return~ $\mathit{And}(\mathit{Pre}(x), \mathit{Predicate}(x))$
\end{alg}
\bigskip
\begin{alg}
\^  $\proc{Vc}()$
\0  $\mathit{vs}:=\emptyset$ \comment collects local VCs
\=  ~for~ ~each~ statement $x$ that is an assertion
\+    insert $\mathit{Implies}(\mathit{Pre}(x), \mathit{Predicate}(x))$ in the set \textit{vs}
\0  ~return~ $\mathit{And}(\mathit{vs})$
\end{alg}}
\caption{VC computation using the strongest postcondition method}
\label{fig:sp-algo}
\end{figure}

\begin{example} 
The VC in~\eqref{eq:vcsp} is represented by the SMT term
in Figure~\ref{fig:vcsp}. Notice the sharing of a subtree.
Enclosed in rectangles are expressions that appear in the
Boogie program (from Figure~\ref{fig:passive-form} on
page~\pageref{fig:passive-form}). These are translated from
the Boogie AST representation to the SMT representation as
explained in Section~\ref{sec:design.backend}. The circled nodes
are created by the algorithm in Figure~\ref{fig:sp-algo} by
calling the builders \textit{True}, \textit{And}, \textit{Or},
and~\textit{Implies}.
\end{example}

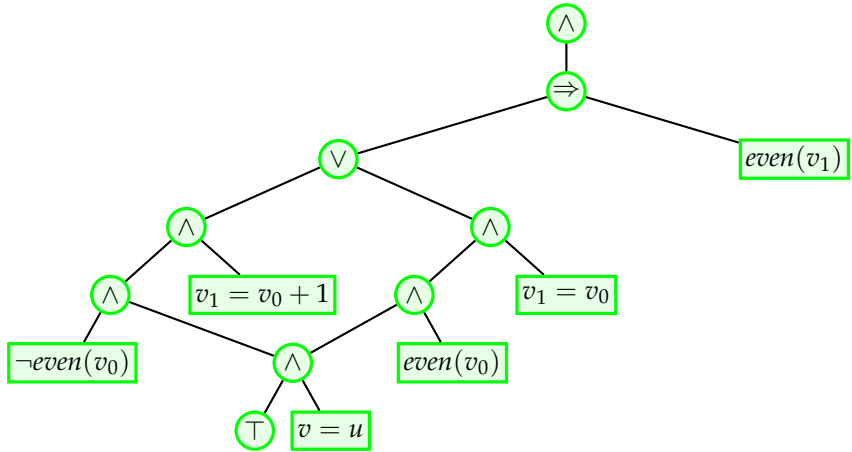
\begin{figure}\centering
\begin{tikzpicture}
  [xscale=.4,yscale=.6,auto,thick,
    every node/.style={predcirc},
    level 2/.style={sibling distance=150mm},
    level 3/.style={sibling distance=100mm},
    level 4/.style={sibling distance=50mm},
    level 5/.style={sibling distance=25mm},
    ]
  \node {$\land$}
    child {node {$\limp$}
      child {node {$\lor$}
        child {node {$\land$}
          child {node (A) {$\land$}
            child {node[predrect] {$\lnot\mathit{even}(v_0)$}}
            child[missing] {}
          }
          child {node[predrect] {$v_1=v_0+1$}}
        }
        child {node {$\land$}
          child {node {$\land$}
            child {node[xshift=-11mm] (B) {$\land$}
              child {node {$\tru$}}
              child {node[predrect] {$v=u$}}
            }
            child {node[predrect] {$\mathit{even}(v_0)$}}
          }
          child {node[predrect] {$v_1=v_0$}}
        }
      }
      child {node[predrect] {$\mathit{even}(v_1)$}}
    };
    \draw (A) -- (B);
\end{tikzpicture}
\caption{Data structure for the VC in~(\ref{eq:vcsp})}
\label{fig:vcsp}
\index{term}
\end{figure}

\begin{remark}
FreeBoogie builds a slightly simpler SMT tree because
builders carry out simplifications. For example, the call
$\mathit{And}(\{p,\tru\})$ returns~$p$.
\label{remark:wpsp.simpl}
\end{remark}

\begin{proposition}
The algorithm in Figure~\ref{fig:sp-algo} computes the VC in
$O(|V|+|E|)$ time. If the program contains no assertions then the
lower bound $\Omega(|V|)$ is~attained.
\label{prop:sptime}
\end{proposition}

\begin{proof}
Because of memoization the methods \textit{Pre} and \textit{Post}
are called at most once for each node. Inserting a term in the
set~\textit{bs} takes $O(1)$~time if the set is implemented as
a hashtable. The method \textit{Pre} analyzes all the incoming
edges of the node it was called for, so in the worst case
it looks once at each edge in the flowgraph. If there are
no assertions then the runtime is dominated by the loop in
\textit{Vc}, which checks there are no assertions.
\end{proof}

Note that the upper bound on the execution time would be
slightly worse if randomized algorithms (such as hashtables) are
forbidden.

\begin{proposition}
The algorithm in Figure~\ref{fig:sp-algo} computes a VC with
size $O(|V|+|E|+n)$, where $n$~is the space needed to represent
all the expressions that appear in the program. If the program
contains no assertions then the lower bound $\Omega(1)$ is
attained.
\label{prop:spspace}
\end{proposition}

\begin{proof}
Each execution of the method \textit{Pre} creates one new node
that contains as many links to children as there are predecessors
in the flowgraph. Each execution of the method \textit{Post}
creates one new node with two children, one of which is an
expression from the program. If there are no assertions 
then $\mathit{vc}_\mathit{sp}\equiv\fls$.
\end{proof}

\begin{remark}
Results similar to Proposition~\ref{prop:sptime} and
Proposition~\ref{prop:spspace} hold for the analogous
implementation of the weakest precondition method.
\end{remark}

\subsection{VCs without Sharing}
\label{sec:wpsp.unshare}

The next step is to send the VC to a theorem prover. One way
is to implement the builder methods (like \textit{True},
\textit{Or},~\dots) in terms of calls to the prover API\null.
Another way is to communicate with the prover through the SMT
language. In the second approach, a data structure is built
inside the VC generator, then it is serialized as a string,
and then the prover parses the string and builds its own data
structures. Hence, the second approach is slower and uses more
memory. Its main advantage is that multiple provers understand
the SMT language. If the VC is printed in the SMT language, then
it could be used as a benchmark for multiple provers. Of course,
another advantage is that the backend of the VC generator does
not have to deal with multiple APIs, some of which might not be
for Java.

SMT terms cannot be printed recursively in the naive way, for
this may result in strings that are exponentially large in the
size of the data structure. Extra work is needed to identify the
shared parts and to eliminate the sharing. Finding which parts
are shared is useful also for deciding how to split the VC into
multiple queries, in case it is too big.

Figure~\ref{fig:elim-sharing} shows a simplified
version of the algorithm used in FreeBoogie. The method
\textit{EliminateSharing} transforms an SMT term, which
represents a predicate, into an equivalent one that has less
sharing. The algorithm plucks sub-predicates repeatedly.
Figure~\ref{fig:plucking} illustrates the plucking process:
It replaces the predicate $(v\gets q)\;p$ by the predicate
$(v=q)\limp p$. Each plucking preserves~validity.
\begin{equation}
|(v\gets q)\;p|=|(v=q)\limp p|
\label{eq:plucking.ok}
\end{equation}
Suppose $(v\gets q)\;p$ is valid and pick a
store~$\sigma$ that satisfies the predicate~$v=q$.
\begin{align}
  &\phantom{\;=\;} (v\gets q)\;\sigma\;u\\
  &= 
    \begin{cases}
    \sigma\;u & \text{if $u$ is $v$} \\
    q\;\sigma & \text{if $u$ is not $v$}
    \end{cases}  && \text{definition of $(v\gets q)$} \\
  &= \sigma\;u && \text{we assumed $\sigma\;v=q\;\sigma$}
\end{align}
So $(v\gets q)\;\sigma=\sigma$.
\begin{align}
  &\phantom{\;=\;} (v\gets q)\;p\;\sigma && \text{we assume $|(v\gets q)\;p|$}\\
  &= p\;((v\gets q)\;\sigma) && \text{definition of $(v\gets q)$} \\
  &= p\;\sigma && \text{because $(v\gets q)\;\sigma=\sigma$}
\end{align}
We have proved
\begin{equation}
|(v\gets q)\;p|\limp|(v=q)\limp p|
\end{equation}
The key to prove the other implication is to note that
$(v\gets q)\;(v=q)$ is valid. Therefore, plucking preserves
validity, even though the predicates $(v\gets q)\;p$ and
$(v=q)\limp p$ do not always describe the same set of stores.
Note also that there is no need to constrain the predicate~$q$ to
be independent of variable~$v$, although this is the case in the
unsharing algorithm.

\begin{example}
\label{ex:plucking}
If $p\equiv (v\limp p_2)$ and $q\equiv p_1$ then the claim is
\begin{equation}
|p_1\limp p_2|=|(v=p_1)\limp(v\limp p_2)|.
\label{eq:plucking-ex}
\end{equation}
\end{example}

The method \textit{Unshare} collects the definitions for all
fresh variables and the method \textit{EliminateSharing} 
uses them to form the result. In other words we need to prove
\begin{equation}
\bigl|(v_1\gets q_1)\;((v_2\gets q_2)\;p)\bigr|
  = \bigl|\bigl((v_1=q_1)\land(v_2=q_2)\bigr)\limp p\bigr|
\end{equation}
which generalizes easily to more fresh variables. Here
is a proof.
\begin{align}
  &\phantom{\;=\;} \bigl|(v_1\gets q_1)\;((v_2\gets q_2)\;p)\bigr| \\
  &= \bigl|(v_1=q_1)\limp((v_2\gets q_2)\;p)\bigr|
    && \text{by \eqref{eq:plucking.ok}} \\
  &= \bigl|(v_2\gets q_2)\;((v_1=q_1)\limp p)\bigr|
    && \text{see below} \\
  &= \bigl|(v_2=q_2)\limp(v_1=q_1)\limp p\bigr|
    && \text{by \eqref{eq:plucking.ok}} \\
  &= \bigl|\bigl((v_1=q_1)\land(v_2=q_2)\bigr)\limp p\bigr|
    && \text{boolean algebra}
\end{align}
The second step of this proof holds only if the
predicate~$q_1$ does not depend on the variable~$v_2$, which is
true if the variable~$v_2$ does not syntactically appear within
the predicate~$q_1$. With $n$~variables, the condition is that
predicate~$q_i$ does not contain variable~$v_j$ if~$i<j$. It is
always possible to find such an ordering since variables, which
are created on line~10 of the method \textit{Unshare}, may only
syntactically contain in their definition variables that were
created earlier.

\begin{figure}\centering\leavevmode\vbox{
\begin{alg}
\0  ~global~ \textit{newDefinitions}  \comment set of $(v\iff t)$
\=  ~global~ \textit{parentCount} \comment term $t$ has $\mathit{parentCount}[t]$ parents
\=  ~global~ \textit{seen}  
\end{alg}
\bigskip
\begin{alg}
\^  $\proc{EliminateSharing}(t)$
\0  clear globals
\=  $\mathit{CountParents}(t)$
\=  $t:=\mathit{Unshare}(t)$
\=  ~return~ $\mathit{Implies}(\mathit{And}(\mathit{newDefinitions}), t)$
\end{alg}
\bigskip
\begin{alg}
\^  $\proc{CountParents}(t)$
\0  ~if~ $t\in\mathit{seen}$
\+    ~return~
\0  insert $t$ in the set \textit{seen}
\=  ~for~ ~each~ child $c$ of $t$
\+    increment $\mathit{parentCount}[c]$
\=    $\mathit{CountParents}(c)$
\end{alg}
\bigskip
\begin{alg}
\^  $\proc{PrintSize}(t)$   \comment memoized
\=  $s:=1$
\=  ~for~ ~each~ child $c$ of $t$
\+    $s:=s+\mathit{PrintSize}(c)$
\0  ~return~ $s$
\end{alg}
\bigskip
\begin{alg}
\^  $\proc{Unshare}(t)$   \comment memoized
\0  ~if~ $t$ or one of its children is not a formula
\+    ~return~ $t$
\0  $\mathit{newChildren}:=\text{empty list}$
\=  ~for~ ~each~ child $c$ of $t$
\+    append $\mathit{Unshare}(c)$ to \textit{newChildren}
\0  $\mathit{t'}:=\mathit{mk}(\mathit{type}(t), \mathit{newChildren})$
\=  $\mathit{ps} := \mathit{PrintSize}(t')$
\=  $\mathit{pc} := \mathit{parentCount}[t]$
\=  ~if~ $\mathit{ps}\times\mathit{pc}-(\mathit{ps}+\mathit{pc})>k$ \comment $k$ is a constant
\+    $v:=\text{fresh variable}$
\=    insert $\mathit{Iff}(v, t')$ in the set \textit{newDefinitions}
\=    $t':=v$
\0  ~return~ $t'$
\end{alg}
}
\caption{(Almost) eliminating sharing}
\label{fig:elim-sharing}
\end{figure}

\begin{figure}\centering
\begin{tikzpicture}[thick]
  \tikzstyle{grn}=[draw=green,fill=green!30]
  \tikzstyle{a}=[arr,shorten >=2pt]
  \tikzstyle{circ}=[circle,minimum size=14pt,inner sep=0pt]

  \def\trit#1{\filldraw[grn] (#1,0)--+(-1,-1)--+(1,-1)--cycle; \node at (#1,-.75) {$q$}}
  \trit{0}; \trit{7};
  \node[grn,circ] (v) at (5,0) {$v$};
  \node[grn,circ] (iff) at (6,1) {$=$};
  \draw[a] (iff) -- (v); \draw[a] (iff) -- (7,0);

  \draw[a] (-.3,.75) -- (0,0); \node at (0,.75) {$\cdots$}; \draw[a] (.3,.75) -- (0,0);
  \draw[a] (4.7,.75) -- (v); \node at (5,.75) {$\cdots$}; \draw[a] (5.3,.75) -- (v);
  \draw[arr,line width=2pt] (1,0)--(4,0);

\end{tikzpicture}
\caption{Plucking}
\label{fig:plucking}
\end{figure}

Apart from being correct, the algorithm in
Figure~\ref{fig:elim-sharing} is also good because it does not
use much time or space.

The method \textit{CountParents} is executed once for each node
in the input (except for line~1, which is executed once for each
edge in the input, and line~2, for which a similar bound holds).
Because it is memoized, the method \textit{Unshare} is also
executed once for each node in the input. For each execution, it
creates at most three new nodes (lines 6,~10, and~11). The method
\textit{PrintSize} is executed at most once for each new node.
Therefore the algorithm runs in linear time and creates at most
three times as many nodes as there are in the input. (Some of the
new nodes `created' on line~6 of the method \textit{Unshare} are
taken from the hash-consing cache.)

The constant~$k$ on line~9 of the method \textit{Unshare}
controls how much plucking is done. A big value reduces the
number of plucks; a small value increases the number of plucks.
The value $\mathit{ps}\times\mathit{pc}$ estimates the length
of printing without plucking, and $\mathit{ps}+\mathit{pc}$
estimates the length of printing with plucking. In particular,
if we pick $k=-1$, then plucking is done if $\mathit{ps}>1$
and $\mathit{pc}>1$. In this case, only nodes without children
(leaves) may have multiple parents in the result. The print size
of such a dag is linear in its size, which in turn is linear
in the size of the input. In other words, if $k=-1$ then the
printing size of $\textit{EliminateSharing}(t)$ is linear in the
size of~$t$.

We have shown the following.

\begin{theorem}
The call $\textit{EliminateSharing}(t)$ returns in $O(n)$~time
and uses $O(n)$ space, where $n$~is the size of the dag~$t$. The
resulting dag represents a predicate that is valid if and only if
the predicate represented by the dag~$t$ is valid. Moreover, when
$k=-1$, the print size of the result is $O(n)$.
\end{theorem}

The implementation in FreeBoogie is slightly more complicated,
but produces queries that the theorem prover Simplify answers
using $30\%$ less time than it needs for the queries produced
by the algorithm described so far. This reduction in time was
measured on the benchmark used in Chapter~\ref{ch:reachability}.

Look again at Example~\ref{ex:plucking} on
page~\pageref{ex:plucking}. Equation~\eqref{eq:plucking-ex} is a
special case of the ``proof by indirect inequality law:''
\begin{equation}
|p_1\limp p_2|=|(v\limp p_1)\limp(v\limp p_2)|.
\end{equation}
It turns out that
\begin{equation}
|(v\gets q)\;p|=|(v\limp q)\limp p|
\label{eq:plucking.ok1}
\end{equation}
holds when the predicate~$p$ is monotonic in the
variable~$v$
\begin{equation}
\bigl|(q_1\limp q_2) \limp \bigl((v\gets q_2)\;p\limp(v\gets q_1)\;p\bigr)\bigr|
\end{equation}
and also
\begin{equation}
|(v\gets q)\;p|=|(q\limp v)\limp p|
\label{eq:plucking.ok2}
\end{equation}
holds when
\begin{equation}
\bigl|(q_1\limp q_2) \limp \bigl((v\gets q_1)\;p\limp(v\gets q_2)\;p\bigr)\bigr|.
\label{eq:plucking.mono}
\end{equation}

FreeBoogie detects monotonicity by examining a syntactic
condition. A predicate~$p$ satisfies~\eqref{eq:plucking.mono} if
variable~$v$ syntactically occurs only in positive positions.
Roughly speaking, a position is positive if it is under an even
number of negations.

\section{Experiments}
\label{sec:wpsp-empiric}

Figure~\ref{fig:exp-time} shows the ratios between the
proving times required for weakest precondition and strongest
postcondition. The tests were ran using the Z3~v1.2
prover on a single core of a dual-core AMD Opteron 2218
($2\times2.6\,\mathrm{GHz}$) with $16\,\mathrm{GiB}$ of DDR2-667
RAM\null. Both methods perform similarly, as can be seen by the
relative symmetry of the graph. Weakest precondition does mildly
better however, as evidenced by the slight shift to the right.
It is interesting to note that in over 40\% of the test cases
the time ratio between the worse method and the better method
is~$>2$. These are the cases outside of the shaded area.

\begin{figure}\centering
 \begin{tikzpicture}[yscale=0.01,xscale=0.4]
   \path[fill=black!20] (-0.6931,0) rectangle (0.6931,300);
   \draw[very thick,draw=blue!70]
(-5.979774,10)--
(-5.551245,6)--
(-5.122715,5)--
(-4.694185,2)--
(-4.265655,8)--
(-3.837125,12)--
(-3.408595,2)--
(-2.980065,5)--
(-2.551535,9)--
(-2.123005,11)--
(-1.694475,5)--
(-1.265945,15)--
(-0.837415,34)--
(-0.408885,128)--
(0.019645,272)--
(0.448175,107)--
(0.876705,54)--
(1.305235,31)--
(1.733765,16)--
(2.162295,16)--
(2.590825,10)--
(3.019355,16)--
(3.447885,15)--
(3.876415,17)--
(4.304945,14)--
(4.733475,4)--
(5.162005,7)--
(5.590534,3)--
+(0,0);
   \draw[->] (-6,0) -- node[yshift=-4mm,below] {$\ln\frac{\text{\textit{sp} time}}{\text{\textit{wp} time}}$} (6,0);
   \draw[->] (0,0) -- (0,300) node[above] {cases};
   \foreach \p in {50,100,...,250} 
     {\draw (-.1,\p)--(.1,\p) node[xshift=-3mm,left] {\p};}
   \foreach \p in {-6,-4,...,6} 
     {\draw (\p,.1)--(\p,-.1) node[below] {\p};}
  \end{tikzpicture}
  \caption{Experimental comparison of proving times}
  \label{fig:exp-time}
\end{figure}
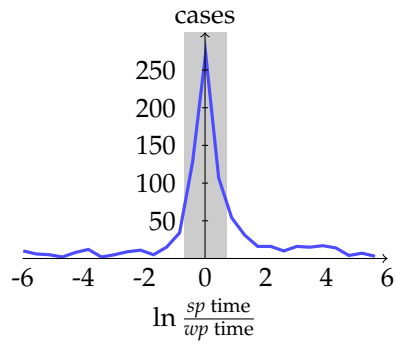

\section{Conclusions}

To reason about core Boogie programs we can use operational
semantics (Chapter~\ref{ch:boogie}), Hoare logic
(Section~\ref{sec:hoare-logic}), and predicate transformers
(Section~\ref{sec:predicate-transf}). Hoare logic plays a central
role, as it is easy to relate both to operational semantics
and to predicate transformers. Hoare logic is also useful to
characterize correctness of Boogie programs in a very intuitive
way (Theorem~\ref{th:correctness}).

The weakest precondition predicate transformer and the strongest
postcondition predicate transformer power two different methods
for generating VCs. The resulting VCs are equivalent (in the
sense that one is valid if and only if the other is), but may
have very different structures. The structure influences the
theorem proving time in unexpected ways and it turns out to be
worth trying both approaches in parallel if two processing units
are available (Section~\ref{sec:wpsp-empiric}).

The generated VC must be sent to a theorem prover. There is a
simple and efficient algorithm for translating the in-memory
representation of the VC into a standard language like SMT
(Section~\ref{sec:wpsp.vcsize}). However, the in-memory
representation is sometimes already too large.

\begin{problem}
Exploit the structure of~$\mathit{vc}_\mathit{wp}$
and~$\mathit{vc}_\mathit{sp}$ in order to split big VCs into
smaller ones.
\end{problem}

The VC produced by the strongest postcondition method seems to be
especially amenable to splitting since it already has the proper
structure---a big conjunction. The main problem is how to best
handle the shared parts of the~conjuncts.

\section{Related Work}
\label{sec:spvswp-related}

\escjava can use both the strongest
postcondition~\cite{flanagan2001passive} and the weakest
precondition method~\cite{leino2005ewp}, but not for arbitrary
acyclic flowgraphs. Note that~\cite{flanagan2001passive}
describes a method analogous to our strongest postcondition
method but only uses the term ``outcome predicate''.
Boogie uses only the weakest precondition
method and can treat unstructured programs~\cite{barnett2005wpu}.

Strongest postcondition was discussed before in the context of a simple
language similar to ours in relation to predicate
abstraction~\cite{flanagan2002pa} (but only for structured programs) and in
relation to proof reuse~\cite{grigore2007ev} (but not shown sound). Tools
that do abstract interpretation~\cite{cousot1977}, software model checkers,
and, in general, tools using some form of symbolic
execution~\cite{king1976} tend to be based on the strongest postcondition.

\chapter{Edit and Verify}
\label{ch:ev}

\chquote{No battle plan survives contact with the enemy.}
{Helmuth Karl Bernhard Graf von Moltke}

\noindent Previous chapters cover various stages of FreeBoogie's
pipeline (Section~\ref{sec:pipeline}). The problem addressed by
this chapter---speeding up the verification condition generation
by exploiting the incremental nature of software development
and verification---requires us to look again at the pipeline as
a whole. The main contributions are (1)~a proof technique for
algorithms that simplify a VC by using an old VC as a reference,
(2)~a heuristic for detecting common parts of two expression
trees, such as two VCs, and (3)~a prototype implementation, which
is part of the SMT solver~\fx.

\section{Motivation}
\label{sec:ev.motivation}

Humans prefer swift tools: A response time~$\gtrsim0.1$~seconds
is noticeable, and a response time~$\gtrsim8$~seconds entices
context switches~\cite{hoover2006}. For example, during a long
compilation a programmer might start browsing the Internet. The
context switch from programming to browsing and back induces a
significant loss of~productivity.

Three techniques are especially useful in improving the
performance of compilers. The first technique, background
compilation, makes the perceived response time smaller than the
processing time. The second technique, modular compilation, makes
sure that time is spent mostly on the modified modules. A module
is typically a file or a class. Modules still need to be linked,
which is done by a second (faster) step for languages like~C and
by the virtual machine at run-time for languages like Java. The
third technique, incremental compilation, is similar to modular
compilation, but with a finer granularity. Typically individual
methods are recompiled separately.

Eclipse uses all three techniques. Background compilation is
implemented in most IDEs; modular compilation is implemented in
most compilers; incremental compilation is not yet implemented in
some widely used compilers, such as~GCC\null.

A core Boogie program corresponds to a procedure implementation
in full Boogie. The correctness of a Boogie procedure
implementation depends only on its preconditions and
postconditions, and on those of the called procedures. It is
easy and natural to verify a Boogie program one procedure
implementation at a time. If we would mirror the compiler
terminology, then we would say that Boogie was designed
for incremental verification. However, the term `modular
verification' is more often used.

Java was designed to be compiled one class at a time; Boogie was
designed to be verified one procedure implementation at a time.
Java now has incremental compilers, which compile parts of a
class separately. Is it possible to verify only parts of a Boogie
procedure implementation? Could this improve the performance of
a program verifier? These are the questions that motivate the
investigations of this chapter.

\section{Overview}
\label{sec:ev.overview}

Figure~\ref{fig:java_evol} shows $16$~annotated programs, one for each
possible choice on whether to include each of the four shaded lines. The
program that includes no shaded line is the original program. Inserting one
shaded line represents an edit operation. In general, an \emph{edit
operation}\index{edit operation} is any change from a type-correct Boogie
program to another type-correct Boogie program. Programmers and IDEs invoke
program verifiers periodically. For each invocation, the frontend
transforms Java into Boogie, the backend transforms Boogie into a VC, and
the theorem prover tries to determine if the VC is valid. Between two
subsequent invocations (on type-correct programs) there is usually not much
change, especially when the program verifier is invoked continuously in the
background by an IDE\null.  Moreover, many transformations are local.
Therefore, much work is duplicated. As an example, consider the following
scenario: A program verifier is run on the original program of
Figure~\ref{fig:java_evol}, then line~1 is added, and then the program
verifier is run again. By adding line~1 nothing essential changes, so the
first and the second run of the program verifier take the same amount of
time. It should be possible to do much better.

{\def\hl{\rlap{\hbox{\color{verylightgray}\strut\vrule width12em}}%
}
\bc
\begin{jml}[escapechar=\"]
"\hl"// comment
abstract class Day {
  public abstract int getMonth();
    ensures 1 <= result && result <= 12;

  public abstract int getYear();
    ensures 1970 <= result;
    "\hl"ensures result <= 2038;

  public abstract int getDay();
    ensures 1 <= result && result <= 31;

  public int dayOfYear() 
    ensures 1 <= result;
    "\hl"ensures result <= 366;
  {
    int offset = 0;
    if (getMonth() > 1) offset += 31;
    if (getMonth() > 2) offset += 28;
    if (getMonth() > 3) offset += 31;
    if (getMonth() > 4) offset += 30;
    if (getMonth() > 5) offset += 31;
    if (getMonth() > 6) offset += 30;
    if (getMonth() > 7) offset += 31;
    if (getMonth() > 8) offset += 31;
    if (getMonth() > 9) offset += 30;
    if (getMonth() > 10) offset += 31;
    if (getMonth() > 11) offset += 30;
    boolean isLeap = getYear() 
             (getYear() 
    "\hl"assert offset <= 334;
    if (isLeap && getMonth() > 2) offset++;
    return offset + getDay();
  }
}
\end{jml}
\ec{Typical evolution of annotated Java code}{fig:java_evol}

FreeBoogie consists almost entirely of evaluators, some of which
are transformers (see Section~\ref{sec:design.ast}). The input
of an evaluator is a Boogie AST\null. Because Boogie ASTs are
immutable, they can be used as keys in a map (dictionary) data
structure. That is why the class \textit{Evaluator}, the base
class of all evaluators, can cache the results of all evaluations
done in FreeBoogie. For example, if the type-checker is invoked
twice on the AST of a procedure implementation, then the second
invocation returns quickly after one hashtable lookup.

Two problems remain. First, if the input file is parsed twice,
then two distinct ASTs are created by the parser, thus rendering
the caches unusable. Second, even if the VC is obtained faster,
the component that takes too much time is usually the prover. Let
us consider these problems in turn.

To not produce the same AST twice when line~1 is inserted we
must not parse the input file, in its entirety, twice. We could
either re-parse only the parts that change or we could not parse
at all. Both solutions work. The easiest and fastest way to
detect what parts changed is to integrate FreeBoogie's parser in
a text editor, such as the text editor of Eclipse. Parsing is
unnecessary when the frontend uses FreeBoogie's API\null.

To reduce the prover time without re-engineering it, we must
simplify the VC\null. For example, if the VC turns out to be
the same as a previous one, for which the prover was called,
then there is nothing more to do. This is the case for inserting
line~1. In particular, if the VC before change is valid, then the
VC after change simplifies to the predicate~$\tru$. Any prover
should return quickly when the query is the predicate~$\tru$.

\section{Simplifying SMT Formulas}
\label{sec:ev.simplify}

We want to simplify a VC, given a similar VC known to be valid.
\begin{problem}\label{pb:ev}
Find a predicate transformer \textit{prune}\index{prune} such that
\[ |\lnot p| \limp \bigl(|\lnot q|=|\lnot(\mathit{prune}\;p\;q)|\bigr)\]
for all predicates $p$~and~$q$.
\end{problem}
Here $\lnot p$~is the VC known to be valid, $\lnot q$ is the current VC,
and $\lnot(\mathit{prune}\;p\;q)$ is the simplified form of the current
VC\null.  It is not essential to work with the negated form of VCs, but it
is closer to the implementation. SMT solvers answer the validity query
$\lnot p$ by reducing it to the satisfiability query~$p$.
\begin{example}
For the program 
\boogieCode|var y:int;|
\boogieCode+assume (forall x:int :: even(x) || odd(x));+
\boogieCode|assume !odd(y);|
\boogieCode|assert even(x);|
we have
\begin{equation}
\lnot\mathit{vc}_\mathit{sp}\equiv
  \bigl(\forall x,\; \mathit{even}(x)\lor\mathit{odd}(x)\bigr)
  \land\lnot\mathit{odd}(y)
  \land\lnot\mathit{even}(y)
\end{equation}
The prover sees that
\begin{align}
&\phantom{\;\equiv\;}
  \lnot\mathit{vc}_\mathit{sp} \\
&\limp
  \bigl(\mathit{even}(y)\lor\mathit{odd}(y)\bigr)
  \land\lnot\mathit{odd}(y)
  \land\lnot\mathit{even}(y)
  &&\text{by setting $x:=y$} \\
&=
  \fls
\end{align}
\end{example}

\begin{remark}
We also want the satisfiability query $\mathit{prune}\;p\;q$ to be easier
to handle than the query~$q$, by typical SMT solvers. Alas, this
requirement is hard to formalize. The size of the predicate representation
is one possible proxy: It should be as small as possible.
\end{remark}

\begin{example}
If we restrict $p\equiv q$ in Problem~\ref{pb:ev}, then
$(\mathit{prune}\;p\;q)\equiv\fls$ is a good~solution.
\end{example}

\begin{example}
Table~\ref{tbl:simpl_ex} shows how \textit{prune} could simplify a
particular VC\null.  To make this example more concrete, the reader might
wish to plug in $p_k\equiv(v>-k)$. Note that $p_1\limp p_3$ is a smaller
feasible simplified VC, but likely harder to prove than $(p_1\land
p_2)\limp p_3$.
\end{example}

\begin{table}\centering
\caption{Example of simplification}
\label{tbl:simpl_ex}
\begin{tabular}{rccc}
\toprule
  & Initial & Modified & Simplified \\ \midrule
&
\vtop{
  \hbox{\strut\textbf{assume} $p_1$}
  \hbox{\strut\textbf{assert} $p_2$}} &
\vtop{
  \hbox{\strut\textbf{assume} $p_1$}
  \hbox{\strut\textbf{assert} $p_2$}
  \hbox{\strut\textbf{assert} $p_3$}} &
\vtop{
  \hbox{\strut\textbf{assume} $p_1$}
  \hbox{\strut\textbf{assume} $p_2$}
  \hbox{\strut\textbf{assert} $p_3$}} \\

$\mathit{vc}_\mathit{sp}$ &
$p_1\limp p_2$ &
$(p_1\limp p_2)\land\bigl((p_1\land p_2) \limp p_3\bigr)$ &
$(p_1\land p_2)\limp p_3$ \\

$\lnot\mathit{vc}_\mathit{sp}$ &
$p_1\land\lnot p_2$ &
$(p_1\land\lnot p_2) \lor(p_1\land p_2\land\lnot p_3)$ &
$p_1\land p_2\land\lnot p_3$ \\
\bottomrule
\end{tabular}
\end{table}

To characterize a wide class of predicate transformers that satisfy the
requirements of Problem~\ref{pb:ev}, it helps to introduce the notion of
interpolant.
\begin{definition}[interpolant]\index{interpolant}
We say that predicate~$r$ is an \emph{interpolant} of predicates
$p$~and~$q$ when both $p\limp r$ and $r\limp q$ hold.
\end{definition}
\begin{remark}
In alternative definitions, predicate~$r$ is said to be an interpolant 
of predicates $\lnot p$~and~$q$, and it is constrained to contain only
symbols that appear in both predicates $p$~and~$q$.
\end{remark}

\begin{theorem}\label{th:ev.gen}
If predicate~$r$ is an interpolant of predicates $\lnot p\land q$
and~$q$, then $|\lnot p|\limp\bigl(|\lnot q|=|\lnot r|\bigr)$.
\end{theorem}

\begin{proof}
We will need three simple facts about predicates, which follow
from~\eqref{eq:valid_notation} on page~\pageref{eq:valid_notation}.
\begin{align}
&\bigl|a=b\bigr|\limp\bigl(|a|=|b|\bigr) \label{eq:valid_eq} \\
&\bigl(|a|\land|b|\bigr)=\bigl|a\land b\bigr| \label{eq:valid_conj} \\
&\bigl(|a|\land|a\limp b|\bigr)\limp\bigl|b\bigr| \label{eq:valid_deduction}
\end{align}
The hypotheses are
\begin{align}
&\phantom{\;=\;}
  \bigl|(\lnot p\land q)\limp r\bigr| \\
&=
  \bigl|\lnot p\limp(q\limp r)\bigr| \label{eq:ev.h1} \\
\intertext{and}
&\phantom{\;=\;}
  \bigl|r\limp q\bigr|. \label{eq:ev.h2}
\end{align}
Now we calculate.
\begin{align}
&\phantom{\;=\;}
  |\lnot q| = |\lnot r| \\
&\Leftarrow
  |\lnot q= \lnot r| 
  &&\text{by \eqref{eq:valid_eq}} \\
&=
  \bigl|(q\limp r)\land(r\limp q)\bigr| \\
&=
  |q\limp r| \land |r\limp q|
  &&\text{by \eqref{eq:valid_conj}} \\
&=
  |q\limp r|
  &&\text{by \eqref{eq:ev.h2}} \\
&\Leftarrow
  |\lnot p|
  &&\text{by \eqref{eq:ev.h1} and \eqref{eq:valid_deduction}}
\end{align}
\end{proof}

Theorem~\ref{th:ev.gen} tells us that any algorithm for computing
interpolants is potentially useful for speeding up incremental
verification. This is significant because interpolants have other
uses~\cite{jain2008,mcmillan2005,jhala2006,mcmillan2008}.  It is
interesting, for example, that any technique developed for inferring loop
invariants is potentially useful for speeding up incremental verification.

\subsection{Pruning of Predicates}
\label{sec:ev.pruning}

Problem~\ref{pb:ev} has the simple solution $\mathit{prune}\;p\;q\equiv q$,
which means ``forget the old VC\null.'' Better solutions explore the
syntactic structure of $p$~and~$q$. Remember that the strongest
postcondition method produces a conjunction of implications
(see~\eqref{eq:spwp.sp.vc} on page~\pageref{eq:spwp.sp.vc}). In this case,
the following equations work well in~practice.
\begin{align}
&\mathit{prune}\;(p\lor r)\;(r\land q)\equiv\fls \label{eq:ev.prune1} \\
&
\begin{aligned}
&\mathit{prune}\;
  \bigl((p_1\land r_1)\lor\ldots\lor(p_m\land r_m)\bigr)\;
  \bigl((r_1\land\ldots\land r_m)\land(q_1\land\ldots\land q_n)\bigr) \\
&\quad\equiv(r_1\land\ldots\land r_m) \\
&\quad\quad\land(\mathit{prune}\;(p_1\lor\ldots\lor p_m)\;q_1)
  \land\ldots\land(\mathit{prune}\;(p_1\lor\ldots\lor p_m)\;q_n)
\end{aligned} \label{eq:ev.prune2} \\
&\mathit{prune}\;p\;(q_1\lor\ldots\lor q_n) 
  \equiv
  (\mathit{prune}\;p\;q_1) \lor\ldots\lor
  (\mathit{prune}\;p\;q_n) 
  \label{eq:ev.prune3} \\
&\mathit{prune}\;p\;q\equiv q \label{eq:ev.prune4}
\end{align}
Algorithms use these equations as left-to-right rewrite rules.
Roughly, they are tried in order, but
\eqref{eq:ev.prune2}~and~\eqref{eq:ev.prune3} are skipped when they reduce
to identity. For example, \eqref{eq:ev.prune3} is applied only if $n\ne 1$.

\begin{example}
The insertion of an assertion, which is illustrated in
Table~\ref{tbl:simpl_ex}, is handled as follows.
\begin{align}
&\phantom{\;\equiv\;}
  \mathit{prune}
    \;\lnot(p_1\limp p_2)
    \;\lnot\bigl((p_1\limp p_2)\land((p_1\land p_2)\limp p_3)\bigr) \\
&\equiv
  \mathit{prune}
    \;(p_1\land\lnot p_2)
    \;\bigl((p_1\land p_2)\lor(p_1\land p_2\land\lnot p_3)\bigr) \\
&\begin{aligned}\equiv
  &\bigl(\mathit{prune}\;(p_1\land\lnot p_2)\;(p_1\land\lnot p_2)\bigr)\\
  &\quad\lor
  \bigl(\mathit{prune}
    \;(p_1\land\lnot p_2)
    \;(p_1\land p_2\land\lnot p_3)
  \bigr)
  \end{aligned}
  &&\text{by \eqref{eq:ev.prune3}} \\
&\equiv
  \mathit{prune}\;(p_1\land\lnot p_2)\;(p_1\land p_2\land\lnot p_3)
  &&\text{by \eqref{eq:ev.prune1}} \\
&\equiv
  p_1\land\bigl(\mathit{prune}\;\lnot p_2\;(p_2\land\lnot p_3)\bigr)
  &&\text{by \eqref{eq:ev.prune2}} \label{eq:ev.prune.useless} \\
&\equiv
  p_1\land p_2\land\lnot p_3
  &&\text{by \eqref{eq:ev.prune4}} \\
&\equiv
  \lnot\bigl((p_1\land p_2)\limp p_3\bigr)
\end{align}
Step~\eqref{eq:ev.prune.useless} is useless, but \eqref{eq:ev.prune2} is
useful in other examples.
\end{example}

\begin{example}
Consider now the case of a program with two assertions out of which one is
modified.
\begin{align}
&\begin{aligned}\phantom{\;\equiv\;}
  \mathit{prune}\;
    &\lnot\bigl((p_1\limp q_1)\land(p_2\limp q_2\land q_3)\bigr) \\
    &\lnot\bigl((p_1\limp q_1)\land(p_2\limp q_2\land q_4)\bigr) 
  \end{aligned}
  \\
&\begin{aligned}\equiv
  \mathit{prune}\;
    &\bigl((p_1\land\lnot q_1)\lor(p_2\land(\lnot q_2\lor\lnot q_3))\bigr)\\
    &\bigl((p_1\land\lnot q_1)\lor(p_2\land(\lnot q_2\lor\lnot q_4))\bigr)
  \end{aligned}
  \\
&\begin{aligned}\equiv\;
  &\Bigl(
  \mathit{prune}\;
    \bigl((p_1\land\lnot q_1)\lor(p_2\land(\lnot q_2\lor\lnot q_3))\bigr)\;
    (p_1\land\lnot q_1)
  \Bigr)\\
  &\quad\begin{aligned}\lor
  \Bigl(
  \mathit{prune}\;
    &\bigl((p_1\land\lnot q_1)\lor(p_2\land(\lnot q_2\lor\lnot q_3))\bigr)\\
    &(p_2\land(\lnot q_2\lor\lnot q_4))
  \Bigr)
  \end{aligned}
  \end{aligned}
  &&\text{by \eqref{eq:ev.prune3}} \\
&\begin{aligned}\equiv
  \mathit{prune}\;
    &\bigl((p_1\land\lnot q_1)\lor(p_2\land(\lnot q_2\lor\lnot q_3))\bigr)\\
    &\bigl(p_2\land(\lnot q_2\lor\lnot q_4)\bigr)
  \end{aligned}
  &&\text{by \eqref{eq:ev.prune1}} \\
&\equiv
  p_2\land\Bigl(
    \mathit{prune}\;
      \bigl((p_1\land\lnot q_1)\lor\lnot q_2\lor\lnot q_3\bigr)\;
      \bigl(\lnot q_2\lor\lnot q_4\bigr)
  \Bigr)
  && \text{by \eqref{eq:ev.prune2}} \\
&\begin{aligned}\equiv
  p_2\land\Bigl(
    &\Bigl(
      \mathit{prune}\;
        \bigl((p_1\land\lnot q_1)\lor\lnot q_2\lor\lnot q_3\bigr)\;
        \lnot q_2
    \Bigr) \\
    \lor&\Bigl(
      \mathit{prune}\;
        \bigl((p_1\land\lnot q_1)\lor\lnot q_2\lor\lnot q_3\bigr)\;
        \lnot q_4
    \Bigr)
  \Bigr)\end{aligned}
  && \text{by \eqref{eq:ev.prune3}} \\
&\equiv
  p_2\land\Bigl(
    \mathit{prune}\;
      \bigl((p_1\land\lnot q_1)\lor\lnot q_2\lor\lnot q_3\bigr)\;
      \lnot q_4
  \Bigr)
  && \text{by \eqref{eq:ev.prune1}} \\
&\equiv
  p_2\land\lnot q_4
  && \text{by \eqref{eq:ev.prune4}} \\
&\equiv
  \lnot(p_2\limp q_4)
\end{align}
In this calculation \eqref{eq:ev.prune2} plays an essential role.
\end{example}

For the proof that \eqref{eq:ev.prune1}--\eqref{eq:ev.prune4} solve
Problem~\ref{pb:ev} it is convenient to introduce some shorthand notations.
The range of~$i$ is implicitly $1.\,.\,m$ and $(\circ i,\;p_i)$ stands for
$p_1\circ\cdots\circ p_m$; similarly, the range of~$j$ is implicitly
$1.\,.\,n$ and $(\circ j,\;p_j)$ stands for $p_1\circ\cdots\circ p_n$.
Equations \eqref{eq:ev.prune2}~and~\eqref{eq:ev.prune3} become shorter.
{\def\x{\mathit{prune}\;p\;(\lor j,\;q_j)}
\begin{align}
&\begin{aligned}
&\mathit{prune}\;
  (\lor i,\;p_i\land r_i)\;((\land i,\;r_i)\land(\land j,\;q_j)) \\
  &\phantom{\x}\equiv 
  (\land i,\;r_i)\land(\land j,\;\mathit{prune}\;(\lor i,\;p_i)\;q_j)
\end{aligned} \label{eq:ev.prune2.short} \\
&\x\equiv(\lor j,\;\mathit{prune}\;p\;q_j) \label{eq:ev.prune3.short}
\end{align}}

\begin{lemma}
\label{lemma:ev1}
The predicate $\mathit{prune}\;p\;q$ given
by~\eqref{eq:ev.prune1}--\eqref{eq:ev.prune4} is weaker than the
predicate~$\lnot p\land q$.
\begin{equation}
|(\lnot p\land q)\limp(\mathit{prune}\;p\;q)| \label{eq:ev.lemma2}
\end{equation}
\end{lemma}

\begin{proof}
The proof is by induction on the total size of the two arguments.

\noindent Branch~\eqref{eq:ev.prune1}:
\begin{align}
&\phantom{\;=\;}
  \lnot(p\lor r) \land (r\land q) \\
&= 
  \fls \\
&\limp
  \mathit{prune}\;(p\lor r)\;(r\land q)
\end{align}
Branch~\eqref{eq:ev.prune2} is more interesting:
\begin{align}
&\phantom{\;=\;}
    \lnot(\lor i,\; p_i\land r_i) \land
    (\land i,\; r_i) \land
    (\land j,\; q_j) \\
&\limp
    (\land i,\; r_i) \land
    \lnot(\lor i,\; p_i) \land
    (\land j, q_j) \\
&=
    (\land i,\; r_i) \land
    (\land j,\; \lnot(\lor i,\; p_i) \land q_j) \\
&\limp
    (\land i,\; r_i) \land
    (\land j,\; \mathit{prune}\;(\lor i,\; p_i)\;q_j)
  &&\text{by induction} \\
&=
  \mathit{prune}\;
    (\lor i,\; p_i\land r_i)\;
    ((\land i,\; r_i)\land(\land j,\; q_j))
  &&\text{by \eqref{eq:ev.prune2.short}}
\end{align}
Branch~\eqref{eq:ev.prune3}:
\begin{align}
&\phantom{\;=\;}
  \lnot p \land (\lor j,\; q_j) \\
&=
  (\lor j,\; \lnot p\land q_j) \\
&\limp
  (\lor j,\; \mathit{prune}\;p\;q_j)
  &&\text{by induction} \\
&=
  \mathit{prune}\;p\;(\lor j,\; q_j)
  &&\text{by \eqref{eq:ev.prune3.short}}
\end{align}
Branch~\eqref{eq:ev.prune4} is immediate.
\end{proof}

\begin{lemma}
\label{lemma:ev2}
The predicate $\mathit{prune}\;p\;q$ given
by~\eqref{eq:ev.prune1}--\eqref{eq:ev.prune4} is stronger
than the predicate~$q$.
\begin{equation}
|(\mathit{prune}\;p\;q)\limp q| \label{eq:ev.lemma1}
\end{equation}
\end{lemma}

\begin{proof}
The proof is by induction on the total size of the two arguments.

\noindent Branch~\eqref{eq:ev.prune1}:
\begin{align}
&\phantom{\;=\;}
  \mathit{prune}\;(p\lor r)\;(r\land q) \\
&=
  \fls
  &&\text{by \eqref{eq:ev.prune1}} \\
&\limp 
  p\land r
\end{align}
Branch~\eqref{eq:ev.prune2}:
\begin{align}
&\phantom{\;=\;}
  \mathit{prune}\;
    (\lor i,\; p_i\land r_i)\;
    \bigl((\land i,\; r_i)\land (\land j,\; q_j) \bigr) \\
&=
  (\land i,\;r_i)\land(\land j,\;\mathit{prune}\;(\lor i,\;p_i)\;q_j)
  &&\text{by \eqref{eq:ev.prune2.short}}  \\
&\limp
  (\land i,\; r_i) \land (\land j,\; q_j)
  &&\text{by induction}
\end{align}
Branch~\eqref{eq:ev.prune3}:
\begin{align}
&\phantom{\;=\;}
  \mathit{prune}\;p\;(\lor j,\;q_j) \\
&=
  (\lor j,\; \mathit{prune}\;p\;q_j)
  &&\text{by \eqref{eq:ev.prune3.short}} \\
&\limp
  (\lor j,\; q_j)
  &&\text{by induction}
\end{align}
Branch~\eqref{eq:ev.prune4} is immediate.
\end{proof}

\begin{theorem}\label{th:ev}
The predicate transformer \textit{prune} defined
by~\eqref{eq:ev.prune1}--\eqref{eq:ev.prune2} solves Problem~\ref{pb:ev} on
page~\pageref{pb:ev}.
\end{theorem}

\begin{proof}
Immediate from Lemmas \ref{lemma:ev1}~and~\ref{lemma:ev2} and 
Theorem~\ref{th:ev.gen}.
\end{proof}

\subsection{Algorithm}

To analyze the algorithm implied
by~\eqref{eq:ev.prune1}--\eqref{eq:ev.prune4} we need a more precise
description; to give a more precise description it helps to define a few
basic operations. We assume that SMT terms are dags. Given a term~$t$ we
may ask for its label $\mathit{Label}(t)$ and for the list
$\mathit{Children}(t)$ of its children.  Possible labels include
$\land$~and~$\lor$.  Given a set~$T$ of terms, we may obtain a term
$\mathit{And}(t)$ and a term $\mathit{Or}(t)$ representing the conjunction
and, respectively, the disjunction of $T$'s elements. Also,
$\mathit{False}()$ returns a representation of~$\fls$.  We assume that all
SMT terms are created using hash-consing, including those returned by
$\mathit{And}(\cdot)$, $\mathit{Or}(\cdot)$, and $\mathit{False}()$. In
particular, $\mathit{False}()$ returns the same reference every time.
Moreover, we assume that the hash-consing mechanism is aware that
$\land$~and~$\lor$ are commutative: $\mathit{And}(S)$ and $\mathit{And}(T)$
return the same reference when the sets $S$~and~$T$ are equal; similarly
for~$\mathit{Or}(\cdot)$. Instead of relying on the hash-consing mechanism
to handle associativity, we treat it explicitly using the basic operation
$\mathit{Flatten}(\cdot,\cdot)$.
\begin{equation*}
\vbox{\footnotesize
\begin{alg}
\^  $\proc{Flatten}(\mathit{op}, t)$ \comment memoized
\=  ~if~ $\mathit{Label}(t)\ne\mathit{op}$
\+    ~return~ $\{t\}$
\-  ~else~
\+    ~return~ $\bigcup\{\mathit{Flatten}(\mathit{op},u)\mid u\in\mathit{Children}(t)\}$
\end{alg}}
\end{equation*}
(This approach leads to fewer special cases, essentially because we can
treat any term~$t$ as $\mathit{And}(\{t\})$ or as $\mathit{Or}(\{t\})$,
without being explicit about it.) The last line of \textit{Flatten} uses
somewhat non-standard notations for set operations. It is important to
understand what they mean and what is their run-time overhead. An easy way
to do so is to show a possible implementation in Java. The comprehension
$\{e(u)\mid u\in s\}$ stands for the call $\mathit{mapE}(s)$, where
\textit{mapE} is defined as follows.
\begin{jml}
static HashSet<T> mapE(Iterable<T> p) {
  HashSet<T> r = new HashSet<T>();
  for (T u : S) r.add(e(u));
  return r;
}
\end{jml}
Note that the argument may be a list or a set. Here $e(u)$ is some
expression not involving side-effects.  Because sets are represented with
\textit{HashSet}s, insertion takes constant time. Similarly, $\cup s$
stands for the call $\mathit{union}(s)$, where \textit{union} is defined
as follows.
\begin{jml}
static HashSet<T> union(HashSet<HashSet<T>> s) {
  HashSet<T> r = new HashSet<T>();
  for (HashSet<T> u : s) for (T v : u) r.add(v);
  return r;
}
\end{jml}

\begin{figure}\centering\leavevmode\vbox{
\begin{alg}
\^  $\proc{Prune}(p,q)$ \comment memoized
\0  $q_\lor := \mathit{Flatten}(\lor, q)$ \comment $\{q_1,\ldots,q_n\}$ in \eqref{eq:ev.prune3}
\=  $q_\land := \mathit{Flatten}(\land, q)$
\=  $p_\lor := \mathit{Flatten}(\lor, p)$
\=  $p_{\lor\land} := \{ \mathit{Flatten}(\land, u) \mid u\in p_\lor\}$
\=  ~if~ $p_\lor\cap q_\land\ne\emptyset$ \comment apply \eqref{eq:ev.prune1}
\+    ~return~ $\mathit{False}()$
\-  $r:=\cup\{ u\cap q_\land \mid u\in p_{\lor\land} \}$ \comment $\{r_1,\ldots,r_m\}$ in \eqref{eq:ev.prune2}
\=  ~if~ $r\ne\emptyset$ \comment apply \eqref{eq:ev.prune2}
\+    $p':=\mathit{Or}(\{\mathit{And}(u-r)\mid u\in p_{\lor\land}\})$ \comment $p_1\lor\ldots\lor p_m$ in \eqref{eq:ev.prune2}
\=    ~return~ $\mathit{And}(r\cup\{\mathit{Prune}(p', u)\mid u\in q_\land-r\})$
\-  ~if~ $|q_\lor|>1$ \comment apply \eqref{eq:ev.prune3}
\+    ~return~ $\mathit{Or}(\{\mathit{Prune}(p,u)\mid u\in q_\lor\})$
\-  ~return~ $q$ \comment apply \eqref{eq:ev.prune4}
\end{alg}}
\caption{Algorithm for pruning SMT trees}\label{fig:ev.alg}
\end{figure}

It is now easy to convert the algorithm in Figure~\ref{fig:ev.alg} into a
Java program; it is yet unclear for how long the program will run.

\begin{proposition}
In the worst case, the algorithm in Figure~\ref{fig:ev.alg} takes at least
exponential time, when the size of the input is defined to be the space
occupied by the representation of predicates $p$~and~$q$.
\end{proposition}

\begin{proof}
Consider the predicates $p_0$, \dots, $p_n$ defined on variables $u_1$,
\dots,~$u_n$ and $v_1$, \dots,~$v_n$ as follows.
\begin{align}
p_0 &\;\equiv\;
  (u_1\land\ldots\land u_n)\lor(v_1\land\ldots\land v_n) \\
p_{k+1} &\;\equiv\; 
  (u_{k+1}\land p_k)\lor(p_k\land v_{k+1})
  &&\text{for $k\in[\,0.\,.\,n)$}
\end{align}
These formulas describe a dag in which $p_k$~is shared by the two disjuncts
of~$p_{k+1}$. The dag has $3(n+1)$~nodes for operators, $2n$~nodes for
variables, $2+2n$~edges for predicate~$p_0$, and $6n$~edges for the other
predicates.  The total size of the input is~$13n+5$.

The algorithm will construct extra predicates that have a shape similar to
predicate~$p_0$, but lack some variables at the tail of the conjuncts.
\begin{multline}
q_{k,S} = 
  (\land i \in [1.\,.\,k]\cup S,\; u_i) \lor 
  (\land i \in [1.\,.\,n]-S,\; v_i) \\
  \text{for $k\in[1.\,.\,n]$ and $S\subseteq(k.\,.\,n]$}
\end{multline}
In particular, $q_{n,\emptyset}$ is $p_0$. The call
$\mathit{Prune}(p_0,p_n)$ will trigger the execution of the body of method
\textit{Prune} with arguments $(q_{k,S},p_k)$ for all $k\in[0.\,.\,n]$
and $S\subseteq(k.\,.\,n]$.

This can be seen by induction. When $k=n$ the set~$S$ must be empty and the
arguments correspond to the initial call. When $k\in[\,0.\,.\,n)$ we have the
following call tree, for all $S\subseteq(k+1.\,.\,n]$.
\begin{align}
&\mathit{Prune}(q_{k+1,S},p_{k+1}) &&\text{takes branch \eqref{eq:ev.prune3}} \\
&\quad\mathit{Prune}(q_{k+1,S},u_{k+1}\land p_k) 
  &&\text{takes branch \eqref{eq:ev.prune2}}\\
&\quad\quad\mathit{Prune}(q_{k,S},p_k) \\
&\quad\mathit{Prune}(q_{k+1,S},p_k\land v_{k+1})
  &&\text{takes branch \eqref{eq:ev.prune2}}\\
&\quad\quad\mathit{Prune}(q_{k,{S\cup\{k+1\}} },p_k)
\end{align}
Hence, \textit{Prune} is recursively called with arguments 
$(q_{k,S},p_k)$ by the previous call with arguments 
$(q_{k+1,{S-\{k+1\}} },p_{k+1})$.

Since for each $k\in[0.\,.\,n]$ there are $2^{n-k}$ sets~$S$,
the body of method \textit{Prune} is executed at least $2^{n+1}-1$
times, which is exponential in the size of the input.
Because predicates~$q_{k,S}$ are newly created, the auxiliary
space used by the algorithm is also at least exponential.
\end{proof}

The proof relies on sharing in the second argument of \textit{Prune}: If
$(u_{k+1}\land p_k)\lor(p_k\land v_{k+1})$ would be simplified to
$p_k\land(u_{k+1}\lor v_{k+1})$, then the run-time would not be exponential.

\begin{proposition}
In the worst case, the algorithm in Figure~\ref{fig:ev.alg} takes at most
linear time, when the size of the input is defined to be the number of
paths starting at the root of dag~$q$.
\end{proposition}

\begin{proof}
The body of $\mathit{Prune}(p,q)$ recursively calls $\mathit{Prune}(p',q')$
only if there is a path from~$q$ to~$q'$.
\end{proof}

The pruning algorithm works when predicates are represented by dags, but is
guaranteed to be efficient only if the dags have little sharing.

\section{Correspondence between Trees}
\label{sec:ev.correspondence}

Experiments showed that the algorithm of Figure~\ref{fig:ev.alg}
is useless by itself, because the SMT dags that represent
the predicates $p$~and~$q$ seldom share much. In particular,
if \escjava is used to produce a VC from the code in
Figure~\ref{fig:java_evol} (on page~\pageref{fig:java_evol}) with
and without the comment line~1, then there is almost no sharing,
despite hash-consing. The problem will still exist if instead of
\escjava we will use FreeBoogie with a Java frontend.

The issue is better understood if we take a step back and
look at the big picture. A frontend transforms the code in
Figure~\ref{fig:java_evol} into a Boogie program; a VC generator
transforms the Boogie program into an SMT query; an SMT solver
answers the query. The interfaces between these software
components are standardized at the language level---multiple
static verifiers understand the Boogie language and multiple
provers understand the SMT language. But, of course, information
must flow in the other direction too, because the user expects to
see a result.

How do we easily substitute an SMT solver for another if they
report results in different ways? And how do we substitute a
VC generator for another if they report results in different
ways? Leino et~al.~\cite{leino2005errors} describe an elegant
solution. The key insight is that no matter how an SMT solver
presents a counter-example, it probably includes a set of
identifiers. Similarly, no matter how a VC generator reports
errors, it probably includes a set of identifiers. Therefore,
extra information, like location, can be recovered from ornate
identifiers. For example, \escjava decorates identifiers by their
location, so that it knows to which occurrence of a certain
identifier in the Java code the counter-example refers to.
(Note that, since different versions are needed for passivation
anyway, they can be obtained by attaching strings that are
useful for error reporting too, instead of 1,~2, 3, \dots)

If the frontend that produces Boogie from the example in
Figure~\ref{fig:java_evol} decorates identifiers with location
information, then inserting a comment on line~1 results in a
Boogie program with completely different identifiers. Almost all
comparisons done by the algorithm of Figure~\ref{fig:ev.alg}
return false, and no simplification is~done.

\subsection{Matching SMT dags}
\label{sec:ev.match_dags}

To solve this problem, \fx improves the sharing in SMT dags
before pruning. The nodes in an SMT dag are labeled with a
string. Some labels, such as ``and'', ``or'', and ``41'', have
a special meaning for the prover while others, such as those
representing variables in predicates, are uninterpreted. The
uninterpreted labels are also called identifiers, because the
only semantic information they carry follows from whether two
of them are equal as strings or not. It is safe to rename
identifiers as long as their semantic information is preserved.
To explain formally why this is the case we would need to dig
more into the structure of predicates. However, there is an
important special case that is easy to prove already. Note that
\begin{equation}
|\lnot p|\limp|\lnot((v\gets e)\;p)|
\label{eq:ev.safe_subst}
\end{equation}
for any substitution~$(v\gets e)$, because, for an arbitrary
store~$\sigma$,
\begin{align}
&\phantom{\;=\;}
  \lnot\;((v\gets e)\;p)\;\sigma \\
&=\lnot((v\gets e)\;p\;\sigma) \\
&=\lnot(p\;((v\gets e)\;\sigma)) \\
&=(\lnot p)\;((v\gets e)\;\sigma) \\
&\Leftarrow |\lnot p|.
\end{align}
In other words, the correctness of pruning is not affected if
the \emph{variables} in predicate~$p$ are renamed beforehand, no
matter what renaming is done.

The heuristic used by \fx improves sharing by renaming only
variables, and is easier to understand in a slightly more abstract
setting.

\begin{problem}
Given are two unordered rooted dags $p$~and~$q$. Their nodes and
edges are labeled with strings. Find a one-to-one correspondence
between some leaves of the dag~$p$ and some leaves of the dag~$q$.
\label{pb:ev.match_dags}
\end{problem}

\begin{remark}
Such labeled dags may represent both SMT dags and Boogie ASTs.
Nodes could be labeled with strings such as ``not'', ``or'', and
``for all''. The edges below ``not'' and ``or'' would be labeled
by the empty string; The edges below ``for all'' would be labeled
with ``bound variable'' and, respectively, ``body''. However, the
exact encoding is not important. What is important is that dags
with labels on nodes and edges can easily encode both SMT dags
and Boogie~ASTs.
\end{remark}

Problem~\ref{pb:ev.match_dags} is an incompletely specified
optimization problem: It is clear what a feasible solution
is, but it is not clear which solutions are better and which
are worse. A cost function that evaluates correspondences is
missing. One possible measure is the size of the hash-consed
data structure that represents both dags $p$~and~$q$ after
corresponding leaves are glued. A cost function is a prerequisite
to any analytic performance analysis. However, analytic
performance analyses are likely to be very difficult to carry
out here and likely to be not very relevant in practice. To
name just one obstacle, an average case analysis requires VCs
to live in a probability space, and it is not at all clear what
probabilities accurately describe real VCs. It is therefore not
very useful to fix a cost function. Instead, pragmatism points
towards experimental algorithmics. In this case, objective
experimental measures include the overall proving time (with and
without improving sharing) and the size of the pruned VC (with
and without improving~sharing).

Intuitively, we expect any reasonable heuristic to work well
in simple cases. For example, if two big VCs can be obtained
one from the other by a one-to-one renaming of variables, then
we expect a good heuristic to recover the necessary one-to-one
renaming. This is the case with VCs generated from the program
in Figure~\ref{fig:java_evol} (on page~\pageref{fig:java_evol})
with and without the comment on line~1, if the frontend decorates
variables with location information. In general, a good heuristic
matches what a human would do.

\begin{example}
The dags from Figure~\ref{fig:ev.match1} represent two fairly
big predicates. All edge labels are empty. Most humans asked to
match the variables of dag~$p$ with the variables of dag~$q$
would match $u_1\edge v_3$ and $u_2\edge v_2$ and perhaps others.
If the sharing in dags $p$~and~$q$ would be eliminated (see
Section~\ref{sec:wpsp.vcsize}) then these correspondences are
harder to notice. \label{ex:ev.match1}
\end{example}

\begin{figure}\centering
\subfigure[dag $p$]{
\begin{tikzpicture}[xscale=2]
  \foreach \n/\y in {0/3,1/2,2/1}
    \node[predcirc] (a\n) at (0,\y) {$\land$};
  \foreach \n/\y in {0/2.5,1/1.5,2/.5}
    \node[predcirc] (o\n) at (0.5,\y) {$\lor$};
  \foreach \n/\y in {1/2,2/1,4/0}
    \node[predrect] (u\n) at (1,\y) {$u_\n$};
  \node[predrect] (u3) at (0,0) {$u_3$};
  \foreach \s/\t in {a0/a1,a0/o0,a1/a2,a1/o1,a2/u3,a2/o2,o0/a1,o0/u1,o1/a2,o1/u2,o2/u3,o2/u4}
    \draw[thick] (\s)--(\t);
\end{tikzpicture}
}\hfil
\subfigure[dag $q$]{
\begin{tikzpicture}[xscale=2]
  \foreach \n/\y in {0/3,1/2}
    \node[predcirc] (a\n) at (0,\y) {$\land$};
  \foreach \n/\y in {0/2.5,1/1.5}
    \node[predcirc] (o\n) at (0.5,\y) {$\lor$};
  \foreach \n/\y in {3/2,2/1}
    \node[predrect] (v\n) at (1,\y) {$v_\n$};
  \node[predrect] (v1) at (0,1) {$v_1$};
  \foreach \s/\t in {a0/a1,a0/o0,a1/o1,a1/v1,o0/a1,o0/u1,o1/v1,o1/u2}
    \draw[thick] (\s)--(\t);
\end{tikzpicture}
}\hfil
\subfigure[bipartite graph]{
\label{fig:ev.bipartite}
\begin{tikzpicture}
  \foreach \n in {1,2,3,4}
    \node[predrect] (u\n) at (0,\n) {$u_\n$};
  \foreach \n in {1,2,3}
    \node[predrect] (v\n) at (1,\n) {$v_\n$};
  \foreach \x in {1,2,3,4} \foreach \y in {1,2,3}
    \draw[thick] (u\x)--(v\y);
\end{tikzpicture}
}
\caption{Dags and bipartite graph for Example~\ref{ex:ev.match1}}
\label{fig:ev.match1}
\end{figure}
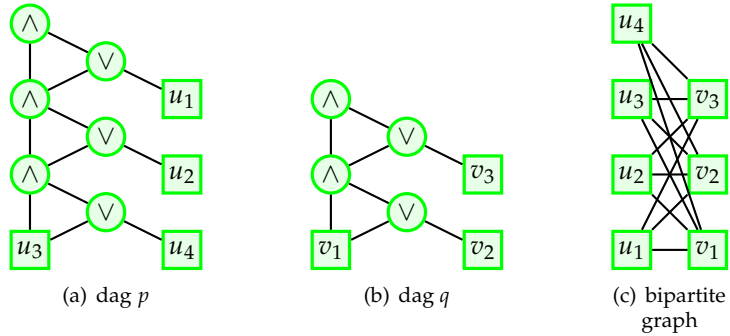

Because sharing benefits the search for identifier correspondences
but hurts pruning, an implementation in FreeBoogie should
\begin{enumerate}
\item find correspondences between dags,
\item then eliminate sharing (Section~\ref{sec:wpsp.unshare}),
\item and then prune the VCs (Section~\ref{sec:ev.pruning}).
\end{enumerate}

\paragraph{Approach}
We begin by constructing a complete bipartite graph---its left
nodes are the leaves of the dag~$p$ and its right nodes are the
leaves of the dag~$q$. For each node~$u$ on the left and each
node~$v$ on the right there is an edge $u\edge v$ whose weight
indicates how similar nodes $u$~and~$v$ are. Then we find a
maximum weight matching in this bipartite graph.

Any matching in a complete bipartite graph formed by the leaves of
dag~$p$ and, respectively, the leaves of dag~$q$ is a one-to-one
correspondence between some leaves of dag~$p$ and some leaves
of dag~$q$ as Problem~\ref{pb:ev.match_dags} requires. Which
matching is chosen depends on the heuristic used to compute the
edge weights. This heuristic is the knob we use to improve the
results, while we keep the overall structure of the solution
unchanged.

\begin{example}
Continuing Example~\ref{ex:ev.match1},
Figure~\ref{fig:ev.bipartite} shows the complete bipartite graph
that we build in the first step. Intuitively, we want the weights
on edges $u_1\edge v_3$ and $u_2\edge v_2$ to be relatively big,
such that the second step selects those edges.
\end{example}

\paragraph{Similarity of Leafs}

The remaining problem is simpler. Given a leaf~$u$ in the
dag~$p$ and a leaf~$v$ in the dag~$q$, we must find a
weight~$w_{p,q}(u,v)$ which corresponds to the intuitive
notion of similarity. The weight might be computed,
for example, by looking at $\mathit{label}(u)$ and
$\mathit{label}(v)$, the strings that label leaves $u$~and~$v$.
But Example~\ref{ex:ev.match1} suggests that the location of a
leaf in the dag is important. The location of a leaf~$u$ in a
labeled unordered dag~$p$, denoted $\mathit{location}_p(u)$, is a
multiset of string sequences, each string sequence corresponding
to a path from leaf~$u$ to the root of dag~$p$.

\begin{definition}
The multiset $\mathit{location}_p(u)$ contains the string
sequence $\mathit{label}(v_1)$, $\mathit{label}(e_1)$,
$\mathit{label}(v_2)$, $\mathit{label}(e_2)$,
\dots,~$\mathit{label}(v_n)$ once for each path
$v_1\stackrel{e_1}{\gets}v_2\stackrel{e_2}{\gets}\cdots v_n$ from
the leaf $u=v_1$ to the $\mathit{root}(p)=v_n$.
\end{definition}

Note that $\mathit{location}_p(u)$ is a set if the dag~$p$ was
built using hash-consing, that is, if structural equality is
equivalent to reference equality.

It is intuitive to use the information from leaf labels
independently from the information about the location of leaves.
\begin{equation}
w_{p,q}(u,v)=c(w_1(\mathit{label}(u),\mathit{label}(v)),
         w_2(\mathit{location}_p(u),\mathit{location}_q(v))
\end{equation}
Function~$w_1$ measures the similarity of two strings,
function~$w_2$ measures the similarity of two multisets, and
function~$c$ combines two similarities. Function~$c$ should
be increasing with respect to each of its arguments. Possible
choices for function~$c$ include
\begin{align}
c(w_1,w_2) &= \alpha w_1 + \beta w_2
  &&\text{for positive $\alpha$~and~$\beta$} \label{eq:ev.combin_sim} \\
c(w_1,w_2) &= w_1 w_2
  &&\text{for positive $w_1$~and~$w_2$}
\end{align}
Possible choices for function~$w_1$ include
\begin{itemize}
\item the longest common subsequence~\cite{hunt1977},
  $w_1(s,t)=\mathit{lcs}(s,t)$,
\item the Jaro--Winkler~\cite{winkler1990} similarity,
\item the inverse of a string distance, such as the edit 
  distance~\cite{levenshtein1966}, and
\item the similarity of the multisets of labels' $n$grams, 
  using various alternatives available for the function~$w_2$.
\end{itemize}
Possible choices for function~$w_2$ include
\begin{itemize}
\item the Jaccard similarity~\cite[Chapter~1]{markov2007}, 
  $w_2(A,B)=1-|A\cap B|/|A\cup B|$,
\item Dice's coefficient~\cite{dice1945}, 
  $w_2(A,B)=2|A\cap B|/(|A|+|B|)$, and
\item the opposite of a set distance, such as the Hamming 
  distance~\cite{hamming1950}.
\end{itemize}
(All these are easily adapted to multisets.)

\fx uses
\begin{align}
c(w_1,w_2)&=w_1+w_2 \\
w_1(s,t)&=\mathit{lcs}(s,t) \\
w_2(A,B)&=2|A\cap B| - \bigl| |A|-|B| \bigr|
\end{align}
No experiment revealed a situation where a correspondence found
by a human was better than the one found by \fx. However, it is
comforting to know that there is ample room for tuning the 
heuristic, in case better results are needed.

\paragraph{Locations within Dags}

In the worst case, the location of a leaf~$u$ in a dag~$p$
has size exponential in the number of nodes in the dag~$p$.

\begin{example}
For dag~$p$ in Figure~\ref{fig:ev.match1} (on
page~\pageref{fig:ev.match1}) we have
\begin{align}
|\mathit{location}_p(u_1)| &= 1\\
|\mathit{location}_p(u_2)| &= 2\\
|\mathit{location}_p(u_4)| &= 4\\
|\mathit{location}_p(u_3)| &= 8 
\end{align}
\end{example}

A naive implementation of the function~$w_2$ is very slow in the
worst case. To compute it quickly, locations must be represented
in a compact form.

\begin{example}
The locations of the leaves of dags $p$~and~$q$ from
Figure~\ref{fig:ev.match1} can be represented by the dag in
Figure~\ref{fig:ev.match2}.
\end{example}

\begin{figure}\centering
\begin{tikzpicture}
  \foreach \n/\y in {0/-3,1/-2,2/-1}
    \node[predcirc] (a\n) at (0,\y) {$\land$};
  \foreach \n/\y in {0/-2.5,1/-1.5,2/-.5}
    \node[predcirc] (o\n) at (1,\y) {$\lor$};
  \foreach \n/\y in {1/-2,2/-1,4/0}
    \node[predrect] (u\n) at (2,\y) {$u_\n$};
  \node[predrect] (u3) at (0,0) {$u_3$};
  \node[predrect] (v2) at (4,-1) {$v_2$};
  \node[predrect] (v3) at (4,-2) {$v_3$};
  \node[predrect] (v1) at (-1.5,-1) {$v_1$};
  \foreach \s/\t in {a0/a1,a0/o0,a1/a2,a1/o1,a2/u3,a2/o2,o0/a1,o0/u1,o1/a2,o1/u2,o2/u3,o2/u4,o0/v3,o1/v2,a1/v1,o1/v1}
    \draw[thick] (\s)--(\t);
\end{tikzpicture}
\caption{The trie of paths from Figure~\ref{fig:ev.match1}}
\label{fig:ev.match2}
\end{figure}
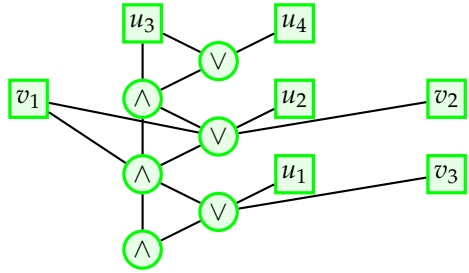

Hash-consing is used again. Notice that once the data structure
in Figure~\ref{fig:ev.match2} is built, a single reference
comparison is enough to conclude that $u_1$~and~$v_3$ have
the same location, and another reference comparison is enough
to conclude that $u_2$~and~$v_2$ have the same location.
Figure~\ref{fig:ev.loc_alg} shows the algorithm that computes the
data structure for representing locations, in the general case.
On line~2 it is assumed that the incoming edges of the original
dag are available. This is usually not the case for SMT dags and
for Boogie ASTs, which is why a simple preprocessing step is
necessary. Hash-consing is invoked on the last line: The call
to \textit{hashCons} returns an old data structure, if there
is one with the same content. Finally, note that the algorithm
in Figure~\ref{fig:ev.loc_alg} assumes that there is a way to
extract labels from nodes and edges of the original dag (the
method~\textit{label}).

\begin{figure}\centering\leavevmode\vbox{
\begin{alg}
\^  $\proc{Location}(x)$ \comment memoized
\=  $\mathit{children}:=[]$ \comment empty list
\=  ~for~ ~each~ incoming edge $e$
\+    append $\bigl(\mathit{label}(e),\mathit{Location}(\mathit{source}(e))\bigr)$ to \textit{children}
\0  ~return~ $\mathit{hashCons}(\mathit{label}(x), \mathit{children})$
\end{alg}}
\caption{Algorithm for computing locations}
\label{fig:ev.loc_alg}
\end{figure}

\subsection{Matching Boogie ASTs}

If sharing is improved at the level of SMT dags, then the caches
of FreeBoogie's transformers are not used. The new Boogie program
is read, a new VC is computed repeating all the work, and only
then are commonalities observed. Unfortunately, an implementation
that improves sharing at the level of Boogie ASTs is much harder
to get right than one at the level of SMT dags, although there
does not seem to be any insurmountable conceptual issue.

One problem is that the substitution must be applied on the
new Boogie AST, which corresponds to dag~$q$. The proof
that any variable substitution in the old VC is sound (see
\eqref{eq:ev.safe_subst} on page~\pageref{eq:ev.safe_subst}) does
not help here. Instead, it must be argued that any one-to-one
renaming does not change the semantics of Boogie programs.
Another problem is that variables with a limited scope are much
more common, so some special treatment for them may be desirable.
(In fact, the implementation in \fx does treat quantifiers
specially because they limit the scope of variables.) On the
bright side, some issues disappear. For example, since there is
no sharing within one Boogie AST (Section~\ref{sec:design.ast}),
locations within a Boogie AST may be represented in the naive
way.

In general though, any solution to Problem~\ref{pb:ev.match_dags}
should be adaptable to Boogie ASTs. Nodes' labels are
the class names; edges' names are the field names. (See
Section~\ref{sec:design.ast}.)

\section{Example}
\label{sec:ev.example}

Let us go back to the example in Figure~\ref{fig:java_evol} on
page~\pageref{fig:java_evol}. The method \textit{dayOfYear}
is the method of interest. The other methods may or may not
have bodies, but that is irrelevant because the verification is
modular. Three types of changes are illustrated. First, adding
line~1 is a trivial change. Second, adding line~8 is a change in
the specification of another method. In general, whenever the
contract of a method is changed, all the methods that depend on
this contract must be re-verified. In particular, if a contract's
postcondition is strengthened, then all dependent methods that
were correct remain correct. Third, adding line~15 or line~30
instructs the verifier to check extra properties. Note that
adding line~15 also entails re-verification of all the methods
that depend on the method \textit{dayOfYear}, which are not
of interest in this example.

Proving the VC generated by \escjava for the program of
Figure~\ref{fig:java_evol} takes about $17$~seconds with Simplify
in all cases. \fx spends about $2$~seconds to prune the VC in
all cases. This could be improved in a careful implementation in
FreeBoogie. For the original program plus line~1 or plus line~8,
the pruned VC is handled by Simplify in $<0.1$~seconds. For the
original program plus line~15, the pruned VC is not handled
faster by Simplify, so verification is slowed down overall. For
the original program plus line~30, the pruned VC is handled in
about $8$~seconds by Simplify, so the proving time is approximatively
cut in half. These numbers show that, in some cases, a proper
implementation of pruning may be useful. Proving the same
method with Z3~2.0 takes $0.5$~seconds and with \fx takes about
$2000$~seconds. This suggests that pruning might be worthwhile
only for methods that are difficult to verify.

It is also possible to obtain the VC with FreeBoogie,
although not completely automatically. The program in
Figure~\ref{fig:java_evol} without annotations can be compiled
with javac ($0.7$~seconds) and then transformed into Boogie
with B2BPL ($0.5$~seconds). Then, annotations must be added
manually in the Boogie code. FreeBoogie then generates a VC
($1.4$~seconds) and then a prover is used (times are as before).
Keep in mind that all these times include repeatedly parsing
and writing to disk around $10$~KB of data. The time taken by
FreeBoogie is quite big by comparison with the other stages of
the pipeline, which suggests that improving sharing at the Boogie
AST level might be worthwhile even for relatively easy problems
like this one.

\section{Conclusions}

The previous sections of this chapter present a design for
a future component of FreeBoogie. The design is explored
theoretically and also practically, through the prototype
implementation in \fx. The most appealing property of the design
is its \emph{flexibility}.

The two important ideas are
\begin{enumerate}
\item 
  finding common parts of two trees under renaming of (some)
  leaves and
\item
  pruning a VC based on an old valid VC.
\end{enumerate}
The solution to the first problem is flexible because it depends
on a similarity measure for which there are many possible
choices. The solution for the second problem is flexible because
of the generic proof technique, which can be used to quickly
check whether other simplification rules are sound.

The main drawback is the lack of substantial evidence that the particular
pruning method proposed in Section~\ref{sec:ev.pruning} does indeed reduce
the size of typical queries.  All we know is that \emph{if} it reduces the
size, then it does so in a sound~way.

A previously correct version of the program should almost always be
available if the verification tool is used from the beginning.  As a future
development, it would be interesting to see how to speed up verification
when all known previous VCs are invalid.

\section{Related Work}
\label{sec:related}

The basic idea, that of taking advantage of the
results of old runs to speed up new runs, appears in
incremental compilation~\cite{schwartz1984}, extreme
model checking~\cite{henzinger2003extreme}, proof
reuse~\cite{beckert2004reuse}, and in many other places. The work
on proof reuse of Beckert and Klebanov~\cite{beckert2004reuse}
is most similar to the work presented in this chapter. However,
they focus on interactive theorem proving, and their similarity
heuristic is quite different.

Detecting differences between trees is an area rich in research results.
However, it seems that the goal of improving sharing by the renaming of
leaves was not addressed before. Usually the goal is to compute the tree
edit distance~\cite{tai1979} and the accompanying edit script that
completely transforms a tree into the other. For unordered trees, such as
those with associative--commutative operators, finding the tree edit
distance is NP-hard~\cite{bille2005}, while for ordered trees cubic time
was achieved~\cite{demaine2009tree}. The cubic time matches that of the
heuristic described in Section~\ref{sec:ev.correspondence}, where the cubic
time Hungarian algorithm~\cite{kuhn1955} is used. (The implementation of
the Hungarian algorithm follows Knuth~\cite[program
\textit{assign\_lisa}]{knuth1993sgb}.) However, profiling shows that
computing the matching takes far less time than computing similarity
weights between all pairs of leaves.

The idea of automatically pruning an SMT query based on the
information in a similar known query seems to be new.

There are many other approaches for improving the speed of
program verifiers. James and Chalin~\cite{james2010esc4} run
multiple different provers in parallel. They also cache the text
of VCs so that they never send the exact same query to a prover
twice. Although they speculate that normalizing identifiers would
reduce the brittleness of the cache, it seems that the techniques
of Section~\ref{sec:ev.correspondence} fit perfectly their
requirements.

Verifast~\cite{verifast} and jStar~\cite{distefano2008jstar}
are Java program verifiers that use symbolic execution as
an alternative to VC generation and their good performance
seems to owe to this design choice. It is interesting to note
that Moore~\cite{moore2006} shows that interpreting a program
according to some operational semantics leads to exactly the same
proof obligations as the subgoals necessary to prove a VC\null.

Babi\'c and Hu~\cite{babic2008calysto} present a set of
techniques for improving efficiency, including caching results,
combining a rough symbolic execution with VC generation, and
using application-specific theorem provers. They also use
``maximally shared graphs,'' which are almost the same as
hash-consed expressions. (Their definition is more general
because it allows cycles, but it seems that all their graphs
are actually dags.) Hash-consing was described in 1958 by
Ershov~\cite{ershov1958}, and is a concept that keeps being
rediscovered, so it probably deserves to be better known.

The SMT community~\cite{barrett2010lib} standardized a
language~\cite{barrett2010lang} for expressing predicates
and a command language for communicating with a solver.

In the theorem proving community, the string
sequences that are part of locations, as defined in
Section~\ref{sec:ev.match_dags}, are known as \emph{path
strings}\index{path string}~\cite[Chapter~26]{robinson2001v2}.

Craig~\cite{craig1957} proved that interpolants always exist:
\begin{quote}\footnotesize\index{interpolant}
\textsc{Lemma 1}. \textit{If $\vdash A\supset A'$ and if $A$~and~$A'$
have a predicate parameter in common, then there is an ``intermediate''
formula~$B$ such that $\vdash A\supset B$,\quad $\vdash B\supset A'$,
and all parameters of $B$ are parameters of both $A$~and~$A'$. Also,
if $\vdash A\supset A'$ and if $A$~and~$A'$ have no predicate parameter
in common, then either $\vdash\lnot A$ or~$\vdash A'$.}
\end{quote}
Interpolants are used in model checking~\cite{mcmillan2005}, predicate
refinement~\cite{jhala2006}, loop invariant inference~\cite{mcmillan2008}.
Interpolation algorithms may, in general, exploit more than the boolean
structure of formulas: For example, Jain~\cite{jain2008} gives algorithms
for finding interpolants for fragments of integer linear arithmetic.

\chapter{Semantic Reachability Analysis}
\label{ch:reachability}

\chquote{To be successful you have to be selfish, or else you
never achieve. And once you get to your highest level, then
you have to be unselfish. Stay reachable. Stay in touch. Don't
isolate.} {Michael Jordan}

\noindent Program verifiers typically check for partial
correctness or termination, but there are other interesting
analyses. This chapter (1)~explains when and why reachability
analysis is useful, (2)~presents its theoretical underpinnings,
and then (3)~presents and analyzes algorithms that make it
practical.

\section{Motivation}

Semantic reachability analysis finds four seemingly unrelated
types of problems: dead code, doomed code, inconsistent
specifications, and bugs in the frontend of the program verifier.

\subsection{Dead Code}

Figure~\ref{fig:ra-dead-code-ex} illustrates two problems.
Line~10 is dead in the standard sense; line~5 is dead only
if annotations are taken into account. 
  
\bc
\begin{jml}
static void m(int x) 
  requires x >= 0;
{
  if (x < 0)
    throw new IllegalArgumentException("x must be nonnegative");
  $\cdots$
  if (y < 1) {
    $\cdots$
    if (y > 1) {
      $\cdots$
  }}
}
\end{jml}
\ec{Dead code}{fig:ra-dead-code-ex}

Dead code analysis is standard in compilers. Java, for example,
forbids any statement following the statement \textbf{return},
because such situations are usually bugs. The novelty
here is that we take into account annotations, which may be
non-executable. When the precondition holds (line~2) the first
\textbf{if} condition (line~4) does not, so no exception is
thrown.

In current practice, programmers check arguments at the beginning of the
method body with a series of guarded \textbf{throw} statements.  The
pattern $\mathbf{if}\ldots\mathbf{throw}\ldots$ is short, but so pervasive
that it still pays of to encapsulate in a library (see, for example,
\cite[\textit{com.google.common.base.Preconditions}]{guava-libraries}).

In the presence of annotations, however, the code that checks if arguments
are legal is mostly redundant.

If all calls to method~$m$ are checked statically against its
specification, then no runtime check is necessary. A runtime
check remains desirable when unverified code may call method~$m$.
There are Java compilers~\cite{burdy2005jml} that generate
bytecode from JML annotations. The translation is not perfect.
One technical problem is that JML preconditions do not have
descriptions. If the expression is particularly simple (such as
$x\ge0$ in Figure~\ref{fig:ra-dead-code-ex}) then it is possible
to construct the description automatically (``$x$~must be
nonnegative''). For longer expressions, though, the automatically
generated description is unwieldy and unuseful. The problem is
just technical since it has the simple solution of modifying
JML to have descriptions for preconditions, postconditions, and
assertions. The second problem is more serious. If the expression
contains quantifiers, which in JML are typically over very large
domains like the integers, then it is not always possible to
check it efficiently at runtime. However, in such situations it
is likely that the programmer would have not written a runtime
check.

Microsoft's Code Contracts~\cite{code-contracts} uses an annotation
language that is designed to work both with run-time checking and
with static analysis tools.

\subsection{Doomed Code}

Figure~\ref{fig:ra-doomed-code-ex} illustrates two problems, one
on line~7 and one on line~14. These are the most common issues
in practice (Section~\ref{sec:ra.case_study}).

\bc
\begin{jml}
abstract class C {
  int f, g;
  static void m1(C x) {
    if (x != null)
      x.f = 0;
    else
      System.out.println(x.f);
  }
  static void m2();
    modifies f, g;
  static void m3()
    modifies f;
  {
    m2();
    $\cdots$
  }
}
\end{jml}
\ec{Doomed code}{fig:ra-doomed-code-ex}

Execution always crashes at line~7, before printing. Any test
that covers line~7 would detect the problem, but test suites
rarely have complete coverage. The example is adapted from
Hoenicke et al.~\cite{hoenicke2009}, whose example is in turn
inspired by an old bug in Eclipse. Such bugs do occur in real
software even if they are easy to detect. It would therefore
be beneficial to find such problems automatically and without
writing any test. (But this is true about any kind of bug.)
Another reason given by Hoenicke et al.~\cite{hoenicke2009} is
compelling: \emph{The lack of annotations does not lead to false
positives}. Most frequent complaints from practicians about
formal methods tools are that
\begin{enumerate}
\item it is too much work to add annotations and
\item there are too many false positives.
\end{enumerate}
To be more precise, the lack of \emph{assumptions} (stemming,
for example, from the lack of preconditions and object
invariants) does not lead to false positives. A program point is
\emph{doomed} if it crashes in every execution. If an assumption
is removed (which may mean, for example, that more values are
allowed for an argument), then all the old executions are still
possible. In other words, removing an assumption never turns a
non-doomed program point into a doomed program point.

Line~14 is also doomed, this time because of annotations, not
because of the code. The \textbf{modifies} clause lists the
fields that a method is allowed to assign to. From the point of
view of the program verifier the execution never proceeds past
line~14, so all potential bugs that follow are hidden.

\subsection{Inconsistent Specifications}

Figure~\ref{fig:ra-bad-spec-ex} illustrates two problems caused
entirely by specifications.

\bc
\begin{jml}
pure static native int m1()
  ensures result == result + 1;
static int m2(int x)
  requires x < 0;
  requires x > 0;
{
  return x / 0;
}
\end{jml}
\ec{Inconsistent specifications}{fig:ra-bad-spec-ex}

The contract of method~\textit{m1} cannot possibly correspond to
a correct implementation, yet there is no implementation to check
against the contract. In this case the implementation is written
in a different language, but sometimes it is simply unavailable,
for example if it comes from a proprietary library. The example
may seem silly since no one would make the mistake to think that
some integer equals its successor. However, inconsistencies
may arise from other causes. One example is specification
inheritance: Since the conflicting annotations are in different
files it is hard for a human to notice. Another example is the
use of pure methods in specifications. The literature concerned
with tackling this latter source of inconsistencies is reviewed
in Section~\ref{sec:ra.conclusions}.

Lines 4~and~5 illustrate another type of bug caused by
inconsistent specifications. The body of method~\textit{m2} is
available and the program verifier will always report it is
correct. Without reachability analysis, the bug is caught only
if method~\textit{m2} is called. Such behavior is somewhat akin
testing, which catches bugs (only) by calling the (potentially)
buggy methods.

In practice, inconsistent specifications are common.

\subsection{Unsoundness and Bugs}

The code in Figure~\ref{fig:ra-loop-unroll}, apart from being
more than a little silly, will always crash at line~4. That is
not the issue relevant here, however. The issue is that the
problem may be missed by an unsound program verifier.

\bc
\begin{jml}
int x;
for (int i = 5; --i >= 0;)
  x = i;
x = 1 / x;
\end{jml}
\ec{Unsoundness of the frontend}{fig:ra-loop-unroll}

Again, the example might seem contrived. Why not get rid of unsoundness in
the first place instead? And why would we expect this problem to appear
often? Before answering these questions let us see why this is indeed a
problem in \escjava (with default options). \escjava was designed primarily
to be useful to programmers and it tackled the problem of false positives
by choosing to be unsound. One example of unsoundness is that, by default,
loops are unrolled three times. This means that all executions for which
the body of the loop is executed more than three times are not analyzed.
All executions in Figure~\ref{fig:ra-loop-unroll} go through the loop body
five times so, just before executing the loop body for the fourth time,
\escjava will top analyzing and miss the problem on line~4. Note that if
the loop would be executed $n$~times instead of $5$~times, where $n$~is a
variable possibly less than four, then the problem on line~4 would be
caught~\cite{detlefs1998}.

We are now in a better position to answer the two questions.

First, a tool may choose to be unsound for engineering reasons.
This is not as bad as it may sound. We saw that a loop with
a variable bound instead of a constant bound does not cause
problems, and variable bounds are probably more common. Also,
this is only the default behavior. A novice user is supposed to
get rid of the errors signaled in this unsound default mode of
execution. An expert user is supposed to ask \escjava to treat
loops in a safe way. In this case, the tool usually requires more
guidance in the form of annotations. 

Second, the unsoundness might not be intended, but rather a bug.
In this guise, reachability analysis is useful as a safeguard
against buggy implementations of the program verifier itself.  It
is common for good programmers to write routines that check for
complex invariants at runtime and use these routines during
development. Reachability analysis has been used in this way in
the Boogie tool from Microsoft Research under
the name of ``smoke testing.''

\subsection{An Unexpected Benefit}

Finally, in conjunction with loop invariant inference techniques,
reachability analysis detects certain non-terminating programs.
Consider the following code fragment.
\begin{jml}
for (int i = 0; i < 10; ++j) 
  sum += i;
$\cdots$
\end{jml}
It is easy to infer the loop invariant $i=0$, which contradicts
the negation $i\ge10$ of the loop condition. The code after the
loop is unreachable because the loop never terminates.

\section{Theory}
\label{sec:ra.theory}

Semantic reachability analysis is performed after passivation
(Chapter~\ref{ch:passive}), so the program it analyzes is an
acyclic flowgraph without assignments.

\begin{definition}\index{semantic reachability}
A node~$x$ in a flowgraph is \emph{semantically reachable}
when the flowgraph has an execution containing the
state~$\langle\sigma,x\rangle$ for some store~$\sigma$.
\end{definition}

For a correct flowgraph (see Theorem~\ref{th:correctness} on
page~\pageref{th:correctness}), a node~$x$ is semantically
reachable when the flowgraph becomes incorrect after replacing
the statement of node~$x$ with \textbf{assert false}. (This
is because according to~\eqref{eq:assert-nok-opsem} on
page~\pageref{eq:assert-nok-opsem} there is a transition
$\langle\sigma,x:\mathbf{assert}\;\mathbf{false}\rangle\leadsto
\mathit{error}$ for all stores~$\sigma$.) A naive algorithm for
semantic reachability analysis replaces each statement in turn by
\textbf{assert false} and reverifies the program. However, this
algorithm handles only correct flowgraphs and is slow.

There is a very easy way to make any flowgraph correct: Transform
all assertions into assumptions! Besides fixing all bugs, this
transformation has the useful property that it maintains semantic
reachability.

\begin{proposition}
Consider a flowgraph~$G$ with no assignments and the
flowgraph~$H$ obtained by replacing each statement
\textbf{assert}~$q$ with the statement \textbf{assume}~$q$.
Node~$y$ is semantically reachable in flowgraph~$G$ if and only
if it is semantically reachable in flowgraph~$H$.
\end{proposition}

\begin{proof}
According the operational semantics of core Boogie 
(see Chapter~\ref{ch:boogie},
especially~\eqref{eq:assume-assert-ok-opsem}),
a flowgraph has an execution \[
  \langle\sigma,x_1:\mathbf{assert}/\mathbf{assume}\;p_1\rangle,\ldots,
  \langle\sigma,x_n:\mathbf{assert}/\mathbf{assume}\;p_n\rangle,
  \langle\sigma,y\rangle
\] when there is a path $x_1\to\cdots\to x_n\to y$ in the
flowgraph and there is a store~$\sigma$ that satisfies 
$p_1\land\ldots\land p_n$. It is irrelevant whether 
intermediate statements are assertions or assumptions
as there is only one rule in the operational semantics
for both.
\end{proof}

There is still the issue of efficiency.
Chapter~\ref{ch:spwp} showed that both
$\mathit{vc}_\mathit{wp}$~and~$\mathit{vc}_\mathit{sp}$ can be
computed in time linear in the size of the flowgraph. Since we
need to reverify the program once for each node in the flowgraph,
the total time to compute the prover queries is quadratic in
the size of the program. Most flowgraphs in practice have at
most hundreds of nodes and hundreds of edges, which means that a
quadratic algorithm will be very fast and the real problem is the
time spent in the prover.

The next section is concerned with reducing the number of calls
to the prover. Before that, let us briefly see how to compute
\emph{all} prover queries at once in linear~time.

If there is an execution that contains the
state~$\langle\sigma,y\rangle$, then the previous
states~$\langle\sigma,x\rangle$ correspond to nodes~$x$ that have
a path~$x\leadsto y$ to node~$y$. Hence, to determine whether
such executions exist we can look at the sub-flowgraph induced by
nodes~$x$ that can reach node~$y$. This observation speeds the
computation of VCs by a factor of two, for both methods---weakest
precondition and strongest postcondition.

Observe now that the precondition~$a_y$ computed by the strongest
postcondition method (Section~\ref{sec:spwp.summary}) depends only on
nodes~$x$ that can reach node~$y$. In other words, the precondition~$a_y$
does not change when we select a sub-flowgraph corresponding to some other
node~$y'$ that is reachable from node~$y$, so it needs not be recomputed.
We simply compute all preconditions~$a_x$ in linear time according to the
equations for the strongest postcondition method in
Section~\ref{sec:spwp.summary}. When we replace statement~$x$ by
\textbf{assert false} it becomes the only assertion in the flowgraph and
the VC according to~\eqref{eq:spwp.sp.vc} is $\lnot a_x$. We proved 
the following proposition.

\begin{proposition}
A node~$y$ is semantically reachable if and only if its
precondition~$a_y$ computed using the strongest postcondition
method (Chapter~\ref{ch:spwp}) is satisfiable, that is, when
$|\lnot a_y|$ holds.
\end{proposition}

According to~\eqref{eq:spwp.sp.post} (on
page~\pageref{eq:spwp.sp.post}) there is no difference
between \textbf{assume} and \textbf{assert} when computing
preconditions~$a_x$ and postconditions~$b_x$ using the strongest
postcondition method. So one advantage of the strongest
postcondition method is that we do not have to transform
assertions into assumptions at all. Another clear advantage of
the strongest postcondition method is that it produces all prover
queries in linear time, as opposed to the weakest precondition
method which requires quadratic time. It is intuitive that `going
forward' is better suited for analyzing semantic reachability.

\paragraph{Types of Problems} 

Nodes that are semantically unreachable are likely
to be caused by bugs. For example, the contradictory
preconditions in Figure~\ref{fig:ra-bad-spec-ex} (on
page~\pageref{fig:ra-bad-spec-ex}) cause the body of
method~\textit{m2} to be semantically unreachable. It is
therefore interesting to find potential causes of unreachable
code.

A node~$x$ is a \emph{blocker} when it is reachable but lets no execution
pass through it. That is, there is no execution containing
$\langle\_,x\rangle,\langle\_,\_\rangle$. Yet in other words, the only
possible successor of state~$\langle\_,x\rangle$ (in an execution) is the
\textit{error} state, which happens exactly when $\lnot b_x$~is valid.  A
node is semantically unreachable if all its predecessors are blockers,
according to~\eqref{eq:spwp.sp.pre} (on page~\pageref{eq:spwp.sp.pre}).

Semantic reachability analysis detects three types of problems:
\begin{enumerate}
\item 
  Semantically unreachable nodes. This includes all dead code.
\item
  Blocker assumptions. These are likely root causes for
  semantically unreachable nodes, so they point closer to
  the actual bug.
\item
  Blocker assertions, also known as doomed program points.
  These are bugs that manifest in every execution.
\end{enumerate}

Formally, there can be no error after a blocker. That is, there
is no execution that goes wrong after it passed through a
blocker, for the simple reason that there is no execution passing
through a blocker. However, intuitively it is helpful to think of
blockers as potentially hiding subsequent bugs.

\section{Algorithm}

Figure~\ref{fig:ra.naive_alg} summarizes the algorithm
suggested by the previous section. The methods \textit{Pre}
and \textit{Post} are those from Figure~\ref{fig:sp-algo} (on
page~\pageref{fig:sp-algo}); the method $\mathit{Not}(p)$
constructs the SMT dag that represents the predicate~$\lnot
p$, perhaps carrying out basic simplifications; the method
\textit{Valid} checks whether its argument represents a valid
predicate by querying the prover. The $|V|$ calls to the method
\textit{Pre} take $O(|V|+|E|)$ time and similarly for method
\textit{Post} (see the proof of Proposition~\ref{prop:sptime} on
page~\pageref{prop:sptime}).

\begin{figure}\centering\leavevmode\vbox{
\begin{alg}
\^  $\proc{NaiveReachabilityAnalysis}(G)$
\=  ~for~ ~each~ statement $x$ of $G$
\+    ~if~ $\mathit{Valid}(\mathit{Not}(\mathit{Pre}(x)))$
\+      report that node $x$ is semantically unreachable
\-    ~else if~ statement $x$ is an assertion ~and~ $\mathit{Valid}(\mathit{Not}(\mathit{Post}(x)))$
\+        report that node $x$ is doomed
\end{alg}}
\caption{A simple algorithm for semantic reachability analysis}
\label{fig:ra.naive_alg}
\end{figure}

This algorithm reports too many errors. For the Boogie program
\begin{boogie}
assume false;
assert $p_1$;
$\cdots$
assert $p_n$;
\end{boogie}
it prints $n$~error messages, one for each assertion. Surely,
the user would prefer one error message that points at the
assumption, which is the cause of all the other errors.

More seriously, it turns out that this algorithm, even if it
takes just linear time to build the prover queries, is unusable
in practice because it is too slow. It is simply not acceptable
to call the prover $|V|$~times for one method. A typical method
has $|V|\approx100$ nodes, and one call to the prover typically
takes between a tenth of a second and one second.

\begin{remark}
All the times are given for an implementation inside \escjava,
which is activated by the option~\texttt{-era}. \escjava has a
robust frontend for Java with JML annotations, so it is easier
to perform meaningful experiments. The implementation first
constructs a flowgraph essentially equivalent to those built by
FreeBoogie, so the implementation should be easy to adapt.
\end{remark}

\subsection{The Propagation Rules}

Figure~\ref{fig:fg.chain.a} shows a path flowgraph. (Node~$1$ is
the initial node.) Node~$4$ is semantically reachable when there
is an execution that contains the state~$\langle\sigma,4\rangle$,
for some store~$\sigma$. Such an execution must correspond
to some path from the initial node~$1$ to node~$4$, and the
only such path is $1\to2\to3\to4$. Hence, the execution
is $\langle\sigma,1\rangle$,~$\langle\sigma,2\rangle$,
$\langle\sigma,3\rangle$, $\langle\sigma,4\rangle$, which means
that the other nodes are semantically reachable too. (The store
is not modified by assertions or assumptions.) This motivates the
study of the following problem, which is more~abstract.

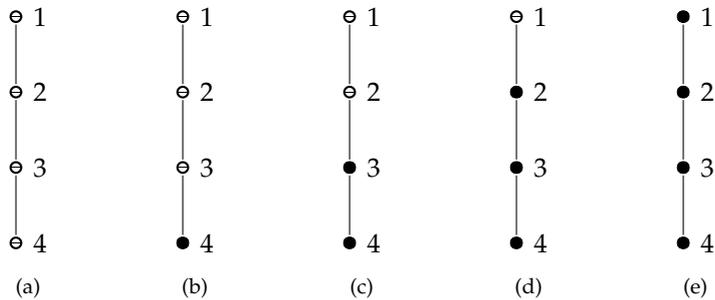
\begin{figure}
\def\sf#1#2#3{
  \subfigure[]{\label{fig:fg.chain.#1}
    \begin{tikzpicture}
    \foreach \n in {#2}
      \enode (\n) at (0,-\n) [label=right:$\n$] {};
    \foreach \n in {#3}
      \fnode (\n) at (0,-\n) [label=right:$\n$] {};
    \draw (1)--(2)--(3)--(4);
    \end{tikzpicture}
  }
}
\centering
\sf a{1,2,3,4}{}\hfil
\sf b{1,2,3}{4}\hfil
\sf c{1,2}{3,4}\hfil
\sf d{1}{2,3,4}\hfil
\sf e{}{1,2,3,4}
\caption{A path flowgraph with $4$ nodes}
\label{fig:fg.chain}
\end{figure}

\begin{problem}\index{problem!semantic reachability}
Alice and Bob play a game. They start by sharing the
flowgraph~$G$. Then Alice secretly colors the nodes with white
and black such that each white node is connected by a white path
with the initial node. Bob's goal is to reconstruct the coloring
by asking Alice as few questions as possible. He is allowed to
ask ``what is the color of node~$x$?''\ for any node~$x$ and
Alice will answer~truthfully.
\label{pb:ra.abstract}
\end{problem}

\begin{remark}
Alice is the theorem prover, Bob is FreeBoogie (or \escjava),
white nodes are semantically reachable nodes, and black nodes are
semantically unreachable nodes.
\end{remark}

Bob's flowgraph initially has only gray nodes and, by the end,
he colors each with either white or black. At some intermediary
stage, he has some white nodes, some black nodes, and the rest
are gray. We say that two colorings are \emph{consistent} when
there is no node that is white in one coloring and black in
the other. Bob maintains his coloring consistent with Alice's
coloring. (Also, because it is not useful, he never changes the
color of a node into gray.) A \emph{valid coloring} has no gray
node and satisfies the condition that
\begin{equation}
\label{eq:ra.fg-cond}
\text{every white node is reachable from the initial node 
by a white path.}
\end{equation}
For example, Figure~\ref{fig:fg.chain} shows all the valid
colorings of a path flowgraph. A \emph{possible coloring} (at
some intermediate stage) is a valid coloring that is consistent
with Bob's coloring. The rule which Bob uses to paint nodes
white or black is simple: If node~$x$ is white in all possible
colorings, then Bob colors it white; if node~$x$ is black in
all possible colorings, then Bob colors it black. Obviously,
this maintains the invariant that Bob's coloring is consistent
with Alice's coloring. To ensure progress, Bob asks Alice from
time to time about the color of a node that is still gray in his
coloring.

\begin{example}
Consider the flowgraph in Figure~\ref{fig:ra.fg.worst}. The
initial node is always reachable. Assuming that Bob starts
with the initial node white and the others gray, he must ask
$5$~questions. After asking $k$~questions he must consider
$2^{5-k}$ possible colorings. However, the preconditions of
all non-initial nodes are the same as the postcondition of
the initial node. Therefore, they are either all semantically
reachable or all semantically unreachable, and this can be
determined with one call to the prover.
\label{ex:ra.bad-model}
\end{example}

\begin{figure}\centering
\begin{tikzpicture}
\enode (r) at (0,1) {};
\foreach \x in {-2,-1,...,2} {
  \gnode (\x) at (\x,0) {};
  \draw (r)--(\x);
}
\end{tikzpicture}
\caption{Worst case flowgraph}
\label{fig:ra.fg.worst}
\end{figure}
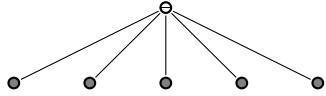

Example~\ref{ex:ra.bad-model} seems to suggest that performance
is hurt because Problem~\ref{pb:ra.abstract} is \emph{too}
abstract. Related to this, it seems that the theorem
prover is never queried about doomed code (unsatisfiable
postconditions), but only about semantically unreachable code
(unsatisfiable preconditions). Both these issues have an
easy fix. Instead of playing the color game on the original
flowgraph, Alice and Bob play it on a modified flowgraph
that has two nodes $x_i$~and~$x_o$ for each node~$x$ in the
original flowgraph. Node~$x_i$ takes over the incoming edges
of node~$x$; node~$x_o$ takes over the outgoing edges of
node~$x$. Also, there is an edge~$x_i\to x_o$. When Bob asks
Alice about the color of node~$x_i$, this means that \escjava
(or FreeBoogie) asks the theorem prover about the satisfiability
of the precondition~$a_x$; when Bob asks Alice about the
color of node~$x_o$, this means that \escjava (or FreeBoogie)
asks the theorem prover about the satisfiability of the
postcondition~$b_x$.

Bob does not want to look at each possible coloring, because
there might be an exponential number of them. Hence, he must
characterize the nodes whose color he can infer without referring
to all possible colorings.

\begin{definition}\index{dominator}
Node $x$ of a flowgraph \emph{dominates} node~$y$ when all
paths from the initial node to node~$y$ contain node~$x$.
Node~$x$ \emph{immediately dominates} node~$y$ if $x\ne y$
and all nodes that dominate node~$y$ also dominate node~$x$.
\label{def:ra.dominator}
\end{definition}

\begin{proposition}
A gray node is white in all possible colorings if and only if
it dominates a white node in the non-black subgraph (in Bob's
coloring). \label{prop:ra.fix-white}
\end{proposition}

\begin{proof}
Because Bob's coloring is consistent with a valid coloring it is
possible, for each white node~$y$, to find a non-black path from
the initial node to node~$y$. Consider the coloring constructed
by picking such a path for each white node and then painting all
these paths white; then paint the remaining gray nodes black.
This coloring is possible by construction. Also, whenever white
node~$y$ is processed, choose a path that avoids node~$x$, if
possible. This shows that if node~$x$ does not dominate a white
node, then there are some possible colorings in which it is
black. 

For the converse, consider a coloring in which node~$x$ is black
and dominates white node~$y$. Then \eqref{eq:ra.fg-cond}~is
violated and the coloring is invalid.
\end{proof}

Similarly, one can prove the following.

\begin{proposition}
A gray node~$x$ is black in all possible colorings if and only if
all the paths from the initial node to node~$x$ contain a black
node (in Bob's coloring). \label{prop:ra.fix-black}
\end{proposition}

An immediate consequence of the last two proposition will prove
useful.

\begin{corollary}
A gray node is sometimes white and sometimes black in possible
colorings if and only if in the non-black subgraph (a)~it does
not dominate a white node and (b)~it is connected to the initial
node. (Again, in Bob's coloring.)
\label{co:ra.propagate}
\end{corollary}

Figure~\ref{fig:alg_propagate} shows how Bob paints his
flowgraph. The method \textit{ComputeDominators} computes the
dominator tree of a graph (Section~\ref{sec:ra.dominators}).
The method \textit{PickQueryNode} returns some gray node
(Section~\ref{sec:ra.heuristic}). The main loop of the method
\textit{RecoverColoring} maintains the invariants:
\begin{enumerate}
\item The coloring is consistent with Alice's coloring.
\item 
  For each gray node, there are possible colorings in which
  it is white and possible colorings in which it is black.
\end{enumerate}
It is easy to see that both invariants hold when all the nodes
are gray. Let us see why they are maintained, first in the case
when Alice says that node~$x$ is white and then in the case when
Alice says that node~$x$ is black.

\begin{figure}\centering\leavevmode\vbox{
\begin{alg}
\^  $\proc{PropagateWhite}(y, T)$
\=  paint node $y$ in white
\=  \textit{PropagateWhite}(parent of $y$ in tree $T$)
\end{alg}
\vskip 3pt
\begin{alg}
\^  $\proc{PropagateBlack}(x,G)$
\=  paint node $x$ in black
\=  ~for~ ~each~ gray successor $y$ of node $x$ in $G$
\+    ~if~ all predecessors of node $y$ are black
\+      $\mathit{PropagateBlack}(y, G)$
\end{alg}
\vskip 3pt
\begin{alg}
\^  $\proc{RecoverColoring}(G)$
\=  paint in gray all the nodes of the flowgraph $G$
\=  ~while~ there are gray nodes
\+    $T:=\mathit{ComputeDominators}(\text{non-black subgraph of $G$})$
\=    $x:=\mathit{PickQueryNode}(G,T)$
\=    ~if~ $\mathit{AskAlice}(x)=\textit{white}$
\+      $\mathit{PropagateWhite}(x,T)$
\-    ~else~
\+      $\mathit{PropagateBlack}(x,G)$
\end{alg}}
\caption{Solution outline for Problem~\ref{pb:ra.abstract}}
\label{fig:alg_propagate}
\end{figure}

If Alice says that node~$x$ is white then there are no possible
colorings in which it is black, so it must be painted white
to maintain invariant~2. This is done on line~1 of the method
\textit{PropagateWhite} when it is first called from line~6
of the method \textit{RecoverColoring}. However, painting
node~$x$ in white might break invariant~2 with respect to
another gray node, because gray nodes must not dominate white
nodes (Corollary~\ref{co:ra.propagate}). This problem is
fixed by painting all the dominators of node~$x$ in white,
which is allowed by Proposition~\ref{prop:ra.fix-white}.
It follows from the definition of an immediate dominator
(Definition~\ref{def:ra.dominator}) that method
\textit{PropagateWhite} visits exactly the dominators of
node~$x$.

If Alice says that node~$x$ is black then, for similar reasons as
in the previous case, it must be painted black. This may break
the other part of invariant~2, namely, it may disconnect
gray nodes from the initial node in the non-black subgraph.
Should this happen, Proposition~\ref{prop:ra.fix-black} says
that all the disconnected gray nodes must be painted black.
This is the purpose of the method \textit{PropagateBlack}. Let
us say that a node is immediately disconnected when all its
parents are black. Initially no node is immediately disconnected
(invariant~2) but one may become so when its parent is painted
black (on line~1 of the method \textit{PropagateBlack}).
Whenever this happens it will be visited. Therefore, the
method \textit{PropagateBlack} visits and paints in black all
immediately blocked nodes. Now it is easy to see that if all
immediately blocked nodes are black, then there is no gray node
that is disconnected from the initial node in the non-black
subgraph. Pick some gray node. Because it is not black, it is
not immediately blocked, so it has a non-black predecessor. Then
repeat.

We have proved that the algorithm is correct. It is also fast.

\begin{proposition}
If calls $\mathit{ComputeDominators}()$, $\mathit{PickQueryNode}()$,
and $\mathit{AskAlice}()$ each takes constant time, then the 
execution of method \textit{RecoverColoring} takes $O(|V|+|E|)$~time.
\end{proposition}

\begin{proof}
The methods \textit{PropagateWhite} and \textit{PropagateBlack}
are called only for gray nodes and the first action paints the
given node in black or white. Therefore, they are called at most
once per node.

The method~\textit{PropagateBlack} examines the outgoing edges of
node~$x$, so it might examine all edges by the end. The condition
on line~3 of the method \textit{PropagateBlack} can be evaluated
in constant time if each node keeps a count of its non-black
parents.
\end{proof}

\subsection{Dominators Tree for Dags}
\label{sec:ra.dominators}

This section is a brief reminder of basic algorithms related to
dominators. It explains how the method \textit{ComputeDominators}
in Figure~\ref{fig:alg_propagate} is implemented.

The dominators tree~$T$ of a flowgraph~$G$ is a tree in which the
parent of node~$x$, denoted $\mathit{idom}(x)$, is the immediate
dominator of node~$x$ in the flowgraph~$G$. The root of the
dominators tree is the initial node of the flowgraph, which is
the only node without an immediate dominator. All other nodes~$y$
obey the equation
\begin{equation}
\mathit{idom}(y)=\mathit{LCA}(\{\,x\mid x\to y\,\}),
\end{equation}
where $\mathit{LCA}(S)$ is the root of the smallest dominators
subtree that contains the set~$S$ of nodes.

The method \textit{ComputeDominators} examines nodes in a
topological order of the flowgraph and inserts them as leaves in
the proper place in the dominators tree that it builds. When
node~$y$ is processed all its predecessors~$x$ have taken their
place in the dominators tree so \textit{LCA} can be computed.

\subsection{Choosing the Query}
\label{sec:ra.heuristic}

A greedy algorithm always asks the query whose answer provides
most information. If the probability of receiving the answer
\textit{black} is~$p$, then the information provided by the
answer is $-p\lg p - (1-p)\lg(1-p)$, which has a maximum at
$p=1/2$ and is symmetric around that point. In other words, the
greedy information theoretic approach says that Bob should always
inquire about the node whose probability of being black is as
close as possible to~$1/2$.

\paragraph{An Example}

To estimate the probability of various nodes being black we
need a model. Let us go back to the example flowgraph in
Figure~\ref{fig:fg.chain} (on page~\pageref{fig:fg.chain}).
The simplest possible model is the one in which each node has
an independent probability~$p$ of having a \emph{bug} and
probability $q=1-p$ of not having a bug. A node with a bug
is black, and so are all that follow it. In this model, the
five possible colorings of Figure~\ref{fig:fg.chain} have
probabilities (a)~$q^4$, (b)~$q^3p$, (c)~$q^2p$, (d)~$qp$, and
(e)~$p$. The probability that node~$x$ is black is the sum of
probabilities of possible colorings in which node~$x$ is black.
For the nodes in Figure~\ref{fig:fg.chain} we have (1)~$1-q$,
(2)~$1-q^2$, (3)~$1-q^3$, and (4)~$1-q^4$.

In general, for a path flowgraph with nodes $1\to2\to\cdots\to
n$, the probability that node~$k$ is black is
\begin{equation}
p_k = 1 - (1 - p)^k
\end{equation}
where $p$ is the probability that a node has a bug. In the common
case, we expect the analysis to be run on code that is mostly
correct. Indeed, bugs are usually clustered in the region of the
code that is actively developed, while the program verifier is
run on the whole code base. Hence, in a first approximation, it
is reasonable to expect $p$ to be very small. Then $p_k\approx
kp$ and, if $p$~is really small, the probability that is closest
to~$1/2$ is~$p_n$. In other words, if we are fairly confident
that the code is correct, then we should query the last node in
the path. If the code is indeed correct, then we are done after
one query.

However, if the first query returns \textit{black} then there is
at least one bug. The expected number of bugs for the model we
use is $b=pn$. (The number of bugs has the probability generating
function
\begin{equation}
B(z)=\sum_k \binom{n}{k} (1-p)^{n-k} p^k z^k = (1-p+pz)^n,
\end{equation}
so the average is $B'(1)=pn$.) Therefore, if we expect the
program to have $b$~bugs, then it makes sense to choose
\begin{equation}
p=b/n \label{eq:ra.bug-estim}
\end{equation}
in our model. Solving $p_k=1/2$ we obtain
\begin{equation}
k=-1/\lg(1-p)\approx(\ln2)/p.
\end{equation}
Plugging in the estimate~\eqref{eq:ra.bug-estim} we finally
arrive at
\begin{equation}
k\approx 0.7\cdot n/b
\end{equation}
In other words, the information theoretic greedy strategy is
a binary search that splits the interval~$i.\,.\>j$ not at
$i+0.5(j-i)$, but at $i+(0.7/b)(j-i)$, where $b$ is the expected
number of bugs in the interval~$i.\,.\>j$.

\paragraph{The General Case}

Most observations for path flowgraphs generalize easily. We will
have two modes of work, one in which we expect no bug and one in
which we expect one bug. When we expect no bug we will pick a
gray node that is far from the initial node. When we expect one
bug we will perform a binary search on a path of nodes in the
dominator tree. The key here is that color propagation rules on
paths in the dominator tree are exactly the same as the color
propagation rules for a path flowgraph.

Figure~\ref{fig:ra.alg} summarizes all the observations made
so far. Line~1 splits each node~$x$ of the flowgraph into a
node~$x_i$ carrying the precondition~$a_x$ and a node~$x_o$
carrying the postcondition~$b_x$. Later these predicates are
accessed using the method \textit{Formula}(). Line~2 colors the
initial node white to establish the invariant of the loop on
line~10.

\begin{figure}\centering\leavevmode\vbox{
\begin{alg}
\^  $\proc{ReachabilityAnalysis}(G)$
\=  construct $H$ by splitting nodes of $G$
\=  paint the initial node of $H$ white and the others gray
\=  $T:=\mathit{ComputeDominators}(H)$
\=  ~while~ $H$ has gray nodes
\+    $x:=\text{a leaf from the deapest level of $T$}$
\=    ~if~ $\mathit{Valid}(\mathit{Not}(\mathit{Formula}(x)))$
\+      $\mathit{PropagateBlack}(x,H)$
\=      let $[y_1,\ldots,y_n]$ be the path in $T$ from $y_1=\mathit{root}(T)$ to $y_n=x$
\=      $i:=1,\quad j:=n$
\=      ~while~ $i+1<j$ \comment $y_i$ is white, $y_j$ is black
\+        $k:=\mathit{round}(i+0.7(j-i))$
\=        ~if~ $\mathit{Valid}(\mathit{Not}(\mathit{Formula}(y_k)))$
\+          $\mathit{PropagateBlack}(y_k, H)$
\=          $j:=k$
\-        ~else~
\+          $\mathit{PropagateWhite}(y_k, T)$
\=          $i:=k$
\2      $T:=\mathit{ComputeDominators}(\text{non-black subgraph of $H$})$
\-    ~else~
\+      $\mathit{PropagateWhite}(x,T)$
\end{alg}
}
\caption{Algorithm for semantic reachability analysis}
\label{fig:ra.alg}
\end{figure}

\subsection{Performance}

If the verified flowgraph has no bugs (no semantically
unreachable statements and no doomed code) then $l$~queries are
necessary and sufficient, where $l$~is the number of leaves in
the flowgraph. The algorithm reaches this bound. In other words,
if there are no bugs, then the prover is queried once for each
\textbf{return} statement in the program. For each bug, there are
$\sim\lg h$ extra queries, where $h$~is the (average) length of a
path in the flowgraph from the initial node to a \textbf{return}
node. In total, there are roughly $l+b\lg h$ queries for a
program with $b$~bugs.

\paragraph{Experiments}

The frontend of \escjava (JavaFE) contains $1890$~methods and is
one of the largest coherent JML-annotated code base. Verifying
it takes $31589$~seconds (almost $9$~hours), out of which
$34.8\%$~is spent in semantic reachability analysis, out of
which $99.8\%$~is spent in the prover. The total number of leaves
in flowgraphs is~$3256$ and the total number of prover queries
is~$3351$. In other words, on average per method
\begin{itemize}
\item 
  semantic reachability analysis takes $5.82$~seconds,
\item 
  out of which $5.81$~seconds are spent in $1.77$ prover queries
  (not much more than the number of \textbf{return} statements, which
  is~$1.72$)
\item
  and the other $0.01$~seconds are spent computing prover
  queries, deciding which queries to perform, and propagating
  semantic reachability information in the flowgraph.
\end{itemize}
(This benchmark was run using the default treatment of loops in
\escjava---unrolling once.)

\section{Case Study}
\label{sec:ra.case_study}

The frontend of \escjava (known as JavaFE) consists of
$217$~annotated classes with $1890$~methods. Semantic
reachability analysis revealed $\approx50$~issues. Short
descriptions of these problems follow, each description preceded
by the number of cases to which it applies.

\subsection{Dead Code}

\begin{itemize}
\item[1] 
  A \textbf{catch} block was unreachable because it was catching
  a \textit{RuntimeException} and the specifications of methods
  called in the \textbf{try} block did not signal that exception type.
\item[9]
  Informal comments indicated that the dead code is intended.
  JML has an equivalent annotation, which \escjava understands
  and should be used instead of the informal comments.
\item[1]
  The code was semantically unreachable in the classic sense (no
  annotation involved.)
\end{itemize}

\subsection{Doomed Code}

\begin{itemize}
\item[6] A few assertions were bound to fail on any input.
\item[9]
  The scenario illustrated by methods \textit{m2}~and~\textit{m3}
  in Figure~\ref{fig:ra-doomed-code-ex} on
  page~\pageref{fig:ra-doomed-code-ex}. The method call
  was bound to fail.
\end{itemize}

\subsection{Inconsistent Specifications}

\begin{itemize}
\item[5] 
  Inconsistencies in the JDK specifications that ship with
  \escjava. They were filed as \escjava bugs \#595, \#550, \#568,
  \#549, and \#545. These are bugs whose cause was hard to track.
  Before, they have escaped for years the \escjava's test-suite 
  that was designed to catch them.
\end{itemize}

\subsection{Unsoundness and Bugs}

\begin{itemize}
\item[1]
  The JDK specification used a JML \emph{informal comment}.
  According to the semantics of JML, an informal comment should
  not cause warning messages. Roughly, this means that, depending
  on the context, an informal comment should be treated sometimes
  as \textbf{true} and sometimes as \textbf{false}. However,
  \escjava always treats is as \textbf{true}. See \escjava bug
  \#547 for details.
\item[4]
  Loops with a constant bound that is bigger than the loop 
  unrolling limit.
\item[2]
  If the \textbf{modifies} clause is missing on a method~$m$,
  \escjava considers that method~$m$ is allowed to modify any
  global variable while checking it, but assumes that method~$m$
  does not modify any global state while checking methods that
  call it. This is obviously unsound (and motivated by the desire
  to not flood new users with irrelevant warnings).
\end{itemize}

\subsection{Others}

\begin{itemize}
\item[12] 
  The cause of these warnings provided by the semantic
  reachability analysis is unclear. Given the experience of
  tracking down the cause of some of the other warnings, it is
  likely that a complex interaction between annotations located
  in different files is involved.
\end{itemize}

\section{Conclusions and Related Work}
\label{sec:ra.conclusions}

This chapter presents the theoretical underpinnings of semantic
reachability analysis for annotated code and an efficient
algorithm for finding 
\begin{enumerate}
\item dead code,
\item doomed code,
\item inconsistent specifications, and
\item unsoundness bugs.
\end{enumerate}


Abstract interpretation~\cite{cousot1977} and symbolic
execution~\cite{king1976} are program verification techniques
that can easily detect semantically unreachable code as a
by-product, although they seldom seem to be used for this
purpose. To ensure termination, these techniques usually
over-approximate the set of possible states, which means that
they might conclude that some code is semantically reachable,
although it is not.

Lermer et al.~\cite{lermer2003} give semantics for execution
paths using predicate transformers. In particular, they define
\emph{dead paths} and \emph{dead commands}. A dead path is one
that is taken by no execution; a dead command is a program whose
VC computed by the weakest precondition method is unsatisfiable.


Hoenicke et al.~\cite{hoenicke2009} explain why it is desirable
to detect doomed code. Their implementation uses the weakest
precondition method to compute prover queries associated with
each node. The irrelevant portions of the flowgraph are turned
off using auxiliary variables rather than explicitly manipulating
the flowgraph (see Section~\ref{sec:ra.theory}).


Chalin~\cite{chalin2006} explains why it is desirable to have
automatic sanity checks for specifications and designs an
automated analysis that finds bugs in JML annotations, other than
plain logic inconsistencies.

Cok and Kiniry~\cite{cok2005methods,cok2004} explain how
\escjava handles method calls that appear in annotations. Darvas
and M\"uller~\cite{darvas2006methods} point out that that
treatment is unsound because some annotations are translated
by \escjava's frontend into inconsistent assumptions. Their
solution is a set of syntactic constraints on the use of
method calls in annotations. Rudich et al.~\cite{rudich2008}
relax these constraints while maintaining soundness. Leino and
Middelkoop~\cite{leino2009methods} replace the syntactical
constraints with queries to the theorem prover. Semantic
reachability analysis will signal most problems caused by the
use of method calls in specification, because they manifest as
inconsistencies at the level of the intermediate representation
(Boogie).  However, being a generic analysis, the error message is
usually not very informative. Another shortcoming of the semantic
reachability analysis compared to the other specialized solutions
is that it does not ensure that specifications are well-founded.


The naive version of semantic reachability analysis is used in
the static verifier Boogie from Microsoft
Research under the name ``smoke testing'' and is used mostly for
catching bugs in front-ends for Boogie.


\escjava~\cite{flanagan2002escjava} is a program verifier for
Java annotated with JML~\cite{leavens2006jml,cok2004}. By design,
\escjava favors user friendliness over soundness.

The method \textit{ComputeDominators} uses the algorithm of
Cooper~\cite{cooper2000}, which is easy to implement and
works fast in practice~\cite{georgiadis2006}. (There are
linear time algorithms, such as the one of Buchsbaum et
al.~\cite{buchsbaum2008}.) Probability generating functions,
such as the one used in Section~\ref{sec:ra.heuristic}, are
treated by most probability textbooks, such as Graham et
al.~\cite[Chapter~8]{graham1998}.

\chapter{Conclusions}
\label{ch:conclusions}

\chquote{I love to travel but I hate to arrive.}{Albert Einstein}

\noindent The design of FreeBoogie (Chapter~\ref{ch:design}) does
not surprise much those who write \emph{program verification}
tools, although it has a few distinctive characteristics.
FreeBoogie spends most of its time transforming step-by-step
a Boogie program into simpler and simpler Boogie programs.
The steps are always as small as possible. For example, even
if the \textbf{call} statements could be
desugared into \textbf{assume} and \textbf{assert} statements
in one step, FreeBoogie does it in two steps. After
each transformation, FreeBoogie type-checks the intermediate
program to rule out egregious bugs. \emph{Correctness has a
very high priority in deciding the design}. The main data
structures of FreeBoogie are immutable. A nice side-effect of
immutability shows that correctness and efficiency are not
always competing goals. Because pieces of Boogie programs are
represented by immutable data structures, they can act as keys
in caches, which makes it easy to avoid re-doing the same work
(Chapter~\ref{ch:ev}).

Operational semantics, Hoare logic, and predicate transformers
are ways to define \emph{programming languages}. This dissertation
uses all three for a subset of Boogie and lingers on the relations
between the three (Chapter~\ref{ch:spwp}). Being able to quickly
switch between different points of view on the semantics of Boogie
programs is essential in understanding the techniques presented
later (Chapters~\ref{ch:ev} and~\ref{ch:reachability}).

The \emph{algorithmic} problems solved in this dissertation
(Chapters~\ref{ch:passive} and~\ref{ch:reachability})
are not particularly difficult. The most difficult part
is the proof of Theorem~\ref{th:coiv-is-np-complete} in
Section~\ref{sec:passive.iv}. Program verifiers try to solve
instances of an undecidable problem, which makes for a good
excuse to not analyze their algorithms. But, there is not much
undecidable going on in many parts of a program verifier. In
particular, a VC generator is a place full of little algorithms
that deserve to be better understood.

\emph{The goal of this dissertation is to give an account of some
interesting but often overlooked parts of a VC generator. The
meta-goal of this dissertation is to find connections between
\emph{program verification}, \emph{programming languages}, and
\emph{algorithms}.}

It also happens that the dissertation could serve as a suitable
guide for a developer who wants to contribute to the VC generator
FreeBoogie. There are $32\,000$~lines of code that handle the
tasks described in this dissertation. (About $23\,000$~of them
are generated by ANTLR, CLOPS, and AstGen.) This infrastructure
could serve for future developments such as invariant inference
(perhaps via abstract interpretation), termination detection
(which would probably entail enriching the Boogie language),
and other experiments with the Boogie language (such as adding
operators from separation logic).

\bigskip
Specific contributions include:
  \begin{itemize}
  \item 
    the design of the VC generator, including discussions
    of the main design decisions;
  \item 
    a precise definition of passivation and a study of its
    algorithmic properties;
  \item
    a comparison between the weakest precondition and the
    strongest postcondition methods of generating VCs;
  \item
    various semantics for a subset of the Boogie 
    language---operational semantics, Hoare logic,
    predicate transformers---and a discussion of the
    relations between them;
  \item
    an algorithm for unsharing expressions, a problem that
    is usually tangled with computing the weakest precondition
    efficiently;
  \item
    a proof technique for the correctness of algorithms
    that simplify a VC based on a known old  VC;
  \item
    a heuristic for detecting common parts of two expression
    trees, such as two VCs;
  \item
    the semantic reachability analysis, in the context of
    program verifiers;
  \item
    an efficient heuristic for finding dead code, doomed code,
    inconsistent specifications, and soundness bugs.
  \end{itemize}

\appendix
\chapter{Notation}
\label{ch:notation}

\chquote{The best notation is no notation; whenever it is possible
to avoid the use of a complicated alphabetic apparatus, avoid it.
A good attitude to the preparation of written mathematical exposition
is to pretend that it is spoken. Pretend that you are explaining
the subject to a friend on a long walk in the woods, with no paper
available; fall back to symbolism only when it is really necessary.}
{Paul Halmos~\cite{steenrod1973}}

{
  \baselineskip=12pt\raggedright
  \def\row#1#2{
    \setbox0=\hbox to 2.5cm{\hfil\strut#1\hfil}
    \setbox1=\vbox{\hsize=8.5cm\noindent\strut#2\strut}
    \dimen0=\ht0\advance\dimen0 by\dp0
    \dimen1=\ht1\advance\dimen1 by\dp1
    \ifdim\dimen0<\dimen1\dimen0=\dimen1\fi
    \hbox to\hsize{%
      \hfil%
      \vbox to\dimen0{\vfil\box0\vfil}%
      \hskip1em plus1em%
      \vbox to\dimen0{\vfil\box1\vfil}%
      \hfil}\vskip5pt plus1pt\relax
  }

  \hrule\row{\textbf{notation}}{\textbf{description}}
  \row{$A$, $B$, \dots}{sets; objects with an internal structure}
  \row{$\Z$}{the set of integers $-1$,~$0$,~$1$, \dots}
  \row{$\Z_+$}{the set of nonnegative integers $0$,~$1$, \dots}
  \row{$\R$}{the set of real numbers}
  \row{$\R_+$}{the set of nonnegative real numbers}
  \row{$a\in A$ \emph{or} $a:A$}{$a$~is an element of the set~$A$}
  \row{$A\times B$}{the cartesian product of the sets~$A$ and~$B$}
  \row{$A\to B$ \emph{or} $B^A$}{the set of functions from $A$ to $B$; $\to$ is
    right associative}
  \row{$f\>x$ \emph{or} $f(x)$}{function application; 
    space is left associative and binds tighter than all else}
  \row{$w(P)$}{$\{\,w(x)\mid x\in P\,\}$}
  \row{$f \circ g$}{function composition: $(f\circ g)\>x=f\>(g\>x)$
    for all~$x$ in the domain of~$g$}
  \row{$\top$}{true}
  \row{$\fls$}{false}
  \row{$\mathbb{B}$}{the set $\{\tru,\fls\}$ of booleans}
  \row{$e$, $f$}{expressions, possibly predicates}
  \row{$p$, $q$, $r$}{predicates}
  \row{$a$}{preconditions}
  \row{$b$}{postconditions}
  \row{$u$, $v$, $w$}{variables}
  \row{$x$, $y$, $z$}{statements, nodes in a graph}
  \row{$m$, $n$}{integer constants}
  \row{$i$, $j$, $k$}{integer variables, indexes}
  \row{\textit{wp}}{the weakest precondition predicate transformer}
  \row{\textit{sp}}{the strongest postcondition predicate transformer}
  \row{$x\to y$}{directed edge from node~$x$ to node~$y$}
  \row{$x\stackrel{P}{\leadsto}y$}
    {denotes a path~$P$ from node~$x$ to node~$y$; the path~$P$ is sometimes seen
    as a set of nodes that includes the~endpoints}
  \row{\tikz[baseline=-.5ex] \enode () at (0,0){};}{graph node}
  \row{\tikz[baseline=-.5ex] \fnode () at (0,0){};}
    {graph node of interest; read-write flowgraph node}
  \row{\tikz[baseline=-.5ex] \ronode () at (0,0){};}
    {read-only flowgraph node}
  \row{\tikz[baseline=-.5ex] \wonode () at (0,0){};}
    {write-only flowgraph node}
  \row{
    \begin{tikzpicture}[scale=.5,baseline=-.5ex]
    \enode (A) at (0,0) [label=left:$x$] {};
    \enode (B) at (1,0) [label=right:$y$] {};
    \draw[densely dotted] (A) -- (B);
    \end{tikzpicture}}
    {read node with incoming edges from the nodes~$x$ and~$y$}
  \row{$G$, $H$}{graphs, flowgraphs, programs}
  \row{$V$}{the set of nodes of a graph}
  \row{$E$}{the set of edges of a graph}
  \row{$P$, $Q$}{paths, sets of nodes}
  \row{$[p]$}{$1$ if $p$ evaluates to $\tru$, $0$ otherwise}
  \row{$|p|$}{$\tru$ if the predicate~$p$ is valid, $\fls$ otherwise}
}

\backmatter
\raggedright
\bibliographystyle{plain}
\bibliography{phd}
\printindex

\end{document}